\documentclass[a4paper,11pt]{article}
\pdfoutput=1

\usepackage{tcs}
\begin{document}

\title{A Dobrushin condition for quantum Markov chains:

\Large  Rapid mixing and conditional mutual information at high temperature}

\author{
Ainesh Bakshi \\
\texttt{ainesh@mit.edu} \\
NYU
\and
Allen Liu \\
\texttt{cliu568@mit.edu} \\
UC Berkeley
\and
Ankur Moitra \\
\texttt{moitra@mit.edu} \\
MIT
\and
Ewin Tang \\
\texttt{ewin@berkeley.edu} \\
UC Berkeley
}
\date{}

\maketitle

\begin{abstract}

A central challenge in quantum physics is to understand the structural properties of many-body systems, both in equilibrium and out of equilibrium.
For classical systems, we have a unified perspective which connects structural properties of systems at thermal equilibrium to the Markov chain dynamics that mix to them.
We lack such a perspective for quantum systems: 
there is no framework to translate the quantitative convergence of the Markovian evolution into strong structural consequences.

We develop a general framework that brings the breadth and flexibility of the classical theory to quantum Gibbs states at high temperature.
At its core is a natural quantum analog of a Dobrushin condition; whenever this condition holds, a concise path-coupling argument proves rapid mixing for the corresponding Markovian evolution.
The same machinery bridges dynamic and structural properties: rapid mixing yields exponential decay of conditional mutual information (CMI) without restrictions on the size of the probed subsystems, resolving a central question in the theory of open quantum systems.
Our key technical insight is an optimal transport viewpoint which couples quantum dynamics to a linear differential equation, enabling precise control over how local deviations from equilibrium propagate to distant sites.

\end{abstract}

\thispagestyle{empty}
\clearpage
\newpage

\microtypesetup{protrusion=false}
\tableofcontents{}
\thispagestyle{empty}
\microtypesetup{protrusion=true}
\clearpage
\setcounter{page}{1}

\section{Introduction}

Thermal states of local Hamiltonians are the cornerstone of quantum simulation and quantum many-body physics.
They capture the finite-temperature behavior of well-studied models in chemistry and material science~\cite{Georges1996DMFT, Sandvik2010Vietri, Schollwock11}, underpin algorithms for computing reaction rates and phase diagrams~\cite{CraigManolopoulos2004RPMDcorr, Habershon2013RPMDreview, Selisko2022qVQT}, and are often viewed as an avenue to achieving quantum advantage~\cite{lloyd96, tovpv11, Pre18} \ewin{I would maybe use different citations here for the chemistry and material science stuff.. TODO}.
Yet, preparing these states on quantum devices and understanding their structural properties, such as locality of correlations, information flow, and mixing under natural dynamics, remains notoriously challenging.

Strikingly, quantum systems lack the flexible and powerful theory developed for classical spin systems over three decades ago. In the classical setting, a well-established toolkit, consisting of Dobrushin conditions, path coupling, and functional inequalities~\cite{dobrushin1970prescribing, dobrushin1987completely, bd97}, yields equivalences and implications between \emph{dynamical properties} of Gibbs measures, such as rapid mixing and \emph{structural properties} such as strong spatial mixing, and exponential decay of correlations~\cite{stroock1992equivalence, weitz2005combinatorial}.
These methods offer simple ways to analyze a model of interest and, by tying dynamical behavior to structural behavior, show that a phase transition cannot change one without the other.\ewin{good?}
In the quantum setting, non-commutativity blunts this established toolkit.
Here, the existing techniques are fragmented and largely incomparable to their classical counterparts. 

To understand what can be shown for quantum systems, the community has primarily converged around studying the dynamical property of \emph{rapid mixing} of natural Markovian evolutions and the structural property of \emph{decay of conditional mutual information (CMI)} for Gibbs states.\footnote{
    This is on the path towards developing the quantum version of the aforementioned classical equivalence, since both are weakened versions of the respective classical notions of temporal and spatial mixing used there.
}
However, existing results on both ends are unsatisfying.

On the dynamical side, the goal is to prove \emph{rapid mixing}: show that a natural, efficiently-implementable choice of Markovian evolution takes an input state to a Gibbs state in time logarithmic in system size, the fastest possible rate of convergence.
Existing results on mixing form a patchwork of sufficient conditions: many attempts, particularly those generalizing classical conditions like log-Sobolev inequalities, apply only to commuting Hamiltonians~\cite{kb16,crf20,bcglpr23,cgkr24,kacr25} or one-dimensional Hamiltonians~\cite{lucia2025spectral}.\footnote{
    There is a closely-related body of work deriving Gibbs state preparation algorithms from structural properties of Gibbs states~\cite{bk18,kb19,kaa21,cr25}.
    These algorithms are not direct Markovian evolution, so they have less of a physical interpretation, but compared to the aforementioned results, they work in greater generality assuming more natural structural properties.
}
The situation has been improved with two recent works of Rouzé, França, and Alhambra~\cite{rfa24,rfa24a}, with the latter one proving rapid mixing via contraction in a norm inspired by the classical literature.
Even with these works, though, their tools are isolated, and do not provide the same level of clarity and cohesion as the classical theory.

On the structural side, focus has coalesced around analyzing the rate of decay of CMI.
It has been conjectured that, for Gibbs states, the amount of correlation between two regions after conditioning on the rest of the system decays exponentially in the distance between the regions, even when the size of said regions scales with system size; this property is referred to as clustering of CMI or approximate (global) Markovianity~\cite{kuwahara24}.
Markovianity is a foundational property of classical Gibbs distributions, and is vital to essentially all results about them~\cite{HC71}.
Yet, despite substantial effort~\cite{brown2012quantum,kb16, kkb20, bluhm2022exponential, kuwahara24, kk25, cr25}, a rigorous proof of the exponential decay of CMI in the quantum setting has remained elusive.

These dynamical and structural properties have important applications beyond the theory of thermal many-body systems.
Preparing Gibbs states through mixing of Markovian evolution is currently being considered as a strong candidate for practical and useful quantum advantage, since Gibbs states are the target of many simulation questions, and Markovian dynamics are believed to mimic how nature prepares its systems~\cite{tovpv11,ckbg23,ckg23}.
Simultaneously, in the broader context of understanding quantum phases, the past two decades of research has shown that bi-partite correlation measures alone are insufficient, evidenced by the fact that topologically ordered phases exhibit genuinely multi-partite correlations~\cite{chen2010local,wen2017colloquium}.  This realization has driven the search for information-theoretic measures that capture multi-partite correlations. Conditional mutual information (CMI) has emerged as the leading candidate for quantifying tri-partite correlations~\cite{christandl2004squashed, bsw15,sutter2018approximate}. It underpins the definition of topological entanglement entropy~\cite{kitaev2006topological, levin2006detecting, kato2016information} and characterizes information scrambling~\cite{ding2016conditional,iyoda2018scrambling}. It has also found applications in identifying measurement-based quantum phase transitions~\cite{gullans2020dynamical,ippoliti2021entanglement, koh2023measurement} and studying entanglement in conformal field theories~\cite{hayden2013holographic, rota2016tripartite, maric2023universality}. 

The aforementioned applications are limited by a lack of techniques to analyze mixing and decay of CMI, even for simple Gibbs state models.
Perhaps unsurprisingly, the interplay between these two properties of Gibbs states also remains poorly understood, largely due to the ad-hoc nature of establishing mixing time bounds.
Unlike in the classical theory, there is still no unifying framework that translates the quantitative convergence of local quantum dynamics into strong structural consequences.

\subsection{Results}
In this work, we present a general quantum framework that brings the clarity and versatility of the classical theory to high-temperature Gibbs states of local Hamiltonians.
We introduce a natural quantum generalization of the \emph{Dobrushin condition}, first studied by Dobrushin in 1970~\cite{dobrushin1970prescribing}.
We show that, when this condition is satisfied, rapid mixing follows via a concise quantum \emph{path coupling} argument, paralleling the classical analysis of Bubley and Dyer~\cite{bd97} and refined by Dyer, Goldberg and Jerrum~\cite{dgj09}. 
Our framework makes explicit the bridge between dynamics and structure: the same analysis that establishes rapid mixing also yields the first rigorous proof of exponential decay of conditional mutual information for high-temperature quantum Gibbs states, thereby resolving a central open question in the theory of open quantum systems.

Our framework begins with the \emph{Dobrushin influence matrix}.
Dobrushin~\cite{dobrushin1970prescribing} defines the $(i,j)$-th entry of the influence matrix to measure how much changing spin $j$ can alter the behavior of Glauber dynamics, a natural local Markov chain, when updating site $i$.
Specifically, for any two configurations that only differ at site $j$, this entry bounds the total variation distance between distributions induced by applying the one-site resampling law at site $i$.
We lift this definition to trace-preserving linear maps applied to density matrices, with the small tweak that we track a quantum analog of Wasserstein distance~\cite{dmtl21} instead of trace distance (see \cref{def:quantum-dobrushin} for a formal definition).
We note that trace distance would suffice for our purposes if the underlying dynamics were truly local, as opposed to quasi-local.   

Our \emph{quantum Dobrushin condition} is then a maximum column-norm bound on the aforementioned influence matrix, matching Dobrushin's condition for establishing uniqueness of the Gibbs measure~\cite{dobrushin1970prescribing}. In 2005, Weitz showed that the Dobrushin condition is sufficient to obtain rapid mixing for Glauber dynamics\ewin{can we say Glauber here? the reader won't know what block dynamics is}, via a path coupling argument~\cite{weitz2005combinatorial}. We provide a quantum analog of Weitz's result, obtaining a simple proof that any map over density matrices that satisfies the quantum Dobrushin condition admits a unique fixed point, and any initial state rapidly converges to the fixed point. We believe this framework is of independent interest and can help obtain mixing time bounds beyond currently known settings.

We then introduce a \emph{(detailed) balanced Lindbladian} and show that the fixed point of this evolution is the Gibbs state. Our Lindbladian is a close cousin of the one introduced by Chen, Kastoryano and Gilyén~\cite{ckg23}. For high-temperature Gibbs states, it enjoys several nice properties, including Kubo-Martin-Schwinger detailed balance, quasi-locality, and efficient implementation on a quantum computer. The jump operators are the imaginary time evolution of $1$-local Pauli operators (see \cref{def:our-l}). We show that our Lindbladian satisfies the quantum Dobrushin condition, and thus, we obtain rapid mixing at high temperature:

\begin{theorem}[Rapid mixing of Lindbladian evolution (informal)]
\label{thm:rapid-mixing-informal}
Given a geometrically local Hamiltonian, let $\sigma$ be the corresponding Gibbs state at inverse temperature $\beta$. For any $\beta < \beta_c$, where $\beta_c$ is a fixed constant, and any initial state $\rho$, the Lindbladian evolution runs for $\bigO{\log(n/\eps)}$ time and outputs a state $\hat{\rho}$ such that $\trnorm{ \hat{\rho} - \sigma} \leq \eps$. 
\end{theorem}

Obtaining such mixing time bounds for Lindbladian evolution had been a major open question in the quantum Markov chain literature. The recent result~\cite{rfa24a} attains a qualitatively similar result, using a different set of conditions and analysis (see~\cref{sec:formalresults-relatedwork} for a comparison).

As mentioned in the introduction, in the classical setting, it is well-known that strong temporal mixing (rapid mixing of a local Markov chain such as Glauber dynamics) implies strong spatial mixing~\cite{weitz04}.
Strong spatial mixing, a structural property of a Gibbs measure, roughly states that the effect of conditioning on the state of a subsystem decays exponentially in the distance to that subsystem, and that moreover, this decay property also holds uniformly for all conditional measures.
\ewin{old:"conditioning on sites only affects the Gibbs measure around those sites, with correlations dying off exponentially in the distance to the sites, uniformly across the system, even with boundary conditions." Did I get this right? Wanted to make it sound a bit more CMI-like. Old version:
Strong spatial mixing roughly states that tiny perturbations to the boundary spins on a small part of the boundary die off exponentially in the distance to the perturbation, uniformly over the system size and location.} \ainesh{this seems weaker than weak spatial mixing}
Strong spatial mixing is the most stringent structural property of Gibbs measures and implies other structural properties such as uniqueness and exponential decay of correlations. In fact, Weitz showed that strong spatial mixing and strong temporal mixing are equivalent, demonstrating that the right tools for establishing rapid mixing can also be used to prove structural properties of Gibbs measures.

We demonstrate the versatility of our quantum framework by using the rapid mixing of Lindbladian dynamics to establish exponential decay of CMI, resolving a central problem in the theory of open quantum systems:

\begin{theorem}[Exponential decay of CMI (informal)]
Let $\sigma$ be the Gibbs state of a local Hamiltonian on a lattice at inverse temperature $\beta$.
Let $A, B, C$ be any tripartition of $[n]$.
For any $\beta < \beta_c$, where $\beta_c$ is a fixed constant, 
\ewin{define $I_\sigma$}
\begin{equation*}
    I_{\sigma}( A : C \mid B ) \leq \bigO{ \abs{A} \cdot \abs{C} }  \cdot \exp\Paren{- \dist(A,C) /\zeta } \,,
\end{equation*}
for a fixed universal constant $\zeta$.
\end{theorem}

A recent result of Kuwahara~\cite{kuwahara24} establishes exponential decay of CMI when the subsystems size of $A$ and $C$ is at most a constant, which Chen and Rouzé~\cite{cr25} improved to exponential decay when the minimum size is at most a constant.
These results hold at any constant temperature.
Our CMI bound holds only at high temperature, but is unencumbered by such restrictions on the size of $A$ and $C$.
Another recent result by Kuwahara and Kato applies to growing $A$ and $C$, but the rate of decay remains sub-exponential~\cite{kk25}.   
Our bound is the first to scale linearly in $\abs{A}$, $\abs{C}$, and decay exponentially in $\dist(A, C)$; these are the conjectured-optimal scaling for all three of these parameters.\footnote{
    The absolute strongest version of this conjecture~\cite[Theorem 1]{kkb20} makes two improvements beyond these scalings: first, it depends on the size of the boundaries of the regions, $I_\sigma(A : C \mid B) \leq \bigO{\min(\abs{\partial A}, \abs{\partial C})\exp(-\dist(A, C) / \zeta)}$; second, it holds even when $A, B, C$ do not cover the entire space.
    With these improvements, this bound becomes one of the most stringent structural properties for Gibbs states, implying several other structural properties, such as clustering of correlations and strong thermal area laws~\cite{kkb20}.
    We are not aware of any results which attain these improvements.
}

Our approach relies on a well-known result~\cite{fr15,bsw15} that bounds the CMI relative to $\sigma$ by the accuracy of a \emph{recovery map}, a map with the goal of, given only the $B$ subsystem of $\sigma_{AB}$, to reconstruct the $C$ subsystem to recover $\sigma$ in full.
We use the semigroup generated by our balanced Lindblandian as the recovery map and show that the subsystem $C$ rapidly converges to its thermal state in Wasserstein norm, with any residual deviations decaying exponentially in the distance from $C$. The key technical challenge is how to track the effect of updating a site in subsystem $C$ on sites in $A$. Since the Lindbladian is quasi-local, there is a small chance that a site in $A$ is updated, and apriori, this could significantly perturb the configurations in $A$. 

To address this challenge, we build on the optimal transport interpretation of the quantum Wasserstein norm: for any two states $\sigma$ and $\rho$, let $X = \sigma - \rho$ represent the amount of signed ``mass'' that needs to be moved to transform $\sigma$ to $\rho$.
Then, a \emph{transport plan} is a decomposition $X = \sum_{i \in [n]} X_i$ such that each $X_i$ is Hermitian and $\tr_i(X_i) = 0$ (so $X_i$ only transports ``mass'' between configurations that differ at site $i$).
A transport plan has an associated \emph{cost vector}, a length $n$ vector that tracks $\trnorm{X_i}$ for each $i \in [n]$, and can be interpreted as how much ``mass'' is being routed through each site. 
The Wasserstein norm is the minimum $\ell_1$ norm of a cost vector, minimizing over all transport plans of $X$.

Our key insight is to directly analyze the transport plans underlying the Wasserstein norm.
When we initialize the $C$ subsystem, the resulting state $\sigma_{AB} \otimes \rho_C$ only deviates from the Gibbs state on that subsystem; so, the corresponding transport plan between the two is supported only on the $C$ subsystem.
Applying our Lindbladian evolution changes these transport plans and their cost vectors.
With our key insight, we show that it induces a classical linear differential equation acting on cost vectors.
This differential equation essentially states that updating any site in $C$ can only increase the cost vector at a site $j$ by at most $\exp(-\dist(C,j))$. Tracking cost vectors in this way supplies precisely the missing control needed for the CMI bound.

\section{Technical overview}

In this section, we define and motivate our key mathematical concepts; then, we describe how to apply them to get the desired statements; finally, we formally state our results.
We also provide a modern treatment of the classical literature, since this is the theory that we generalize~\cite{lp17,liu2023lecturenotes}.

\subsection{Key ideas}

For the sake of exposition, in this section, we begin with some natural simplifications: the evolutions will be local instead of quasi-local, and we will be working with the discrete setting instead of continuous.
Handling these changes takes up the brunt of this work; without them, our ideas are conceptually and technically clean.

Recall that our goal is to prove \emph{rapid mixing} of a quantum channel.
Let $\Phi$ be a quantum channel; we want to show that, for a fixed point $\sigma$ of $\Phi$,
\begin{equation}
    \label{eq:tech-mixing}
    \trnorm{\Phi^t(\rho) - \sigma} \leq \eps
\end{equation}
for some $t = \bigO{\log(n/\eps)}$. From the classical literature, we know that bounding mixing time is tightly related to showing that $\Phi$ is a contraction in an appropriate distance measure.
Such a contraction statement takes roughly the following form: for some notion of distance $D$, for any state $\rho$ and the fixed point $\sigma$, $D(\Phi(\rho) \| \sigma)  < (1-\gamma) D(\rho \| \sigma)$, where $\gamma$ is the amount of contraction. 
For example, the two standard functional-inequality routes to mixing time bounds are through (i) a Poincaré inequality (equivalently, a lower bound on the spectral gap) and (ii) a (modified) log-Sobolev inequality.
These correspond contraction in $\chi^2$-divergence and in relative entropy, respectively.

Poincar\'{e} inequalities only imply mixing times that scale polynomially in the system size, $t = \poly(n)$, so they are not strong enough for our purposes.
On the other hand, (modified) log-Sobolev inequalities do imply rapid mixing.
However, establishing such strong functional inequalities in the quantum setting remains notoriously difficult and is currently known only in restricted regimes, most notably for certain classes of commuting Hamiltonians~\cite{crf20,kacr25}.
This motivates us to revisit classical techniques for rapid mixing; among these, \emph{path coupling} offers the simplest and most direct route.


\paragraph{Path coupling.}
In the classical Markov chain literature, path coupling, introduced by Bubley and Dyer~\cite{bd97}, is a technique that proves a contraction in Wasserstein $1$-norm (Earth-Mover distance). The key insight of path coupling is to reduce proving a global contraction to a local contraction.  
Consider imposing a locality graph on the space of configurations; in the case of $\qubits$ bits, this space is $\braces{0,1}^\qubits$, so we could consider the hypercube graph, where there is an edge between two bitstrings if they differ in exactly one bit.
To prove a global contraction between any two configurations $\sigma, \rho \in  \braces{0,1}^\qubits$, it suffices to prove contraction for the \emph{edges} of this locality graph: consider two random variables $Z$, $Z'$ initialized to points $b$ and $b'$ on the hypercube, so the distance between them is 1.
Then, find a way to couple the distributions $\Phi(Z)$ and $\Phi(Z')$ such that the expected distance between them is smaller than $1$:
\begin{equation}
    \label{eq:tech-classical-wass}
    \inf_{\text{coupling } R} \E_{(a, a') \sim R}[\dist(a, a')] \leq (1 - \gamma) \dist(b, b'),
\end{equation}
where $R$ is a coupling if it is a distribution over $2\qubits$ bits whose marginal on the first $\qubits$ bits is $\Phi(Z)$ and second $\qubits$ bits is $\Phi(Z')$ and $a, a'$ are the corresponding points on the hypercube. 
\ewindefer{From Kuikui's lecture note 3}
Finally, for two arbitrary configurations $\sigma, \rho$, consider the shortest path metric on the locality graph. Running the coupling argument for the neighbor couplings along this path independently, along each edge, and using linearity yields a global contraction.



Path coupling solves two issues in bounding mixing times.
First, it reduces the difficult task of showing contraction for all states to the simpler task of showing it along neighboring configurations.
Second, it makes the locality of the dynamics quantitative: a single-site update changes the expected distance by a local rule that depends only on the neighborhood of the updated site.
For example, with path coupling, we can understand which areas of the space contract with a particular jump, and prove statements about mixing on pieces of a larger state space.
For these reasons, it remains a workhorse for proving rapid mixing, including for sampling problems involving graph colorings~\cite{vigoda1999improved,chen2019improved}, independent sets in hypergraphs~\cite{bordewich2006stopping}, and linear extensions of a partial order~\cite{bd97}.

\paragraph{Quantum path coupling.}
Next, we describe our quantum generalization of path coupling. Prior work has extensively studied generalizing transportation metrics to the quantum setting.
We use the quantum Wasserstein norm of order $1$ introduced by De Palma, Marvian, Trevisan, and Lloyd~\cite{dmtl21}. For a traceless Hermitian operator, this norm is defined as follows: 
\begin{equation*}
        \wnorm{X} = \inf_{X_{1} \ldots , X_{n} } \braces*{\frac{1}{2} \sum_{i=1}^n \trnorm{X_{i}} \,\Big|\, \tr_i(X_i) = 0,\, X = \sum_{i=1}^n X_{i},\, X_{i} \text{ Hermitian}}.
\end{equation*}
For any two states, $\rho, \sigma$, we can informally interpret $X = \sigma - \rho $ as the ``signed'' amount of mass that we need to rearrange to turn $\rho$ into $\sigma$. The decomposition of $X$ into $\sum_i X_i$ induces a path between $\sigma$ and $\rho$ with a well-defined notion of a ``neighbor'', i.e.\ we can interpret $X_{i}$ as the portion of $\rho$ that differs from $\sigma$ at site $i$. The condition on $\tr_i(X_{i})$ enforces that the signed mass is only moved between states that differ at site $i$. \aineshdefer{should i do an example? }


This definition recovers classical Wasserstein distance with the hypercube graph as discussed above: for classical density matrices $\rho$ and $\sigma$, $\wnorm{\rho - \sigma}$ is the Wasserstein distance between their respective distributions.
With this, our goal becomes to prove contraction for a quantum channel $\Phi$ in Wasserstein norm: for all states $\rho$ and $\varrho$,
\begin{equation}
    \wnorm{\Phi(\rho) - \Phi(\varrho)} \leq (1 - \gamma) \wnorm{\rho - \varrho}.
\end{equation}
Equivalently, our goal is to prove that for any traceless Hermitian $X$, $\wnorm{\Phi(X)} \leq (1 - \gamma) \wnorm{X}$.
If we can prove this, then the channel $\Phi$ mixes as in \eqref{eq:tech-mixing} in $\bigO{\frac{1}{\gamma} \log\frac{\qubits}{\eps}}$ applications.

\begin{remark}[Prior work on Wasserstein norm]
Wasserstein norm has appeared several times in the literature of mixing times of open quantum systems~\cite{cgkr24,bardet2024entropy,kacr25}.
To our knowledge, though, we are the first to prove contraction in Wasserstein norm.
\end{remark}

\paragraph{A Dobrushin condition for rapid mixing.}
With proving contraction in Wasserstein distance in mind, we now aim to give a Dobrushin condition for rapid mixing.\footnote{
    A historical remark: the Dobrushin condition predates path coupling, and was originally defined as a condition for uniqueness of the Gibbs measure.
    Modern treatments of the Dobrushin condition treat it as a sufficient condition for proving rapid mixing via path coupling; this is the perspective we take here.
    We additionally note that Dobrushin conditions are sometimes defined in terms of the Gibbs distribution instead of the relevant dynamics, since the update rules of Glauber dynamics come from conditional probabilities of the Gibbs distribution.
}
Consider a fully classical channel $\Phi$.
Suppose we can write $\Phi$ as a linear combination of updates, $\Phi = \frac{1}{\qubits}(\Phi_1 + \dots + \Phi_\qubits)$, where the update $\Phi_i$ only changes site $i$, though in a way which can depend on the configuration of other sites.
Classical Glauber dynamics takes this form.
Then we can define a corresponding \emph{influence matrix} $D$, where $D_{i,j}$ corresponds to the amount that an update on site $i$ affects the behavior on site $j$:
\ewindefer{If I did it correctly, $D$ below is exactly the transpose of $R$ as it is defined in the ``matrix norms'' paper}
\ewindefer{Maybe write this as TV distance}
\begin{equation}
    \label{eq:tech-classical-dobrushin}
    D_{i,j}^{(\textup{c})} = \max_{(b, b') \text{ $j$-neighbors}} \frac12\trnorm{\tr_{[\qubits] \setminus \braces{i}}(\Phi_i(\proj{b} - \proj{b'}))}.
\end{equation}
\ewindefer{Add a connector so this sentence doesn't look like it applies to the previous equation}
The influence matrix controls the behavior of $\Phi$ in Wasserstein norm.
For two configurations that differ on site $j$, applying $\Phi_i$, for $i \neq j$, will increase the distance between them by at most $ D^{(\textup{c})}_{i,j}$, since by definition, $D^{(\textup{c})}_{i,j}$  is the maximum amount the marginal at site $i$ can change when the chain updates site $i$. In the case where $i = j$, updating $i$ will move the configurations closer, by decreasing the distance between them by $\geq 1 - D^{(\textup{c})}_{j,j}$.
Since $\Phi$ is a uniform mixture of the $\Phi_i$'s, $\Phi$ contracts the distance between these neighboring configurations provided $\sum_{i \neq j} D^{(\textup{c})}_{i,j} < 1 - D^{(\textup{c})}_{j,j}$.
In other words, if the columns of $D^{(\textup{c})}$ sum to at most $1 - \gamma$, then $\Phi$ is a contraction in Wasserstein metric.
This assumption on the column sums is the Dobrushin condition.\footnote{
    This is technically the \emph{Dobrushin--Shlosman} condition; the \emph{Dobrushin} condition would be a bound on the maximum row sum.
    However, the actual choice of norm is not important: max row sum, max column sum, and spectral norm bounds (operator $\infty$, $1$, and $2$-norm respectively) all suffice, as do any submultiplicative matrix norm in general.
    This can be seen as an easy consequence of~\cite{dgj09}. 
}

The Dobrushin condition, as described above, is specialized to the structure of Glauber dynamics; in particular, it uses that Glauber dynamics updates only one site per iteration.
We cannot expect this to ever hold in the quantum setting; the best we can hope for is that $\Phi$ updates two sites at a time.\footnote{
    In fact, even classical Glauber dynamics no longer satisfies this property in the quantum world, since it needs to measure nearby sites in order to perform its update rule.
    If the sites are qubits, then this changes the state of those sites.
}
However, the principle underlying the proof of rapid mixing extends beyond this setting.
By tracking how a local rule redistributes influence, we obtain Dobrushin-type contraction criteria for any sufficiently local update. This is exactly the level of control needed for analyzing the quantum dynamics.

\begin{remark}[Prior work on mixing conditions]
    We know of one prior attempt to give a Dobrushin-like conditions for rapid mixing in the quantum setting.
    Rouzé, França, and Alhambra prove rapid mixing via contraction in ``oscillator norm'' (in the dual space).
    We perform a more detailed comparison in \cref{rmk:osc}; in short, we observe that this norm is in fact related to Wasserstein norm.
    The relative strength of these norms for proving mixing is not clear, but since Wasserstein norm has a path-coupling interpretation, it has additional flexibility which we use later to prove CMI.
\end{remark}

\paragraph{Transport plan dynamics.}
In the classical world, for a coupling $R$ of two distributions $(\mu, \nu)$, we can define its associated cost vector to be $r_i = \Pr_{(b, b') \sim R}[b_i \neq b_i']$, so that $r_i$ only tracks disagreement between configurations at site $i$. Then, the vector $r= (r_1, r_2, \ldots, r_n)$ represents where the disagreement between $\mu$ and $\nu$ resides. 
By definition, $\sum_{i=1}^\qubits r_i = \E_{(b, b') \sim R}[\dist(b, b')]$, so for an optimal coupling, this sum equals the Wasserstein distance between $\mu$ and $\nu$. 

Recall that Glauber dynamics chooses a site $i$ uniformly at random and updates that site. We can then analyze the effect of this update on a fixed site $j$. If $j=i$, which happens with probability $1/n$, we fix the disagreement between $\mu$ and $\nu$, and by definition it can be at most $D^{(\textup{c})}_{j,j}$. Otherwise, we pick some $i \neq j$, and the disagreement at $j$ can only be sustained by the influence of $i$ on $j$. Therefore, in the worst case, the total contribution to site $j$ is $\sum_{i \in n} D^{(\textup{c})}_{i,j}$, scaled by the chance of picking $i$. Combining these effects together, the cost vector admits the following coordinate-wise update:
\begin{equation*}
    r_j' = \parens[\Big]{1 - \frac{1}{n}} r_j + \frac{1}{n}\sum_{i \in [n]} D^{(\textup{c})}_{i,j} r_i 
\end{equation*}
We can then describe this update succinctly using the matrix $\wh{D} \coloneqq (1 - \frac{1}{\qubits}) I + \frac{1}{\qubits} D^{(\textup{c})}$: if $\mu$ and $\nu$ have a coupling with cost vector $r$, then after one step of Glauber dynamics, we can construct a new coupling between $\Phi(\mu)$ and $\Phi(\nu)$ with cost vector entrywise bounded by $\wh{D}r$.
By iterating this, we can then conclude that
\begin{align*}
    \wnorm{\Phi^t(\mu) - \Phi^t(\nu)} \leq \norm{\wh{D}^{t} r}_1 \leq \norm{\wh{D}}_{1 \to 1}^t \norm{r}_1 = \norm{\wh{D}}_{1 \to 1}^t \wnorm{\mu - \nu}.
\end{align*}
If $\norm{D^{(\textup{c})}}_{1 \to 1} \leq 1 - \gamma$, then $\norm{\wh{D}}_{1 \to 1} \leq 1 - \frac{\gamma}{\qubits}$, and so we get the desired contraction.


We show that this entire argument lifts to the quantum setting.
Quantumly, for a transport plan $(X_1, \dots, X_\qubits)$, we can define its associated cost vector to be $x_i = \frac12 \trnorm{X_i}$.\footnote{
    It turns out that a coupling between classical distributions induces a transport plan whose `quantum' cost vector is equal to the coupling's `classical' cost vector.
}
Then we can consider the dynamics induced on the cost vector: from a transport plan of $\rho - \varrho$ with cost vector $x$, we can construct a transport plan of $\Phi(\rho - \varrho)$ with cost vector $y$ which is entrywise bounded by $Q x$. The update matrix $Q$ determines the behavior of the cost vector under evolution by $\Phi$; if it contracts, then $\Phi$ is rapidly mixing. Crucially, we have now managed to reduce the contraction of a quantum channel $\Phi$ to the contraction of the `classical' update matrix $Q$ applied to a cost vector.
This demonstrates rapid mixing in the same way as before:
\begin{align*}
    \wnorm{\Phi^t(\rho - \varrho)}
    \leq \norm{Q^t x}_1
    \leq \norm{Q}_{1 \to 1}^t \norm{x}_1
    = \norm{Q}_{1 \to 1}^t \norm{\rho - \varrho}_1.
\end{align*}
For continuous-time, these cost vector dynamics take the form of a linear differential equation, $e^{Q t} x$.
With the update matrix in hand, we derive a Dobrushin condition which suffices to show that $\norm{Q}_{1 \to 1}$ is smaller than $1$.
To do so, we define an influence matrix like the classical one:
\begin{equation}
    \label{eq:tech-quantum-dobrushin}
    D_{i,j}^{(\Phi)} = \max_{(\rho, \varrho) \text{ $j$-neighbors}} \wnorm{\Phi_i(\rho - \varrho)}.
\end{equation}
Then, if the column sums of $D^{(\Phi)}$ are smaller than $1$, then $\Phi$ satisfies rapid mixing (\cref{thm:dobrushin-implies-mixing}).
Our influence matrix has one minor difference: we work with quantum Wasserstein distance instead of total variation distance.
This tweak is to handle dynamics which are less structured than Glauber dynamics: if an update acts on more than one site, this will lead to a commensurate increase in Wasserstein distance.

\begin{remark}[Prior work with related techniques]
    Describing path coupling through a dynamics over transport plans appears to be a fairly novel perspective, even in the classical Markov chain literature.
    These dynamics commonly appear implicitly, and are discussed more explicitly in more sophisticated path coupling arguments~\cite{dgj09,weitz04}.
    
    In the quantum information literature, reducing analyzing a quantum dynamics to analyzing a linear differential equation is a technique that appears with some frequency.
    For example, Lieb--Robinson bounds can be proven in this manner~\cite{aly23}.
    However, these linear differential equations typically only have a loose, intuitive relationship to the underlying dynamics.
    The update matrices we define here explicitly control these dynamics in Wasserstein norm, which we find useful for a variety of applications, as we discuss later.
\end{remark}

\paragraph{A detailed balanced Lindbladian.}
We study the quantum channel $\Phi(\rho) = e^{\calL t}(\rho)$, where the Lindbladian $\calL$ is defined as follows: 
\begin{equation}
\label{eqn:lindbladian-intro}
    \calL(\rho) = \sum_{j = 1}^{n} \sum_{P \in \jumps{j}} \underbrace{ - \ii [G^P , \rho] }_{\textrm{coherent term}} + \underbrace{ A^P \rho (A^P)^\dagger - \frac12\braces{(A^P)^\dagger A^P, \rho} }_{\textrm{dissipative term}} \,,
\end{equation}
where $P$ iterates over the set of single-qubit Paulis, $[\cdot , \cdot]$ is the matrix commutator, and $\braces{\cdot, \cdot}$ is the matrix anti-commutator (see~\cite{lidar19} for background on the Lindblad equation). 
With the goal of preparing the Gibbs state $\sigma = e^{-\beta H}/ \tr(e^{-\beta H})$ in mind, we introduce a set of jump operators $A^P =  \sigma^{1/4} P  \sigma^{-1/4}$ and set the corresponding coherent term to be  
\begin{equation*}
    G^P = -\frac{\ii}{2} \cdot \Paren{ (A^P)^\dagger A^P } \circ \braces*{\frac{\sigma_i^{1/2} - \sigma_j^{1/2}}{\sigma_i^{1/2} + \sigma_j^{1/2}}}_{ij}\, , 
\end{equation*}
where $\circ$ is the Hadamard product. We then show that for our choice of jump operators and coherent terms, the Gibbs state $\sigma$ is a fixed point. Further, our Lindbladian satisfies a Kubo--Martin--Schwinger detailed balance condition (see \cref{sec:detail} for details), though this not a necessary condition for our analysis. We analyze the coherent and dissipative terms separately and we show that both these terms are quasi-local. This makes our Lindbladian amenable to applying our mixing time framework to obtain guarantees on how quickly the Lindbladian evolution converges to the Gibbs state.

\paragraph{Proving that the Dobrushin condition holds.}
In light of our framework, to prove rapid mixing, it suffices to show that the Lindbladian evolution satisfies our quantum Dobrushin condition.
Most of our technical work is to describe update matrices for the site-specific updates of the operator $\Phi$, which eventually enable us to bound the column sums. For the purposes of this overview, we assume the Hamiltonian given to us is $2$-local on a $2d$-lattice (see \cref{def:ham-params} for a general definition and the technical sections for the exact dependence on locality and degree). 

We can use the properties of the quantum Wasserstein distance to find update matrices for channels.
We begin by describing a toy example: consider the following two-local channel,
\begin{align*}
    \Phi_i(\rho) = 0.6\rho + 0.3\Delta_i(\rho) + 0.1\Psi_{i,i+1}(\rho),
\end{align*}
which with probability $0.6$ leaves $\rho$ unchanged, with probability $0.3$ applies a fully depolarizing channel (i.e.\ replaces the $i$-th site of $\rho$ with a maximally mixed state), and with probability $0.1$ applies a channel to sites $i$ and $i+1$.
Then, using elementary facts about Wasserstein norm (\cref{subsec:wass}), we can show that $\Phi_i$ has an update matrix of
\begin{align*}
    Q^{(i)} = 0.6 I + 0.3(I - E_{\braces{i}}) + 0.1 E_{\braces{i,i+1}},
\end{align*}
where $E_S$ is the matrix where $E_{i,j} = 1$ if $i \in S$ and $j \in S$.
Here, we use that update matrices combine linearly, so we just need to describe the update matrices of the components of $\Phi_i$.
The depolarizing channel on site $i$ contracts configurations that disagree on site $i$.

The principle behind this local example extends to our quasi-local Lindbladian.
We work with a first-order approximation of the Lindbladian, denoted by $\Phi(\rho) = \calI + \frac{\delta}{n} \calL$, where we eventually take the limit $\delta \to 0$. We then consider decomposing $\Phi(\rho) = \frac{1}{n} \sum_{i \in [n]} \Phi_i(\rho)$, where  
\begin{equation*}
    \Phi_i(\rho ) = \rho + \delta \sum_{P \in \jumps{j}}  \Paren{  - \ii [G^P , \rho]   +   A^P \rho (A^P)^\dagger - \frac12\braces{(A^P)^\dagger A^P, \rho}   }\,.
\end{equation*}
Our task now is to bound the appropriate norm of the quantum Dobrushin influence matrix, $\norm{D^{(\Phi)}}_{1 \to 1}$, where $$ D^{(\Phi)}_{i,j}= \max_{X} \left\{  \frac{1}{n} \wnorm{\Phi_i(X)} \,\Big|\, \tr_j(X)=0 \,,\,   \wnorm{X} = 1  \right\} \,.$$
To this end, we show that $\Phi_i$ has an update matrix of $I + \delta Q^{(i)} $.\footnote{For the purposes of this overview, we drop $\bigO{\delta^2}$ terms.}
The matrix $Q^{(i)}$ inherits the quasi-locality of the Lindbladian:
\begin{equation}
\label{eqn:quantum-update-matrix}
    Q^{(i)} = -4 E_{\braces{i}} + \underbrace{  \sum_{k > 0} \sum_{\substack{\cluster{a} \in [\terms]^k \\ \cluster{a} \text{ cluster from } i}} \sum_{r > 0} \mu_{\cluster{a}, r} \cdot E_{S_{\cluster{a}, r}}  }_{ \Gamma^{(i)} } \, ,
\end{equation}
\ewin{It probably makes more sense for this sum to be over subsets of sites, but that might be too big of a change.}
where $S_{\cluster{a}} = \supp(H_{a_1}) \cup \dots \cup \supp(H_{a_k})$ is the support of cluster $\cluster{a}$ and $S_{\cluster{a},r}$ is the set of all sites within distance $r$ of $S_{\cluster{a}}$.
\ewin{You haven't defined the Hamiltonian yet, so $H_a$ and $m$ are not defined.}\ainesh{done at the top.}
Further,  we can bound $\mu_{\cluster{a}, r} $ by $\bigO{ \beta^k \Paren{4\beta }^{r-1} /( (k-1)!)     }$ (see \cref{thm:wasserstein-growth-from-jump} for a formal statement).
Establishing \cref{eqn:quantum-update-matrix} is the bulk of the technical analysis, and relies on proving quasi-locality properties of the coherent and dissipative terms. We then show that whenever $\beta < \bigO{1}$, $\norm{\Gamma^{(i)} }_{1\to 1} \leq 1$. To obtain this bound, we require understanding the combinatorics of clusters in a graph: in particular, we show that the number of connected clusters of size $\ell$ in the graph induced by the Hamiltonian interactions that contain a particular site $a$ is at most $( c e )^{\ell-1} \ell!$, for a fixed constant $c$. Putting everything together allows us to conclude that $\norm{D^{(\Phi)}}_{1\to 1} \leq 1 - \delta/n$ whenever $\beta < \bigO{1}$, establishing the Dobrushin condition above this critical temperature (see \cref{lemma:gibbs-sat-dobrushin} for a formal statement). The quantum path coupling framework then immediately implies rapid mixing, establishing \cref{thm:rapid-mixing-informal}.
\ewin{Why not assume that $H$ is 2-local on a 2-d lattice and replace all of the parameters with constants?}\ainesh{fixed.}






\paragraph{Clustering of conditional mutual information.}
Next, we describe how to derive the exponential decay of CMI, a structural property of the Gibbs state, from rapid mixing of the Lindbladian evolution, a dynamical property.  We begin with the notion of a \emph{recovery map}, where we imagine that we are given a Gibbs state, except that a region of sites is corrupted and replaced with an arbitrary state. If we have a local recovery map that mixes rapidly, we can repair the corrupted state by operating in a small ball around the corrupted region.
Good recovery maps imply bounds on CMI (\cref{fact:cmi-to-recovery-map}): given a tripartition of $[n]$ into $\calA, \calB, \calC$, and a mixed state $\sigma$ on $[n]$, let $\rho$ be the state attained by discarding the  $\calC$ register of $\sigma$ and applying a recovery map to $\sigma_{\calA \calB}$. Then, $I_{\sigma}(A : C \mid B) \leq \bigO{ \abs{A}} \trnorm{\sigma - \rho}^{1/2}.$

This reduces our task to exhibiting a sufficiently accurate and local recovery map. We pick the recovery map to be the evolution corresponding to the Lindbladian defined in \cref{eqn:lindbladian-intro}, but restricted to a $\Delta$-neighborhood around $\calC$. In particular, let $\calC_{\Delta}$ be the set of sites that are within distance $\Delta$ of $\calC$. Then, the recovery map is $e^{ \calL_{\textrm{CMI} } t}$, where 
\begin{equation*}
    \calL_{\textrm{CMI}}  = \sum_{j \in \calC_{\Delta}} \sum_{P \in \jumps{j}}   - \ii [G^P , \rho]   +   A^P \rho (A^P)^\dagger - \frac12\braces{(A^P)^\dagger A^P, \rho} \,,
\end{equation*}
and the jump operators and coherent terms are restricted to $\calC_{\Delta}$ (see~\cref{def:our-l-truncated} for details).
The key statement we prove is that $\calL_{\textrm{CMI}}$ mixes exponentially fast on $\calC_\Delta$, but the size of the boundary needs to scale linearly in the evolution time:
\begin{equation*}
    \trnorm{ e^{ \calL_{\textrm{CMI}} t}( \sigma_{\calA\calB} \otimes \rho_{\calC} ) - \sigma_{\calA \calB\calC}  } \leq e^{-t/2} \cdot \abs{\calC } + e^{2t} \cdot 10^{-\Delta},
\end{equation*}
which can be balanced by picking $ t = \dist(\calA,\calC)/4$ and $\Delta= \dist(\calA,\calC)/8$. Unfortunately, the dynamics induced by $\calL_{\textrm{CMI}}$ is not strictly local, a necessary condition for the recovery map. However, using the support size of the cost vector corresponding to the transport plan, we can prove that it is well approximated by a truly local channel, by truncating to a $\bigO{\Delta}$ radius around $\calC$.

To get some intuition for why the Lindbladian does not disturb the subsystem $\calA$, we describe the much simpler case where the evolution is discrete and truly local.
Then, for two states $\rho$ and $\sigma$, we can consider the cost vector $r_t$ after $t$ iterations, which corresponds to a transport plan for $\Phi^t(\rho - \sigma)$.
To analyze its behavior, we can imagine applying $\Phi$, to $\rho$ and $\sigma$ in parallel, with the same randomness on which jumps are chosen. The key insight is that the locality of the evolution ensures that disagreements cannot teleport across the system. In particular,  a disagreement at time $t$ at site $j_t \in \calA$ must have been caused by a disagreement at time $t-1$ at a neighbor of $j$, and so on and so forth, resulting in a path $j_1 \to j_2 \ldots  j_{t-1} \to j_t$. We can then bound the probability that the length of the disagreement (measured via the support of the cost vector) is larger than $t$ by the product of the influences for each edge in the path, i.e.\ $ \prod_{\ell \in [t]} \sum_{i \in [n]} D^{(\Phi)}_{i, j_\ell } < (1-\gamma)^t$, for some constant $\gamma$. Since our channel satisfies the quantum Dobrushin condition, this quantity decays exponentially in $t$ and we can sum this over all starting point $j_1$ in $\calC_{\Delta}$ (we refer the reader to \cref{sec:cmi} for a formal proof). Our analysis can be interpreted as an extension of \emph{disagreement percolation} to the quantum setting~\cite{weitz04}.
\ewin{this doesn't make sense.. what is a coupling of evolutions? and I think that if $\Phi$ satisfies Dobrushin and is constant-local, then CMI follows without needing a disagreement percolation argument.}
\ainesh{the coupling argument is how the classical disagreement percolation works,we can change this ify ou want.}

\subsection{Formal results and comparison to related work}
\label{sec:formalresults-relatedwork}

In this section, we formally state our main results.
We begin by defining the class of Hamiltonians and resulting Lindbladians that we work with.

\subsubsection{Our Hamiltonians and Lindbladians} \label{subsubsec:ham}

\begin{restatable}[Hamiltonian]{defn}{ham}
    \label{def:hamiltonian}
    A \emph{Hamiltonian} on $\qubits$ qubits is an operator $H \in \mathbb{C}^{2^\qubits \times 2^\qubits}$ that we consider as a sum of $\terms$ local \emph{terms} $H_a$, with $H = \sum_{a=1}^\terms H_a$.
    We also refer to these qubits as \emph{sites}.
    For normalization, we assume that the terms have bounded operator norm, $\opnorm{H_a} \leq 1$.
\end{restatable}

\begin{definition}[Graphs and distances induced by a Hamiltonian] \label{def:graph-intro}
    For an $\qubits$-qubit Hamiltonian $H = \sum_{a=1}^\terms H_a$, we define its underlying \emph{interaction hypergraph} $\graf$ to have vertices labeled by sites, $\braces{1,2,\dots,\qubits}$, and hyperedges corresponding to $\supp(H_a)$ for every $a \in \braces{1,2,\dots,\terms}$.

    This interaction graph induces a distance metric on sites, which we denote by $\dist$.
    For $i, j \in [\qubits]$,
    \begin{align*}
        \dist(i, j) = \min \braces{k : \text{some } \cluster{a} = (a_1,\dots,a_k) \in [\terms]^k \text{ is a path on } \graf \text{ connecting } i \text{ and } j}.
    \end{align*}
\end{definition}

\begin{definition}[Locality $\locality$, degree $\degree$, and growth parameters $\growth$, $\power$ of a Hamiltonian ]
    \label{def:ham-params}
    Let $H = \sum_a H_a$ be a Hamiltonian.
    The \emph{locality} of $H$, which we denote $\locality$, is the largest support size of any term: $\locality = \max_{a \in [\terms]} \abs{\supp(H_a)}$.
    The \emph{degree} of $H$, which we denote $\degree$, is the maximum number of terms whose support intersects with some $\supp(H_a)$.
    The \emph{exponential growth parameter} $\growth$ of $H$ is the smallest $\growth \geq 1$ such that
    \begin{align*}
        \abs{\braces{j \in [\qubits] \mid \dist(i, j) \leq r}} &\leq \growth^r
        \text{ for all $i \in [\qubits]$ and for all $r \geq 0$.}
    \intertext{Further, the \emph{polynomial growth parameter} $\power$ of $H$ is the smallest $\power \geq 1$ such that}
        \abs{\braces{j \in [\qubits] \mid \dist(i, j) \leq r}} &\leq (1 + r)^\power
        \text{ for all $i \in [\qubits]$ and for all $r \geq 0$.}
    \end{align*}
    Throughout, we assume that $\degree \geq 2$ and $\locality \geq 1$.
\end{definition}

Concretely, local Hamiltonians on constant-degree graphs will have bounded $\locality$, $\degree$, and $\growth$ parameters, but if the graph is an expander $\power$ will scale with system size.
Local Hamiltonians on a constant-degree lattice will have all parameters bounded.
These are technical parameters which only come into play when stating the critical temperature for our results; for most purposes, they can be thought of as constant.

With this class of Hamiltonians in mind, we define a Lindbladian evolution: for an initial state $\rho$, we consider $e^{t \calL^*}(\rho)$, where
\begin{align*}
    \calL^*(\rho) &= \sum_{i = 1}^{\qubits} \calL_i^*(\rho)
    = \sum_{i = 1}^\qubits \sum_{P \in \jumps{i}}  - \ii [G^P , \rho] + A^P \rho (A^P)^\dagger - \frac12\braces{(A^P)^\dagger A^P, \rho}\, .
\end{align*}
In other words, the evolution is a sum of evolutions over each site; each site has three jumps, corresponding to the non-identity Paulis on that site.
For such a 1-local Pauli, we take the corresponding jump operator to be $A^P = \sigma^{1/4} P \sigma^{-1/4}$.
We then pick $G^{P}$ such that $\sigma$ is a fixed point of the evolution (see \cref{def:our-l} for a closed form expression). We note that our Lindbladian is a related to the one introduced by~\cite{ckg23}, but easier to analyze.

\subsubsection{Our formal results}

We show that the Markovian evolution corresponding to this Lindbladian mixes rapidly:


\begin{restatable}[Rapid mixing for high temperature Gibbs states]{thm}{rapid}
    \label{thm:main-rapid}
    Let $H$ be a local Hamiltonian as in \cref{def:ham-params}, and let $\sigma$ be its Gibbs state at inverse temperature $\beta < \beta_c= 1/(10000 \locality^3 \growth^2 \degree)$.  From an arbitrary initial state $\rho$, the Lindbladian evolution (\cref{def:our-l}) run for $\bigO{ \log(n/\eps)}$ time, outputs a state $\wh{\rho}$ such that $\trnorm{ \wh{\rho} - \sigma } \leq \eps$.
\end{restatable}

We note that the normalization above is standard in the quantum literature, since the Lindbladian is defined ``extensively'', i.e.\ a $\delta$ time evolution updates the state by $\bigO{n \delta}$. This corresponds to a $\bigO{ n \log(n)} $ running time classically. 

Recent work of Rouz{\'e}, Fran{\c{c}}a, and Alhambra~\cite{rfa24} also proves rapid mixing for a similar Lindbladian. However, our results are flexible enough to allow us to prove rapid mixing for a discretized Lindbladian, where at each time step we evolve for a constant time in a region around a particular site.



\begin{restatable}[Rapid mixing with a discrete channel]{thm}{discrete}
    \label{thm:main-discrete}
    Let $H$ be a local Hamiltonian as in \cref{def:ham-params}, and let $\sigma$ be its Gibbs state at inverse temperature $\beta < \beta_c= 1/(10000 \locality^3 \growth^2 \degree)$.
    Consider the channel formed by sampling a random site $i \in [\qubits]$ and then performing the constant time Lindbladian evolution $e^{0.1 \calL_i^*}$ (\cref{def:discrete-dynamics}).
    From an arbitrary initial state $\rho$, iterating this channel $\bigO{ n \log(n/\eps)}$ times outputs a state $\wh{\rho}$ such that $\trnorm{ \wh{\rho} - \sigma } \leq \eps.$
\end{restatable}


\begin{remark}[Efficient preparation]
    We do not give formal results about the efficiency of implementing our Lindbladians.
    However, our Lindbladians should be able to be implemented efficiently through identical arguments to previous works~\cite{ckbg23,ckg23}, giving the same $\qubits \polylog(\qubits / \eps)$ gate complexity for Gibbs state preparation as in prior work~\cite{rfa24a}.
    Note that these algorithms state their resources used in terms of Hamiltonian simulation time; since we work with geometrically local Hamiltonians, we can simulate them efficiently without incurring additional factors of system size~\cite{hhkl21}.
\end{remark}

Finally, we use our bounds for rapid mixing of the Lindbladian to construct a recovery map which is accurate enough to obtain clustering of conditioning mutual information, resolving a central question in the theory of open quantum systems:

\begin{restatable}[High-temperature Gibbs states are globally Markovian]{thm}{cmi}
    \label{thm:cmi-main}
    Let $A, B, C$ be a tripartition of $[\qubits]$. Let $H$ be a local Hamiltonian as in \cref{def:ham-params}, and let $\sigma$ be its Gibbs state at inverse temperature $\beta < \beta_c = 1/(\const{} e^{16 \power} \growth \locality^2 (\degree + 1))$.
    Then
    \begin{align*}
        I_\sigma(A : C \mid B) = \bigO{\power^{\power} \abs{A} \cdot \abs{C} }  \cdot \exp\Paren{ - \dist(A, C)/\zeta}\,,
    \end{align*}
    for a fixed universal constant $\zeta>1$.
\end{restatable}

\begin{remark}[Comparison to prior work]
Kuwahara, Kato and Brandao claimed a similar bound in~\cite{kkb20}, which has since been retracted due to an error in the cluster expansion analysis~\cite{kkb25}. Recent work of Chen and Rouz{\'e}~\cite{cr25} obtains a \emph{local Markov property}, where the dependence on the minimum size of $A$ or $C$ is exponential, limiting it to settings where at least one of the subsystems probed is of constant size. However, their bound continues to hold at low temperature. Another recent work of Kato and Kuwahara~\cite{kk25} obtains a global Markov property in the high temperature regime, but their decay rate is sub-exponential in the distance between $A$ and $C$.   
\end{remark}

\section{Preliminaries and background}

\paragraph{Notation.}
First, some basic notation.
We denote $\ii \coloneqq \sqrt{-1}$.
We use the Iverson bracket: for a proposition $P$, $\iver{P}$ is equal to 1 if $P$ is true, and 0 otherwise, so that $\sum_{k \geq 0} \frac{1}{2k+1} = \sum_{\ell \geq 0} \frac{1}{\ell} \iver{\ell \text{ is odd}}$.

We use $\bigO{\cdot}$ and $\bigOmega{\cdot}$ for big O notation.
Throughout, our expressions will include parameters which we will eventually take to be $0$, or limit to zero.
For example, we might write $e^x = \lim_{\delta \to 0} (1 + \delta x + \delta^2 x)^{1/\delta}$.
In such instances, we use a special form of big O notation, $1 + \delta x + \bigOs{\delta^2}$, to denote such terms which will eventually vanish due to these limits (\cref{def:big-o-star}).

\paragraph{Linear algebra.}
We work in the Hilbert space corresponding to a system of $\qubits$ qubits, $\mathbb{C}^2 \otimes \dots \otimes \mathbb{C}^2$; we use $\dims = 2^\qubits$ to refer to the dimension of the space.
We use $\id$ to refer to the identity matrix, but we abuse notation and use $0$ to refer to a zero scalar, vector, or matrix, depending on context.
For a subset of indices $S$, we let $e_S$ to be the indicator vector for the set: it is $1$ when for indices $i \in S$ and $0$ otherwise.
When $\abs{S} = 1$, we sometimes abuse notation to write $e_i$ instead of $e_{\braces{i}}$; we will clarify when this causes confusion.
We let $E_S$ be defined similarly: $E_S = e_Se_S^\dagger$, so in other words, $[E_S]_{i,j} = \iver{i \in S \text{ and } j \in S}$.

For a vector $v$, we use $\norm{v}$ to denote its Euclidean norm and $\norm{v}_1$ to denote its $\ell_1$ norm.
We may also write vectors with the notation $v = (v_i)_{i \in [\dims]}$; when brackets are used instead of parentheses, it denotes a set.
For vectors with real entries, we use $u \vleq v$ to refer to entrywise comparison, i.e.\ $u_i \leq v_i$ for all entries $i$.

For a matrix $A$, we use $A^\dagger$ to denote its conjugate transpose, $\opnorm{A}$ to denote its operator norm (Schatten $\infty$-norm), $\trnorm{A}$ to denote its trace norm (Schatten $1$-norm), and $\norm{A}_{1 \to 1} = \max_{i \in \dims} \norm{Ae_{\braces{i}}}_1$ to denote its $1 \to 1$ operator norm.

We will also work with channels, i.e.\ trace-preserving, completely positive linear maps from density matrices to density matrices.
We will sometimes relax these conditions, and work with more general classes of functions on operators.
We will call such a function a \emph{linear map}, or a \emph{map} for short, if it is linear and takes Hermitian matrices to Hermitian matrices.
We use $\calI$ to denote the identity map (as opposed to $I$, which denotes the identity matrix).

For a linear map on matrices $\Phi$, we use $\dnorm{\Phi}$ to denote its diamond norm:
\begin{definition}[{Diamond norm \cite[Definition 3.43]{watrous18}}]
Let $\Phi$ be a linear map.
Its \emph{diamond norm} is defined by
\[
    \dnorm{\Phi}
    =\sup_{\rho : \trnorm{\rho} = 1 }\,
    \trnorm{ (\Phi\otimes \calI)(\rho) }\,,
\]
where the maximization is over extensions of the domain of $\Phi$.
\end{definition}

A useful basis for us to work in is the basis of (tensor products of) Pauli matrices.

\begin{definition}[Pauli matrices] \label{def:paulis}
    The Pauli matrices are the following $2 \times 2$ Hermitian matrices.
    \begin{equation*}
    \sigma_\id = \begin{pmatrix}
        1 & 0 \\ 0 & 1
    \end{pmatrix}, \qquad \sigma_x = \begin{pmatrix}
        0 & 1 \\
        1 & 0
    \end{pmatrix}, \qquad \sigma_y = \begin{pmatrix}
        0 & -\ii \\
        \ii & 0
    \end{pmatrix}, \qquad \sigma_z = \begin{pmatrix}
        1 & 0\\
        0& -1
    \end{pmatrix}.
    \end{equation*}
    We also consider tensor products of Pauli matrices, $P_1 \otimes \dots \otimes P_\qubits$ where $P_i \in \{\sigma_\id, \sigma_{x}, \sigma_{y}, \sigma_{z}\}$ for all $i \in [\qubits]$.
    The set of such products of Pauli matrices, form an orthogonal basis for the vector space of $2^\qubits \times 2^\qubits$ (complex) Hermitian matrices under the trace inner product.
\end{definition}

\begin{definition}[Support of an operator]
    For an operator $P \in \mathbb{C}^{2^\qubits \times 2^\qubits}$ on $\qubits$ qubits, its \emph{support}, $\supp(P) \subseteq [\qubits]$ is the subset of qudits that $P$ acts non-trivially on.
    That is, $\supp(P)$ is the minimal set of qubits such that $P$ can be written as $P = O_{\supp(P)} \otimes \id_{[n] \setminus \supp(P)}$ for some operator $O$.
\end{definition}

\subsection{Lindbladian evolution}

Next, we introduce the notion of Lindbladian evolution, the quantum analog of continuous-time Markov chain evolution.
We follow the notation from lecture notes of Lidar~\cite{lidar19} and Lin~\cite{Lin24}.

\begin{definition}[General Lindbladian~\cite{lindblad76,gorini1976completely}] \label{def:lindbladian}
For a Hermitian matrix $G \in \C^{2^\qubits \times 2^\qubits}$ and matrices $K_k \in \C^{2^\qubits \times 2^\qubits}$, the Lindbladian with coherent term $G$ and jump operators $K_k$ maps $\rho$ to
\[
    \calL(\rho) = - \ii [G, \rho] + \sum_k \parens[\big]{K_k \rho K_k^\dagger - \frac{1}{2} \braces{K_k^\dagger K_k, \rho} }\,.
\]
The Lindbladian $\calL$ induces a family of quantum channels $e^{\calL t}$ by defining
\begin{align*}
    e^{\calL t} \coloneqq \lim_{\delta \to 0} \parens[\Big]{\calI + \delta \calL}^{t / \delta}.
\end{align*}
\end{definition}
The above family $e^{\calL t}$ is a dynamical semigroup, meaning that $e^{\calL s} e^{\calL t} = e^{\calL (s + t)}$ and $e^{\calL 0} = \calI$.
Conversely, all dynamical semigroups of quantum channels take the form of $e^{\calL t}$ for some $\calL$ of the above form, up to some mild boundedness conditions \cite[Proposition 5]{lindblad76}.
So, there is a formal sense in which all Markovian, time-independent dynamics can be described by Lindbladian evolution.\footnote{
    For those coming from the classical Markov chains literature, recall the analogous result there: $e^{Qt}$ is a valid transition matrix for all $t$ if and only if $Q$'s on-diagonal entries are non-positive, its off-diagonal entries are non-negative, and its columns sum to zero \ewin{Theorem 2.1.2 of Markov Chains Norris}.
}
As with classical systems, quantum thermalization is not always accurately modeled by such a semigroup; the standard derivation of the Lindblad equation shows how it arises under some assumptions.

The norm of a Lindbladian describes the rate at which the system evolves.
One natural choice of norm to use is diamond norm.
We can relate the diamond norm of a Lindbladian to the operator norm of its associated matrices.

\begin{lemma}[Size of a Lindbladian] \label{lem:dnorm-channel}
    Consider the map $\calL(\rho) = - \ii [G,\rho] + K \rho K^\dagger - \frac12 \braces{K^\dagger K, \rho}$.
    Then its diamond norm can be bounded as $\dnorm{\calL} \leq 2\opnorm{G} + 2\opnorm{K}^2$.
\end{lemma}
\begin{proof}
By properties of the diamond norm, and using that $\calL$ maps Hermitian matrices to Hermitian matrices~\cite[Theorem 3.51]{watrous18}, it suffices to bound the trace norm of $(\calL \otimes \calI)(uu^\dagger)$, where $u$ is a unit vector in a higher-dimensional space.
\begin{align*}
    [\calL \otimes \calI](uu^\dagger)
    &= - \ii [G,uu^\dagger] + K uu^\dagger K^\dagger - \frac12 \braces{K^\dagger K, uu^\dagger} \\
    \trnorm{[\calL \otimes \calI](uu^\dagger)}
    &\leq 2\opnorm{G}\trnorm{uu^\dagger} + \opnorm{K}^2 \trnorm{uu^\dagger} + \opnorm{K^\dagger K}\trnorm{uu^\dagger} \\
    &= 2\opnorm{G} + 2\opnorm{K}^2.
\end{align*}
This gives the desired bound on the diamond norm of $\calL$.
\end{proof}

An important feature of Lindbladians is that the diamond norm of $\calI + \calL$ is essentially $1 + \bigO{\dnorm{\calL}^2}$.
The dependence being quadratic and not linear explains why $e^{\calL t}$ is a valid quantum channel.

\begin{lemma}[Size of a Lindbladian timestep] \label{lem:dnorm-lindblad-bound}
    Consider the map $\Phi(\rho) = \rho - \ii [G,\rho] + K \rho K^\dagger - \frac12 \braces{K^\dagger K, \rho}$.
    Then its diamond norm can be bounded by $\dnorm{\Phi} \leq 1 + 2\opnorm{G}^2 + \frac12 \opnorm{K^\dagger K}^2$.
\end{lemma}
\begin{proof}
We use that a map which can be written in terms of Kraus operators, $\rho \mapsto \sum_k K_k \rho K_k^\dagger$, has a diamond norm equal to $\opnorm{\sum_k K_k^\dagger K_k}$.

So, we can rewrite $\Phi$ in terms of Kraus operators:
\begin{align*}
    \Phi(\rho)
    &= \underbrace{K \rho K^\dagger + (\id - \tfrac12 K^\dagger K - \ii G) \rho (\id - \tfrac12 K^\dagger K + \ii G)}_{\Phi_1(\rho)} - \underbrace{(-\tfrac12 K^\dagger K - \ii G) \rho (-\tfrac12 K^\dagger K + \ii G)}_{\Phi_2(\rho)}.
\end{align*}
Then, we bound the diamond norm of $\Phi$ through its Kraus operators.
\begin{align*}
    \dnorm{\Phi}
    &\leq \dnorm{\Phi_1} + \dnorm{\Phi_2} \\
    &= \opnorm[\Big]{K^\dagger K + (\id - \tfrac12 K^\dagger K + \ii G)(\id - \tfrac12 K^\dagger K - \ii G)} + \opnorm[\Big]{(-\tfrac12 K^\dagger K + \ii G)(-\tfrac12 K^\dagger K - \ii G)} \\
    &= \opnorm[\Big]{\id + (-\tfrac12 K^\dagger K + \ii G)(- \tfrac12 K^\dagger K - \ii G)} + \opnorm[\Big]{(-\tfrac12 K^\dagger K + \ii G)(-\tfrac12 K^\dagger K - \ii G)} \\
    &\leq 1 + 2\opnorm[\Big]{(-\tfrac12 K^\dagger K + \ii G)(-\tfrac12 K^\dagger K - \ii G)} \\
    &= 1 + 2\opnorm[\Big]{\frac14 (K^\dagger K)^2 + G^2} \\
    &\leq 1 + \frac12 \opnorm{K^\dagger K}^2 + 2\opnorm{G}^2 
\end{align*}
where the last equality follows from observing that $( K^\dagger K + \ii G)( K^\dagger K - \ii G) =  \Paren{ K^\dagger K}^2 + G^2$.
\end{proof}

\begin{definition}[Big O for infinitesimals] \label{def:big-o-star}
    When working with a Lindbladian evolution, we will typically work with $\delta$-sized time steps, eventually taking $\delta \to 0$:
    \begin{align*}
        e^{t \calL}(\rho) = \lim_{\delta \to 0} (\calI + \delta \calL)^{t/\delta}(\rho).
    \end{align*}
    When working with these $\delta$-sized time steps, errors on the order of $\delta^2$ will be zero after this limit is taken.
    As a result, we will introduce the big O notation $\bigOs{\delta^2}$ to refer to big O notation where all problem parameters like $\terms$ and $\qubits$ are suppressed, and we only keep the dependence on parameters like $\delta$ which will eventually be taken to zero.
\end{definition}

For example, applying \cref{lem:dnorm-lindblad-bound} to a $\delta$-sized Lindblad step, we get the following:

\begin{corollary}[Trace norm growth under a Lindblad step]
\label{cor:trace-norm-lindblad-step}
Let $\Phi(X) = X - \delta(i [G, X] + K\rho K^\dagger - \frac12 \braces{K^\dagger K, \rho})$ for some operator $G$.
Then, using that $\trnorm{\Phi(X)} \leq \dnorm{\Phi} \cdot \trnorm{X}$,
\begin{align*}
    \trnorm{\Phi(X)} \leq \parens{1+\bigOs{\delta^2}} \cdot \trnorm{X}.
\end{align*}
\end{corollary}

\subsection{Quantum Wasserstein distance} \label{subsec:wass}

We use the quantum Wasserstein distance of order 1, introduced by De Palma, Marvian, Trevisan, and Lloyd~\cite{dmtl21}, as a distance metric over quantum states.
We describe the properties of this metric, and refer the reader to the original work for the proofs of these properties.

\begin{definition}[Quantum Wasserstein distance {\cite[Definition 6]{dmtl21}}]
\label{def:quantum-wasserstein}
    For a traceless Hermitian operator $X$, we define its quantum Wasserstein norm of order 1 as follows:
    \begin{align*}
        \wnorm{X} = \min \parens*{\frac{1}{2} \sum_{i=1}^n \trnorm{X_{i}} \,\Big|\, \tr_i(X_{i}) = 0,\, X = \sum_{i=1}^n X_{i},\, X_{i} \text{ Hermitian}}
    \end{align*}
    This minimum is achieved by some $\parens{X_{i}}_{i \in [\qubits]}$.
\end{definition}

This distance is a transportation metric: for $X = \rho - \sigma$, it represents the smallest mass that must be moved to take $\rho$ to $\sigma$, where mass can only be moved across edges of the Boolean cube.\footnote{
    More precisely, on classical distributions, this metric reduces to classical Wasserstein (Earth Mover) distance: for $p$ and $q$ distributions over $\braces{0,1}^\qubits$, $W_1(p, q) = \min_{\pi \text{ coupling of } (p, q)} \E_{(x,y) \sim \pi}[h(x, y)]$ where $h(x, y)$ is the distance is the Hamming distance between bitstrings $x$ and $y$.
}
In order to get finer control over these transportation costs, we introduce the notion of a `plan' for how to move mass, and the site-specific cost associated with a plan.

\begin{definition}[Transport plan and cost vector]
    For a traceless Hermitian matrix $X$, we say it has a \emph{transport plan} of $\parens{X_i}_{i \in [\qubits]}$ if:
    \begin{enumerate}[label=(\alph*)]
        \item The $X_i$'s are also Hermitian and traceless;
        \item $X = \sum_{i=1}^\qubits X_i$;
        \item For every $i \in [\qubits]$, $\tr_i(X) = 0$.
    \end{enumerate}
    For a transport plan $\parens{X_i}_{i \in [\qubits]}$, we define its \emph{cost vector} to be the length-$\qubits$ vector of trace norms,
    \begin{align*}
        \mu(\parens{X_i}_i)
        = \parens[\Big]{\frac12 \trnorm{X_i}}_{i \in [\qubits]}
        = \Paren{\wnorm{X_i}}_{i \in [\qubits]}.
    \end{align*}
    (The last equality follows from \cref{lem:w1-to-trace}.)
\end{definition}
With these definitions, we can write $\wnorm{X}$ as the problem of minimizing the $\ell_1$ cost over transport plans of $X$:
\begin{align*}
    \wnorm{X} = \min_{\parens{X_i}_i \text{ plan}} \norm{ \mu(\parens{X_i}_i) }_1.
\end{align*}
Because transport plans are linear, they combine in the expected way.
\begin{fact}[Linearity of transport plans] \label{fact:linear-plans}
    If two traceless Hermitian matrices $X$ and $Y$ have transport plans $\parens{X_i}_{i \in [\qubits]}$ and $\parens{Y_i}_{i \in \qubits}$ with cost vectors $x$ and $y$, then for all $\alpha, \beta \in \R$, the matrix $\alpha X + \beta Y$ has a transport plan $\parens{\alpha X_i + \beta Y_i}_{i \in [\qubits]}$ with a cost vector $z$ which satisfies $z \vleq \alpha x + \beta y$.
\end{fact}


Quantum Wasserstein distance can be related to trace distance, in the same way that they can be related in the classical setting.

\begin{lemma}[Relating Wasserstein to trace~{\cite[Proposition 2]{dmtl21}}] \label{lem:w1-to-trace}
    For a traceless matrix $X$,
    \begin{align*}
        \frac12 \trnorm{X} \leq \wnorm{X} \leq \frac{n}{2} \trnorm{X}.
    \end{align*}
    Moreover, if $\tr_i(X) = 0$ for some $i \in [\qubits]$, then there is a transport plan with cost vector $\frac12 \trnorm{X} e_i$. In particular, this implies that $\frac12\trnorm{X} = \wnorm{X}$.
\end{lemma}

We can refine the upper bound on quantum Wasserstein when the operator $X$ is zero after tracing out a small subsystem, $\tr_S(X) = 0$.
Again recalling that we will eventually take $X = \rho - \sigma$, if $\tr_S(X) = 0$, then $\rho$ and $\sigma$ are identical outside of $S$.
So, intuitively, to convert $\rho$ to $\sigma$, one only has to move mass within $S$.
The following lemma formalizes this intuition.

\begin{lemma}[{Relating Wasserstein to trace for few-site deviations \cite[Proposition 5]{dmtl21}}] \label{lem:w1-few-site}
    Let $X$ be a traceless Hermitian operator such that, for some $S \subseteq [\qubits]$, $\tr_S(X) = 0$.
    Then there is a transport plan of $X$ whose cost vector $x$ satisfies $x \vleq \trnorm{X} e_S$.
    In particular, this implies that $\wnorm{X} \leq \abs{S} \trnorm{X}$.
\end{lemma}

As a corollary, the Wasserstein distance of a commutator can be bounded via trace norm, if one of the operators in the commutator has small support.
This uses the cyclicity of partial trace.

\begin{fact}[Cyclicity of partial trace] \label{fact:cyclicity}
    For two matrices $A$ and $B$, $\tr_{\supp(A)}(AB) = \tr_{\supp(A)}(BA)$.
\end{fact}

Thus, we know that $\tr_{\supp(A)}([A, B]) = 0$, and can conclude the following.

\begin{corollary}[Wasserstein norm of a commutator] \label{lem:w1-commutator}
    For two Hermitian matrices $A$ and $B$, the Hermitian matrix $\ii[A, B]$ has a transport plan whose cost vector $x$ satisfies $x \vleq \trnorm{\ii[A, B]} \cdot  e_{\supp(A)}$.
    In particular, $\wnorm{\ii [A, B]} \leq \abs{\supp(A)} \trnorm{\ii [A, B]}$.
\end{corollary}

We will eventually be trying to show that certain maps act as contractions in Wasserstein norm, $\wnorm{\Phi(X)} < \wnorm{X}$.
To do this, we need to track how to take a transport plan of $X$ and construct a transport plan for $\Phi(X)$, mirroring coupling arguments for classical Markov chains.
A useful notion for us will be that of an \emph{update matrix}.

\begin{definition}[Update matrix] \label{def:wasserstein-dynamics}
    Let $\Phi$ be a trace-preserving map.
    Then $Q$ is an update matrix of $\Phi$ if, for any traceless Hermitian matrix $X$ and any transport plan $\parens{X_i}_i$ of $X$, the matrix $\Phi(X)$ has a transport plan with cost vector $y \vleq Q x$, where $x$ is the cost vector of $\parens{X_i}_i$.
\end{definition}

The behavior of cost vectors under evolution need not be linear: we merely bound it in a linear way.
A linear bound is natural because transport plans (and therefore cost vectors) can be combined linearly (\cref{fact:linear-plans}).
This will be sufficient for us.

For a continuous-time evolution, its update matrix takes the form of a matrix exponential.

\begin{lemma}[Update matrices of continuous-time evolutions]
    \label{lem:update-matrix}
    Let $\calL$ be a trace-preserving map such that, for all sufficiently small $\delta > 0$, $\calI + \delta \calL$ has an update matrix of the form $\id + \delta Q + \bigOs{\delta^2}$.
    Then, for all $t \geq 0$, $e^{\calL t} = \lim_{\delta \to 0} [\calI + \delta \calL]^{t/\delta}$ is a map with an update matrix of the form $e^{Qt}$.
\end{lemma}
\begin{proof}
Update matrices compose as expected, in the sense that if $\Phi$ and $\Phi'$ have update matrices of $Q$ and $Q'$, then $Q' Q$ is an update matrix of $\Phi'$ composed with $\Phi$.
Consequently, $[\calI + \delta \calL]^{t/\delta}$ has an update matrix of the form $(\id + \delta Q + \bigOs{\delta^2})^{t / \delta}$ for sufficiently small $\delta > 0$.
Taking $\delta \to 0$, we get that $e^{Q t}$ is an update matrix for $e^{\calL t}$.
\end{proof}

\subsection{Hamiltonian class} \label{subsec:ham}

We now describe the class of Hamiltonians we work with, expanding on the exposition in \cref{subsubsec:ham}.
We work with local Hamiltonians on qubit systems; we expect our results to extend to other kinds of systems in a natural way.
Recall our definition of a Hamiltonian.

\ham*

\begin{definition}[$\locality$-locality of a Hamiltonian]
    Let $H = \sum_a H_a$ be a Hamiltonian.
    The \emph{locality} of $H$, which we denote $\locality$, is the largest support size of any term: $\locality = \max_{a \in [\terms]} \abs{\supp(H_a)}$.
\end{definition}

We will be referring to supports of terms often, so we give them a shorthand.

\begin{definition}[Notation for a cluster]
For a sequence of terms $\cluster{a} = (a_1,\dots,a_k)$, we let
\begin{align}
    S_{\cluster{a}} = \supp(H_{a_1}) \cup \dots \cup \supp(H_{a_k})
\end{align}
be the supports of the associated terms.
We use boldface letters $\cluster{a}$ to refer to sequences of terms.
With this in mind, we will sometimes write the summation $\sum_{k \geq 0} \sum_{a_1,\dots,a_k \in [\terms]}$ as $\sum_{\cluster{a}}$, with the understanding that we are summing over $\cluster{a} \in \bigcup_{k \geq 0} [\terms]^k$.
Typically, $\cluster{a}$ will also obey a connectedness ``cluster'' property which we define in \cref{lem:hadamard}; we will specify when the sum is only over clusters.
\end{definition}

\begin{definition}[Restricted Hamiltonian]
    For a subset of sites $S \subseteq [\qubits]$, we define the Hamiltonian restricted to those sites to be
    \begin{align*}
        H^{(S)} = \sum_{\substack{a \in [\terms] \\ \supp(H_a) \subseteq S}} H_a.
    \end{align*}
\end{definition}

\subsubsection{Geometrically local Hamiltonians}

\begin{definition}[Graphs and distances induced by a Hamiltonian]
    For an $\qubits$-qubit Hamiltonian $H = \sum_{a=1}^\terms H_a$, we define its underlying \emph{interaction hypergraph} $\graf$ to have vertices labeled by sites, $\braces{1,2,\dots,\qubits}$, and hyperedges corresponding to $\supp(H_a)$ for every $a \in \braces{1,2,\dots,\terms}$.

    This interaction graph induces a distance metric on sites, which we denote by $\dist$.
    For $i, j \in [\qubits]$,
    \begin{align*}
        \dist(i, j) = \min \braces{k : \text{some } \cluster{a} = (a_1,\dots,a_k) \in [\terms]^k \text{ is a path on } \graf \text{ connecting } i \text{ and } j}.
    \end{align*}
    We can similarly extend this metric to sets, where $\dist(S, T)$ is the length of the smallest path connecting $S$ and $T$.
    We will consider balls in this metric:
    \begin{align*}
        \ball(S, r) = \braces[\big]{i \in [\qubits] : \dist(i, S) \leq r}.
    \end{align*}
    We define $H$'s underlying \emph{dual interaction graph} $\graph$ to have vertices labeled by $\braces{1,2,\dots,\terms}$ and an edge between $a$ and $b$ if and only if $\supp(H_a) \cap \supp(H_b) \neq \varnothing$.
\end{definition}

\begin{definition}[Degree $\degree$ of a Hamiltonian]
    Let $H = \sum_a H_a$ be a Hamiltonian.
    The \emph{degree} of $H$, which we denote $\degree$, is the maximum degree of a vertex in $\graph$.
\end{definition}

The degree of a Hamiltonian is sometimes defined as the degree of the hypergraph $\graf$; these quantities differ by a factor of at most $\locality$.
The locality and degree parameters are more standard definitions: the class of Hamiltonians where these are bounded are known as \emph{low-intersection} Hamiltonians.

\begin{definition}[Growth parameters $\growth$, $\power$ of a Hamiltonian]
    \label{def:ham-growth-parameter}
    Let $H = \sum_a H_a$ be a Hamiltonian.
    The \emph{exponential growth parameter} $\growth$ of $H$ is the smallest $\growth \geq 1$ such that
    \begin{align*}
        \abs{\braces{j \in [\qubits] \mid \dist(i, j) \leq r}} &\leq \growth^r
        \text{ for all $i \in [\qubits]$ and for all $r \geq 0$.}
    \intertext{Further, the \emph{polynomial growth parameter} $\power$ of $H$ is the smallest $\power \geq 1$ such that}
        \abs{\braces{j \in [\qubits] \mid \dist(i, j) \leq r}} &\leq (1 + r)^\power
        \text{ for all $i \in [\qubits]$ and for all $r \geq 0$.}
    \end{align*}
    
\end{definition}

We define our exponential and polynomial growth parameters in terms of sizes of balls in $\graph$, but it's common to see these alternatively defined as the constants which make the series $\sum_{j \in [\qubits]} \growth^{-\dist(i, j)}$ and $\sum_{i \in [\qubits]} \frac{1}{(1 + \dist(i, j))^\power}$ converge~\cite{hk06}.
These definitions are equivalent up to constants.

Concretely, local Hamiltonians on constant-degree graphs will have bounded $\locality$, $\degree$, and $\growth$ parameters, but if the graph is an expander $\power$ will scale with system size.
Local Hamiltonians on a constant-degree lattice will have all parameters bounded.

Bounded polynomial growth implies bounded exponential growth, since from the definition we know that $\growth \leq e^\power$.
Further, bounded locality and degree imply bounded exponential growth, by an induction argument: $\growth \leq (\degree + 1)(\locality - 1)$.

\begin{fact}[Bounds on growth parameter for local Hamiltonians on graphs]
    If a Hamiltonian $H = \sum_a H_a$ is $\locality$-local and has degree $\degree$, then $\growth \leq (\degree + 1)(\locality - 1)$.
\end{fact}
\begin{proof}
Fix a site $i \in [\qubits]$.
This site $i$ has at most $\degree + 1$ terms which are supported on it; then, each such term adds at most $\locality-1$ new sites to $\abs{\ball(i, 1)}$.
By inducting this argument, we have that $\abs{\ball(i, r)} \leq ((\degree + 1)(\locality - 1))^r$.
\end{proof}

Throughout this paper, we will assume an arbitrary Hamiltonian as specified in \cref{def:hamiltonian}, and use all of the above parameters at will.
We assume throughout the non-degeneracy conditions that $\degree \geq 2$ and $\locality \geq 1$.

\subsubsection{Lemmas on converging series}

Next, we describe how to prove convergence properties of several series that appear in our analysis.
These arguments follow from analyzing the combinatorics associated to the Hamiltonian; the analyses in later sections will have the same basic themes, but with greater sophistication in order to control Wasserstein norm.

\begin{lemma} \label{lem:hadamard}
    For all matrices $X$ and all $\alpha \in \C$, the following equality holds:
    \begin{align*}
        e^{\alpha H} X e^{-\alpha H}
        = \sum_{k \geq 0} \sum_{a_1,\dots,a_k \in [\terms]} \frac{\alpha^k}{k!} [H_{a_k},[\dots[H_{a_2},[H_{a_1}, X]]\dots]].
    \end{align*}
    Further, the size of the $k$th order term can be bounded as follows:
    $$\opnorm{[H_{a_k},[\dots[H_{a_2},[H_{a_1}, X]]\dots]]} \leq 2^k \opnorm{X}.$$
    Moreover, this nested commutator is zero unless, for all $\ell \in [k]$, $H_{a_\ell}$'s support intersects the support of some $H_{a_1},\dots,H_{a_{\ell-1}}$ or $X$.
    When $\cluster{a} = (a_1,\dots,a_k)$ satisfies this criterion, we call it a \emph{cluster from $\supp(X)$}; if $\supp(X) \subseteq \supp(H_b)$ for some $b \in [\terms]$, the number of clusters of length $k$ from $\supp(X)$ is at most $k!(\degree+1)^k$.
\end{lemma}

\begin{proof}
The equality is the Hadamard formula:
\begin{align*}
    e^{\alpha H} X e^{-\alpha H}
    &= \sum_{k \geq 0} \frac{1}{k!} [\alpha H, X]_k
\end{align*}
where $[\alpha H, X]_k$ denotes the nested commutator, $[A, B]_k = [A, [A, B]_{k-1}]$ with $[A, B]_1 = [A, B] = AB - BA$.
To get the equality in the lemma statement, we use bi-linearity of the commutator to write
\begin{align*}
    [\alpha H, X]_k = \alpha^k \sum_{a_1,\dots,a_k} [H_{a_k},[\dots[H_{a_2},[H_{a_1}, X]]\dots]].
\end{align*}
The operator norm bound on this commutator follows from observing that, for any $A$ and $B$, $\opnorm{[A, B]} \leq 2\opnorm{A}\opnorm{B}$ by triangle inequality.
Iterating this argument, we get $$\opnorm{[H_{a_k},[\dots[H_{a_2},[H_{a_1}, X]]\dots]]} \leq 2^k \opnorm{X}.$$
The nested commutator $[H_{a_k},[\dots[H_{a_2},[H_{a_1}, X]]\dots]]$ is equal to zero if some $H_{a_\ell}$ commutes with $[H_{a_{\ell - 1}},[\dots[H_{a_2},[H_{a_1}, X]]\dots]]$: a necessary condition for not commuting is that the $\supp(H_{a_\ell})$ intersects $\supp(H_{a_{\ell - 1}}) \cup \dots \cup \supp(H_{a_{\ell - 1}}) \cup \supp(X)$, because if that is the case, the $\ell$th commutator in the nested commutator is between two matrices with disjoint supports.

When $\supp(X) \subseteq \supp(H_b)$ for some $b \in [\terms]$, we can count the number of clusters in the dual interaction graph $\graph$: this necessary condition holds if $a_{\ell}$ neighbors at least one of $a_{\ell - 1}, \dots, a_1, b$.
There are at most $\ell(\degree + 1)$ options for $a_\ell$ which satisfy the cluster property; thus, the number of clusters of size $k$ is at most
\begin{equation*}
    \prod_{\ell = 1}^k (\ell(\degree + 1)) \leq \ell! (\degree + 1)^\ell. \qedhere
\end{equation*}
\end{proof}

\begin{lemma}[Bounded support Hadamard expansion] \label{lem:had-bounds}
    When $0 \leq \beta < 1/(2(\degree + 1))$ and $X$ is a matrix with $\supp(X) \subseteq \supp(H_b)$ for some $b \in [\terms]$,
    \begin{align}
        \sum_{k \geq 0} \sum_{a_1,\dots,a_k \in [\terms]} \opnorm[\Big]{\frac{\beta^k}{k!} [H_{a_k},[\dots[H_{a_2},[H_{a_1}, X]]\dots]]}
        &\leq \frac{1}{1 - 2\beta(\degree + 1)}\label{eq:had-exp-1}
    \end{align}
\end{lemma}
\begin{proof}
We use \cref{lem:hadamard} to conclude that
\begin{align*}
    \opnorm[\Big]{\frac{\beta^k}{k!} [H_{a_k},[\dots[H_{a_2},[H_{a_1}, X]]\dots]]}
    \leq \frac{(2\beta)^k}{k!}.
\end{align*}
Further, when $(a_1,\dots,a_k)$ is not a cluster, the above expression is zero.
So,
\begin{align*}
    \text{\eqref{eq:had-exp-1}}
    &\leq \sum_{k \geq 0} \sum_{a_1,\dots,a_k \in [\terms]} \frac{(2\beta)^k}{k!} \iver{(a_1,\dots,a_k) \text{ is a cluster from } \supp(X)} \\
    &= \sum_{k \geq 0} \frac{(2\beta)^k}{k!} \sum_{a_1,\dots,a_k \in [\terms]} \iver{(a_1,\dots,a_k) \text{ is a cluster from } \supp(X)} \\
    &\leq \sum_{k \geq 0} \frac{(2\beta)^k}{k!} (k! (\degree + 1)^k) \\
    &= \sum_{k \geq 0} (2\beta(\degree + 1))^k \\
    &= \frac{1}{1 - 2\beta(\degree + 1)}. \qedhere
\end{align*}
\end{proof}

\begin{lemma}[Proposition 3.6 in~\cite{hkt21}] \label{lem:num-clusters}
    When $\degree \geq 2, \locality\geq 1$, the number of connected clusters of size $\ell$ containing a particular $a \in [\terms]$ is at most $e^2 (e\degree)^{\ell - 1} \ell!$.
\end{lemma}
\begin{proof}
    Note that \cite{hkt21} gives a bound for \emph{unordered} clusters, i.e.\ multisets of size $\ell$.
    We get a bound for ordered clusters by multiplying by the number of orderings, which is at most $\ell!$.

    This bound is slightly tighter than the stated bound in \cite{hkt21}.
    However, we can conclude this bound from the tight closed form bound for the number of connected clusters presented in Equations 28 and 29 of the proof of Proposition 3.6 in \cite{hkt21}.
    \begin{align*}
        & \sum_{k = 1}^\ell \frac{\degree}{k(\degree - 1) + 1} \binom{k(\degree - 1) + 1}{k - 1} \binom{\ell - 1}{k - 1} \\
        &\leq \sum_{k = 1}^\ell \parens[\Big]{e\frac{k(\degree - 1) + 1}{k - 1}}^{k-1} \binom{\ell - 1}{k - 1} \\
        &= \sum_{k = 1}^\ell (e(\degree - 1))^{k-1} \parens[\Big]{1 + \frac{\degree}{(\degree - 1)(k-1)}}^{k-1} \binom{\ell - 1}{k - 1} \\
        &\leq e^2 \sum_{k = 1}^\ell (e(\degree - 1))^{k-1} \binom{\ell - 1}{k - 1} \\
        &= e^2 (e(\degree - 1) + 1)^{\ell - 1}
    \end{align*}
    In the last inequality, we use that $\degree \geq 2$; the base of the exponent is at most $e \degree$.
\end{proof}

\begin{corollary} \label{lem:num-clusters-quasi}
    For some site $i \in [\qubits]$, summing over all the clusters of size $\ell \geq 1$, we have that
    \begin{align*}
        \sum_{\cluster{a} \in [\terms]^\ell} (10 \growth)^{-\dist(S_{\cluster{a}}, i)} \leq 2 \locality (2 e \degree)^{\ell - 1} \ell!.
    \end{align*}
\end{corollary}
\begin{proof}
We reorganize the sum according to the distance from $S_{\cluster{a}}$ to $i$.
Then, the number of clusters at a given distance from $i$ can be related to the size of balls centered at $i$.
\begin{multline}
        \sum_{\cluster{a} \in [\terms]^\ell} (2 \growth)^{-\dist(S_{\cluster{a}}, i)}
        = \sum_{r \geq 0} \sum_{\substack{\cluster{a} \in [\terms]^\ell\\ \dist(S_{\cluster{a}}, i) = r }} (2 \growth)^{-r} \\
        \leq \sum_{r \geq 0} (\# \textrm{ of $j \in [\qubits]$ distance $r$ from $i$}) \cdot (\# \textrm{ of clusters containing $j$}) \cdot  (2 \growth)^{-r} \\
        \leq \sum_{r\geq 0} \growth^r \cdot (\degree + 1) \cdot e^2(e \degree)^{\ell - 1} \ell! \cdot (2 \growth)^{-r}  \leq 2e^2 (e \degree)^{\ell - 1} \ell!
\end{multline}
The number of clusters of size $\ell$ which are distance $r$ from $i$ can be counted as follows.
Such a cluster $\cluster{a}$ must contain a site which is distance $r$ from $i$.
There are $\growth^r$ such sites, by \cref{def:ham-growth-parameter}.
For each site, there are at most $\degree + 1$ terms which include it.
By \cref{lem:num-clusters}, every term has at most $e^2(e\degree)^{\ell - 1}$ clusters which include it.
\end{proof}

\begin{corollary} \label{lem:num-clusters-quasi-2}
    For some site $i \in [\qubits]$, summing over all the clusters of size $\ell \geq 1$, we have that
    \begin{align*}
        \sum_{\cluster{a} \in [\terms]^\ell} (2 \growth)^{-\max(0,\dist(S_{\cluster{a}}, i) - 1)} \leq 3 e^2 \growth (e \degree)^{\ell - 1} \ell!\,.
    \end{align*}
\end{corollary}
\begin{proof}
We proceed similarly as above, but handle the case where $r = 0$ separately.
This gives the following bound.
\begin{align*}
    \sum_{\cluster{a} \in [\terms]^\ell} (2 \growth)^{-\max(0,\dist(S_{\cluster{a}}, i) - 1)}
    \leq (\degree + 1) e^2(e \degree)^{\ell - 1} \ell!\parens[\Big]{1 + \sum_{r \geq 1} \growth^r (2\growth)^{-(r-1)}}
    \leq e^2(e \degree)^{\ell - 1}(1 + 2\growth) \ell!
\end{align*}
Since $\growth \geq 1$, this implies the desired statement.
\end{proof}

\section{Defining and understanding our Lindbladian}

In this section, we introduce a family of Lindbladian jump operators that are essentially local and maintain the Gibbs state as the unique fixed point.
We find this Lindbladian convenient to work and obtain thermalization bounds for. 

\begin{definition}[Balanced Lindbladian]
\label{def:our-l}
Given a target Gibbs state $\sigma = \frac{e^{-\beta H}}{\tr(e^{-\beta H})}$ at inverse temperature $\beta>0$, we consider the Lindblad evolution given by
\begin{align*}
    \calL^*(\rho) = \sum_{j = 1}^\qubits \calL_j^*(\rho)
    &= \sum_{j = 1}^\qubits \sum_{P \in \jumps{j}} \calL^P(\rho) \\
    &\text{ where } \calL^P(\rho) = - \ii [G^P, \rho] + A^P \rho (A^P)^\dagger - \frac12\braces{(A^P)^\dagger A^P, \rho}.
\end{align*}
In other words, there is a Lindbladian term for every $1$-local Pauli matrix.
For a (Hermitian but not necessarily Pauli) jump operator $P$, we define the jump operators $A^P$ as follows:
    \begin{align*}
        A^P =  \sigma^{1/4} P \sigma^{-1/4} .
    \end{align*}
    Further, we take $G^P$ such that in the basis where $\sigma$ is diagonal,
    \begin{align*}
        G_{i,j}^P &= -\ii\parens*{\frac{\sigma_i^{1/2} - \sigma_j^{1/2}}{2\cdot \sigma_i^{1/4}\cdot \sigma_j^{1/4}(\sigma_i^{1/2} + \sigma_j^{1/2})}} \parens*{  \sum_{k} \sigma_k^{1/2} P_{ik}P_{kj}}.
    \end{align*}
    Note that $G^P$ is Hermitian and can be rewritten as follows, where $\circ$ denotes entrywise product:
    \begin{align*}
        G^P & = -\frac{\ii}{2} \cdot \Paren{ (A^P)^\dagger A^P } \circ \braces*{\frac{\sigma_i^{1/2} - \sigma_j^{1/2}}{\sigma_i^{1/2} + \sigma_j^{1/2}}}_{ij}.
    \end{align*}
\end{definition}

\begin{remark}[Comparison to other Lindbladians]
    Though we will analyze this Lindbladian, we expect our arguments to extend naturally to other commonly considered local Lindbladians, including those posed by \cite{ckg23, ding2024simulating, sa25}.
    The major difference between \cref{def:our-l} and \cite{ckg23} is that our jump term is moderately simpler, being  the jump operator under imaginary time evolution, $\sigma^{1/4} P \sigma^{-1/4}$.
    The cost of this simplicity is that this jump operator, and therefore, the Lindbladian is non-local below a critical temperature.
    Nevertheless, this Lindbladian may have further applications, particularly to settings in which the jump operator is still local.
\end{remark}

\ewin{Would be nice if there was a summary of the main results in this section.}

\subsection{Basic properties}

We now analyze this Lindbladian.
First, as desired, the Gibbs state is a fixed point for it.

\begin{lemma}[Stationarity of the Gibbs state]
    \label{lem:stationarity}
    Let $\calL^P$ and $\calL^*$ be the Lindbladians introduced in \cref{def:our-l}.
    Then the Gibbs state $\sigma$ is a fixed point of these Lindbladians: $\calL^P(\sigma) = 0$ for every Hermitian $P$, and moreover, $\calL^*(\sigma) = 0$.
\end{lemma}
\begin{proof}
Because $\calL^*(\sigma)$ is a sum over $\calL^P(\sigma)$ for all $1$-local Pauli matrices, it suffices to show that $\sigma$ is a fixed point of every $\calL^P(\sigma)$.
First, we use that the jump operators are $A^P = \sigma^{1/4} P \sigma^{-1/4}$, so that $A^P \sigma (A^P)^\dagger = \sigma^{1/2} (A^P)^\dagger A^P \sigma^{1/2}$.
Then, we have
\begin{align*}
    \calL^P(\sigma)
    &=  - \ii [G^P , \sigma] + \parens*{A^P \sigma (A^P)^\dagger - \frac12\braces{(A^P)^\dagger A^P, \sigma}} \\
    &= - \ii [G^P, \sigma] + \parens*{\sigma^{1/2} (A^P)^\dagger A^P \sigma^{1/2} - \frac12\braces{(A^P)^\dagger A^P, \sigma}}
\end{align*} 
We choose to work in the energy basis, so that $\sigma = \diag(\braces{\sigma_i})$.
\begin{align*}
    \calL^P(\sigma)_{i,j}
    &= - \ii \cdot G^P_{i,j}(\sigma_j - \sigma_i) +  \parens[\Big]{\sigma_i^{1/2}\sigma_j^{1/2} \bracks[\Big]{ (A^P)^\dagger A^P }_{i,j} - \frac{\Paren{ \sigma_i + \sigma_j}}{2} \bracks[\Big]{ (A^P)^\dagger A^P }_{i,j}  } \\
    &= - \ii \cdot G^P_{i,j}(\sigma_j - \sigma_i) - \frac{1}{2}\parens{\sigma_i^{1/2} - \sigma_j^{1/2}}^2\bracks[\Big]{ (A^P)^\dagger A^P }_{i,j}
\end{align*}
Finally, recall that we set $G^{P}_{i,j} = -\ii\cdot \parens[\Big]{\frac{\sigma_i^{1/2} - \sigma_j^{1/2}}{2 (\sigma_i^{1/2} + \sigma_j^{1/2})}} \bracks{(A^P)^\dagger A^P}_{i,j} $, so we have
\begin{align*}
    \calL^P(\sigma)_{i,j}
    &= \parens[\Bigg]{(-\ii)^2 \parens[\Big]{\frac{\sigma_i^{1/2} - \sigma_j^{1/2}}{2 (\sigma_i^{1/2} + \sigma_j^{1/2})}} (\sigma_j - \sigma_i) - \frac{1}{2}(\sigma_i^{1/2} - \sigma_j^{1/2})^2}\bracks[\Big]{ (A^P)^\dagger A^P }_{i,j}
    = 0
\end{align*}
Therefore, $\calL^P(\sigma) = 0$, which concludes the proof. 
\end{proof}

In fact, this Lindbladian satisfies a stronger condition, known as KMS detailed balance (see \cref{sec:detail} for more details).

The operator $G$ can be written as an average over real-time evolutions, which will be useful when arguing about its locality behavior.

\begin{lemma}[Fourier transform expression for the coherent term]\label{claim:fourier-transform-expression}
For the Lindbladian  in \cref{def:our-l} we have
\[
    G^P =  \int_{-\infty}^{\infty} e^{\ii H \omega} \parens[\Big]{\frac{1}{2\ii} [P \sigma^{1/2} P, \sigma^{-1/2}]} e^{-\ii H \omega} f(\omega) \diff \omega
\]
where $f(\omega) = 1/(\beta \cdot \cosh(\frac{2\pi}{\beta} \omega))$.
\end{lemma}
\begin{proof}
Using the definition of $G^P$ and that $A^P = \sigma^{1/4} P \sigma^{-1/4}$, we have
\begin{equation}
\begin{split}
\label{eqn:fourier-transform-Gp}
    G^P
    &= -\frac{\ii}{2} \cdot \parens{ (A^P)^\dagger A^P } \circ \braces*{\frac{\sigma_i^{1/2} - \sigma_j^{1/2}}{\sigma_i^{1/2} + \sigma_j^{1/2}}}_{ij} \\
    &= \parens[\Big]{ \frac{1}{2\ii} P \sigma^{1/2} P } \circ \braces*{ \frac{\sigma_i^{1/2} - \sigma_j^{1/2}}{\sigma_i^{1/4}\sigma_j^{1/4}(\sigma_i^{1/2} + \sigma_j^{1/2})}}_{ij} \\
    &= \parens[\Big]{\frac{1}{2\ii} [P \sigma^{1/2} P, \sigma^{-1/2}]} \circ \braces*{\frac{\sigma_i^{1/2} - \sigma_j^{1/2}}{\sigma_i^{1/4}\sigma_j^{1/4}(\sigma_i^{1/2} + \sigma_j^{1/2})(\sigma_j^{-1/2} - \sigma_i^{-1/2})}}_{ij} \\
    &= \parens[\Big]{\frac{1}{2\ii} [P \sigma^{1/2} P, \sigma^{-1/2}]} \circ \braces*{\frac{(\sigma_i/\sigma_j)^{1/4} - (\sigma_j/\sigma_i)^{1/4}}{(\sigma_i/\sigma_j)^{1/2} - (\sigma_j/\sigma_i)^{1/2}}}_{ij} \\
    &= \parens[\Big]{\frac{1}{2\ii} [P \sigma^{1/2} P, \sigma^{-1/2}]} \circ \braces*{ \frac{1}{(\sigma_i/\sigma_j)^{1/4} + (\sigma_j/\sigma_i)^{1/4}}}_{ij} \\
    &= \parens[\Big]{\frac{1}{2\ii} [P \sigma^{1/2} P, \sigma^{-1/2}]} \circ \braces*{ g(\lambda_i - \lambda_j) }_{ij} \text{ where }\sigma_i = e^{-\beta \lambda_i} \text{ and } g(x) = \frac{1}{2 \cosh(\beta x/4)}
\end{split}
\end{equation}
We can then compute the Fourier transform of $g$:
\begin{align*}
    g(x) = \frac{1}{\sqrt{2\pi}}\int_{-\infty}^\infty e^{\ii x \omega} \wh{g}(\omega) \diff\omega \quad\text{where}\quad \wh{g}(\omega) = \frac{1}{\sqrt{2\pi}}\int_{-\infty}^\infty e^{-\ii \omega x} g(x) \diff x.
\end{align*}
Now, we rely on the Fourier transform of hyperbolic secant, $\frac{1}{\sqrt{2\pi}} \int_{-\infty}^\infty \frac{e^{-\ii \omega x}}{\cosh(x)} \diff x = \sqrt{\frac{\pi}{2}} / \cosh(\frac{\pi}{2}\omega)$:
\begin{align*}
    \wh{g}(\omega) &= \frac{1}{\sqrt{2\pi}}\int_{-\infty}^\infty \frac{e^{-\ii \omega x}}{2\cosh(\beta x/4)} \diff x
    = \frac{2}{\beta\sqrt{2\pi}} \frac{\pi}{\cosh(\frac{2\pi}{\beta}\omega)}
    = \frac{\sqrt{2\pi}}{\beta \cosh\Paren{2\pi \omega/ \beta}}.
\end{align*}
Plugging the Fourier transform back into \eqref{eqn:fourier-transform-Gp}, we get the desired expression:
\begin{align*}
    G^P
    &= \frac{1}{\sqrt{2\pi}}\int_{-\infty}^\infty \wh{g}(\omega) \parens[\Big]{\frac{1}{2\ii} [P \sigma^{1/2} P, \sigma^{-1/2}]} \circ \braces*{ e^{\ii \omega (\lambda_i - \lambda_j)} }_{ij} \diff \omega \\
    &= \frac{1}{\sqrt{2\pi}}\int_{-\infty}^\infty \wh{g}(\omega) e^{\ii H \omega}\parens[\Big]{\frac{1}{2\ii} [P \sigma^{1/2} P, \sigma^{-1/2}]} e^{-\ii H  \omega} \diff \omega \\
    &= \int_{-\infty}^\infty \frac{1}{\beta \cosh(\frac{2\pi}{\beta}\omega)} e^{\ii H \omega}\parens[\Big]{\frac{1}{2\ii} [P \sigma^{1/2} P, \sigma^{-1/2}]} e^{-\ii H  \omega} \diff \omega
    \qedhere
\end{align*}
\end{proof}

Next, we show that the coefficients in the Fourier basis remain bounded, even when integrating against a polynomially growing function.

\begin{lemma}[Fourier integral on monomials] \label{lem:fourier-integral-on-monomial}
    With $f(\omega) = 1/(\beta \cdot \cosh(\frac{2\pi}{\beta} \omega))$ as in \cref{claim:fourier-transform-expression}, for any $r \in \Z_{\geq 0}$, we have
    \begin{align*}
        \int_{-\infty}^\infty f(\omega) \abs{\omega}^r \diff \omega
        \leq \frac{2}{\pi} \parens[\big]{\frac{\beta}{2\pi}}^r r! \, .
    \end{align*}
\end{lemma}
\begin{proof}
Observe,
\begin{align*}
    \int_{-\infty}^\infty f(\omega) \abs{\omega}^r \diff \omega
    = \frac{1}{\beta}\int_{-\infty}^\infty \frac{\abs{\omega}^r}{\cosh(\frac{2\pi}{\beta} \omega)} \diff \omega = \frac{2}{\beta}\int_{0}^\infty \frac{\omega^r}{\cosh(\frac{2\pi}{\beta} \omega)} \diff \omega
\end{align*}
Substituting $\phi = \frac{2\pi}{\beta} \omega$, we have
\begin{equation*}
\begin{split}
    \frac{2}{\beta}\int_{0}^\infty \frac{\omega^r}{\cosh(\frac{2\pi}{\beta} \omega)} \diff \omega = \frac{1}{\pi} \parens[\big]{\frac{\beta}{2\pi}}^r \int_{0}^\infty \frac{\phi^r}{\cosh(\phi)} \diff \phi  \leq \frac{2}{\pi} \parens[\big]{\frac{\beta}{2\pi}}^r \int_{0}^\infty \frac{\phi^r}{e^\phi} \diff \phi = \frac{2}{\pi} \parens[\big]{\frac{\beta}{2\pi}}^r r! \,,
\end{split}
\end{equation*}
where the inequality uses $\cosh{x} \geq e^{\abs{x}}/2$.
\end{proof}

\subsection{Quasi-locality}

We begin by showing the following modification to the Lieb--Robinson bound for an operator supported on a sub-system:

\begin{lemma}[{Lieb--Robinson bound, \cite[Lemma 5]{hhkl21}}] \label{lem:lieb-robinson}
    \ewin{checked}
    For an operator $A$ supported on $S$ and $\ell > 0$, recall that
    \begin{align*}
        H^{(\ball(S, \ell))} = \sum_{a = 1}^\terms  \iver{\supp(H_a) \subseteq \ball(S, \ell)} H_a.
    \end{align*}
    Then, for all $\omega \in \mathbb{R}$,
    \begin{align*}
        \opnorm[\Big]{
            e^{\ii H \omega} A e^{-\ii H \omega} - e^{\ii H^{(\ball(S, \ell))} \omega} A e^{\ii H^{(\ball(S, \ell))} \omega}
        } \leq \abs{S}  \cdot \opnorm{A} \cdot \frac{(2(\degree + 1)\locality\abs{\omega})^\ell}{\ell!}.
    \end{align*}
\end{lemma}
\begin{proof}
We invoke \cite[Lemma 5]{hhkl21} with $X \gets S$, $O_X \gets A$, and $\Omega \gets \ball(S, \ell)$, and can therefore conclude that 
\begin{multline*}
    \opnorm[\Big]{
            e^{\ii H \omega} A e^{-\ii H \omega} - e^{\ii H^{(\ball(S, \ell))} \omega} A e^{\ii H^{(\ball(S, \ell))} \omega}
        } \leq \abs{S}  \cdot \opnorm{A} \cdot \frac{(2\zeta_0 \cdot \abs{\omega})^\ell}{\ell!} \\
    \text{ where }
    \zeta_0 = \max_{i \in [\qubits]} \sum_{\substack{a \in [\terms] \\ \supp(H_a) \ni i}} \abs{\supp(H_a)} \cdot \opnorm{H_a}
\end{multline*}
The sum has at most $\degree + 1$ terms, and $\abs{\supp(H_a)} \leq \locality$, giving the bound $\zeta_0 \leq (\degree + 1)\locality$.
\end{proof}

Next, we show that the coherent term is quasi-local via cluster expansion:
\begin{lemma}[Coherent term is quasi-local] \label{lem:coherent-subterm-expansion}
    For some $P \in \jumps{i}$, let
    \begin{align*}
        F^P(\omega) = e^{\ii H \omega}\parens[\Big]{\frac{1}{2i} [P \sigma^{1/2} P , \sigma^{-1/2}]}e^{-\ii H \omega}.
    \end{align*}
    Then we can write it in the following way:
    \begin{align*}
        F^P(\omega) = \sum_{k > 0} \sum_{\cluster{a} \in [\terms]^k} \sum_{r \geq 0} F_{\cluster{a}, r}^{P,\omega}
    \end{align*}
    where $F_{\cluster{a}, r}^{P, \omega}$ is supported on $\ball(S_{\cluster{a}}, r)$, the ball of distance at most $r$ from $S_{\cluster{a}}$; $F_{\cluster{a}, r}^{P,\omega}$ is nonzero only when $\cluster{a}$ is a cluster from $i$;  $\opnorm{F_{\cluster{a}, r}^{P,\omega}} \leq \frac{2\beta^k}{k!}$ for all $r$; and for any $r > 0$,
    \begin{align*}
        \opnorm{F_{\cluster{a}, r}^{P,\omega}}
        &\leq \abs{S_{\cluster{a}}} \cdot  \frac{\beta^k}{k!} \cdot \parens[\Big]{
            \frac{(2(\degree + 1)\locality \abs{\omega})^{r-1}}{(r-1)!}
            + \frac{(2(\degree + 1)\locality \abs{\omega})^r}{r!}
        }
    \end{align*}
\end{lemma}
\begin{proof}
\ewin{2025-08-05 This lemma checked}
We begin by writing out 
\begin{align*}
    \frac{1}{2i} [P \sigma^{1/2} P , \sigma^{-1/2}]
    = \frac{1}{2i} (P e^{-\frac{\beta}{2}H} P e^{\frac{\beta}{2} H} - e^{\frac{\beta}{2}H} P e^{-\frac{\beta}{2} H} P )
\end{align*}
Notice that this is the skew-Hermitian part of $P e^{-\frac{\beta}{2} H} P e^{\frac{\beta}{2}H} $, which we denote $\Skew(M) = \frac{1}{2\ii}(M - M^\dagger)$.
Then, by \cref{lem:hadamard},
\begin{align*}
    \frac{1}{2i} [P \sigma^{1/2} P , \sigma^{-1/2}]
    = \sum_{k \geq 0} \sum_{\cluster{a} \in [\terms]^k} \underbrace{
        \Skew \parens[\Big]{\frac{(-\beta/2)^k}{k!} P [H_{a_k},[\dots[H_{a_1}, P]\dots]]}
    }_{F^{P,\omega}_{\cluster{a}}},
\end{align*}
where $F^{P,\omega}_{\cluster{a}}$ is only non-zero when $\cluster{a}$ is a cluster from $i$; in that case, $\opnorm{F^{P,\omega}_{\cluster{a}}} \leq \frac{\beta^k}{k!}$, and the support of $F^{P,\omega}_{\cluster{a}}$ is contained in $S_{\cluster{a}} = \supp(H_{a_1}) \cup \dots \cup \supp(H_{a_k})$.
Additionally, no $k = 0$ term appears, since the corresponding expression is $\Skew(I) = 0$.
Now, we can write
\begin{align*}
    F^P(\omega) &= \sum_{k > 0} \sum_{\cluster{a} \in [\terms]^k} e^{\ii H \omega} F^{P,\omega}_{\cluster{a}} e^{-\ii H \omega}.
\end{align*}
We further decompose the evolution into shells of increasing radius by writing
\begin{align*}
    e^{\ii H \omega} F^{P,\omega}_{\cluster{a}} e^{-\ii H \omega}
    = \underbrace{F_{\cluster{a}}^{P,\omega}}_{F_{\cluster{a}, 0}^{P,\omega}} + \sum_{r \geq 1} \underbrace{
        e^{\ii H^{(\ball(S_{\cluster{a}},r))} \omega} F^{P,\omega}_{\cluster{a}} e^{-\ii H^{(\ball(S_{\cluster{a}},r))} \omega}
        - e^{\ii H^{(\ball(S_{\cluster{a}},r-1))} \omega} F^{P,\omega}_{\cluster{a}} e^{-\ii H^{(\ball(S_{\cluster{a}},r-1))} \omega}
    }_{F^{P,\omega}_{\cluster{a}, r}}.
\end{align*}
Under these definition, $\supp(F^{P,\omega}_{\cluster{a}, r}) \subseteq \ball(S_{\cluster{a}}, r)$ for all $r \geq 0$.
Further, by \cref{lem:lieb-robinson},
\begin{align*}
    \opnorm{F^{P,\omega}_{\cluster{a}, r}}
    &\leq \opnorm{
        e^{\ii H^{(\ball(S_{\cluster{a}},r))} \omega} F^{P,\omega}_{\cluster{a}} e^{-\ii H^{(\ball(S_{\cluster{a}},r))} \omega}
        - e^{\ii H \omega} F^{P,\omega}_{\cluster{a}} e^{-\ii H \omega}
    } \\ & \qquad + \opnorm{
        e^{\ii H \omega} F^{P,\omega}_{\cluster{a}} e^{-\ii H \omega}
        - e^{\ii H^{(\ball(S_{\cluster{a}},r-1))} \omega} F^{P,\omega}_{\cluster{a}} e^{-\ii H^{(\ball(S_{\cluster{a}},r-1))} \omega}
    } \\
    &\leq \abs{S_{\cluster{a}}} \cdot \opnorm{F^{P,\omega}_{\cluster{a}}}\parens[\Big]{\frac{(2(\degree + 1)\locality \abs{\omega})^r}{r!} + 
        \frac{(2(\degree + 1)\locality \abs{\omega})^{r-1}}{(r-1)!} 
    } \\
    &\leq \abs{S_{\cluster{a}}} \cdot  \frac{\beta^k}{k!}\parens[\Big]{
        \frac{(2(\degree + 1)\locality \abs{\omega})^{r-1}}{(r-1)!}
        + \frac{(2(\degree + 1)\locality \abs{\omega})^r}{r!}
    }
\end{align*}
This bound holds when $r > 0$; for $r \geq 0$, the bound
\begin{align*}
    \opnorm{F^{P,\omega}_{\cluster{a}, r}}
    &\leq 2 \opnorm{F^{P,\omega}_{\cluster{a}}} \leq 2\frac{\beta^k}{k!}
\end{align*}
also holds by unitarity.
\end{proof}

\subsubsection{Truncating to a local Lindbladian}

In this subsection, we compare the quasi-local Lindbladian to a truly local one, obtained by truncating their operations to small neighborhoods. \aineshdefer{editing here.}

\begin{definition}[Truncated balanced Lindbladian]
    \label{def:our-l-truncated}
    For a site $j \in [\qubits]$, we define the Lindbladian $\calL_{j,\delta}^*$ to be the version of $\calL_{j}^*$ where we truncate the corresponding series to only affect a ball around $j$:
    \begin{align*}
        \calL_{j, \delta}^*(\rho)
        &= \sum_{P \in \jumps{j}} \calL_\delta^P(\rho) \\
        &= \sum_{P \in \jumps{j}} \parens[\Big]{-\ii [G_\delta^P, \rho] + A_\delta^P \rho (A_\delta^P)^\dagger - \frac12 \braces{(A_\delta^P)^\dagger A_\delta^P, \rho}}.
    \end{align*}
    where we define
    \begin{align*}
        A_\delta^P &\coloneqq \sum_{k = 0}^{K} \sum_{\cluster{a} \in [\terms]^k} A_{\cluster{a}}^P \\
        G_\delta^P &\coloneqq \int_{-\infty}^\infty \sum_{k > 0}^K \sum_{\cluster{a} \in [\terms]^k} \sum_{r \geq 0}^R F_{\cluster{a}, r}^{P,\omega} f(\omega) \diff \omega
    \end{align*}
    for $K, R = 1 + \ceil{\log(1000/\delta)/ \log(1/(3(\degree + 1)\locality^2 \beta))}$.
\end{definition}

As a direct consequence of truncating these series, the Lindbladian becomes supported on a small ball around $j$.

\begin{lemma}[Support size of the truncated Lindbladian]
    \label{lem:truncation-support}
    $\calL_{j, \delta}^*$ is supported on $\ball(j, K + R)$, where $K + R = \bigO{\log(1/\delta) / \log(1/(\degree \locality^2 \beta))}$.
\end{lemma}

Using our prior quasi-locality bounds, we can show that the error from truncation is small.

\begin{lemma}[Controlling error from Lindbladian truncation]
    \label{lem:truncation-error}
    For the truncated balanced Lindbladian as described in \cref{def:our-l-truncated}, provided that $\beta \leq 1/(10 (\degree + 1) \locality^2 \beta)$,
    \begin{align*}
        \dnorm{\calL_j^* - \calL_{j, \delta}^*} \leq \delta.
    \end{align*}
\end{lemma}
\begin{proof}
By properties of diamond norm~\cite{watrous18}, to bound the diamond norm of $\calL_j^* - \calL_{j, \delta}^*$, it suffices to bound the trace norm of $[(\calL_j^* - \calL_{j, \delta}^*) \otimes \calI](\rho)$ where $\calI$ is the identity channel on an $\qubits$-qubit space and $\rho$ is any state on $2\qubits$ qubits.
\begin{equation}
    \begin{split}
    \calL_j^*(\rho)
    &= \sum_{P \in \jumps{j}} - \ii [G^P, \rho]
        + A^P \rho (A^P)^\dagger - \frac12 \braces{(A^P)^\dagger A^P, \rho} \\
    \calL_{j, \delta}^*(\rho)
    &= \sum_{P \in \jumps{j}} -\ii [G_\delta^P, \rho]
        + A_\delta^P \rho (A_\delta^P)^\dagger - \frac12 \braces{(A_\delta^P)^\dagger A_\delta^P, \rho}
    \end{split}
\end{equation}
We can then compare their difference:
\begin{equation}
    \begin{split}
    \label{eq:truncate-total}
    & \trnorm{\calL_j^*(\rho) - \calL_{j, \delta}^*(\rho)} \\
    &\leq \sum_{P \in \jumps{j}} \trnorm{- \ii [G^P - G_\delta^P, \rho]}
        + \trnorm[\Big]{A^P \rho (A^P)^\dagger - A_\delta^P \rho (A_\delta^P)^\dagger} + \trnorm[\Big]{\frac12 \braces{(A^P)^\dagger A^P - (A_\delta^P)^\dagger A_\delta^P, \rho}} \\
    &\leq \sum_{P \in \jumps{j}} 2\opnorm{G^P - G_\delta^P}
        + 2\opnorm{A^P - A_\delta^P}(\opnorm{A^P} + \opnorm{A_\delta^P})
    \end{split}
\end{equation}
It remains to bound the operator norm quantities above.
Throughout, we will use that, for $\ell \geq 0$,
\begin{align*}
    \sum_{k > \ell} \sum_{\substack{\cluster{a} \in [\terms]^k \\ \cluster{a} \text{ cluster from } j}} \frac{\beta^k}{k!}
    \leq \sum_{k > \ell} ((e^2 (e \degree)^{k-1} k!))\frac{\beta^k}{k!}
    = e^2 \beta \sum_{k > \ell} (e \degree \beta)^{k-1}
    \leq 2e^2 \beta (e \degree \beta)^{\ell-1}
    \leq 2e (e \degree \beta)^{\ell},
\end{align*}
where we use \cref{lem:num-clusters} and that $e \degree \beta \leq 1/10$.
The same bound holds when $\frac{\beta^k}{k!}$ is replaced with $\frac{k \beta^k}{k!}$; the only change is that we need that $\sum_{k > \ell} k (e \degree \beta)^{k-1} \leq 2(\ell + 1)(e \degree \beta)^{\ell}$ given the same condition $e \degree \beta \leq 1/10$.
First, the operator norm of $A^P$:
\begin{equation}
    \label{eq:truncate-A-norm}
    \opnorm{A^P}
    \leq \sum_{k = 0}^\infty \sum_{\cluster{a} \in [\terms]^k} \opnorm{A_{\cluster{a}}^P}
    \leq 1 + \sum_{k = 1}^\infty \sum_{\cluster{a} \in [\terms]^k} \frac{\beta^k}{k!}
    \leq 1 + 2e
\end{equation}
Since $A_\delta^P$ is the same series, except truncated at $K$, $\opnorm{A^P} \leq 1 + 2e$
\begin{equation}
    \begin{split}
    \label{eq:truncate-A}
    \opnorm{A^P - A_\delta^P}
    &= \opnorm[\Big]{\sum_{k > K} \sum_{\cluster{a} \in [\terms]^k} A_{\cluster{a}}^P} \\
    &\leq \sum_{k > K} \sum_{\substack{\cluster{a} \in [\terms]^k \\ \cluster{a} \text{ cluster from } j}} \frac{\beta^k}{k!} \\
    &\leq 2e (e \degree \beta)^{K}
    \end{split}
\end{equation}
Now, to bound the coherent part,
\begin{equation}
    \begin{split}
    \label{eq:truncate-G}
    \opnorm{G^P - G_\delta^P}
    &= \opnorm[\Big]{\int_{-\infty}^\infty \parens[\Big]{
        \sum_{k > 0} \sum_{\cluster{a} \in [\terms]^k} \sum_{r \geq 0} F_{\cluster{a}, r}^{P, \omega}
        - \sum_{k > 0}^K \sum_{\cluster{a} \in [\terms]^k} \sum_{r = 0}^R F_{\cluster{a}, r}^{P, \omega}
    } f(\omega) \diff \omega} \\
    &= \opnorm[\Big]{\int_{-\infty}^\infty \parens[\Big]{
        \sum_{k > 0} \sum_{\cluster{a} \in [\terms]^k} \sum_{r \geq 0} F_{\cluster{a}, r}^{P, \omega} \iver{k > K \text{ or } r > R}
    } f(\omega) \diff \omega} \\
    &= \underbrace{
        \opnorm[\Big]{\int_{-\infty}^\infty \parens[\Big]{
        \sum_{k > 0} \sum_{\cluster{a} \in [\terms]^k} \sum_{r \geq 0} F_{\cluster{a}, r}^{P, \omega} \iver{k > K}
    } f(\omega) \diff \omega}}_{\text{\eqref{eq:truncate-G}.(1)}} \\
    & \qquad \qquad + \underbrace{\opnorm[\Big]{\int_{-\infty}^\infty \parens[\Big]{
        \sum_{k > 0} \sum_{\cluster{a} \in [\terms]^k} \sum_{r \geq 0} F_{\cluster{a}, r}^{P, \omega} \iver{r > R}
    } f(\omega) \diff \omega}}_{\text{\eqref{eq:truncate-G}.(2)}}
    \end{split}
\end{equation}
These two pieces we bound individually.
The first bound is simpler.
\begin{equation}
    \begin{split}
    \label{eq:truncate-G-1}
    \text{\eqref{eq:truncate-G}.(1)}
    &= \opnorm[\Big]{\int_{-\infty}^\infty \parens[\Big]{
        \sum_{k > K} \sum_{\cluster{a} \in [\terms]^k} e^{\ii H \omega} F_{\cluster{a}}^{P, \omega} e^{-\ii H \omega}
    } f(\omega) \diff \omega} \\
    &\leq \int_{-\infty}^\infty \parens[\Big]{
        \sum_{k > K} \sum_{\cluster{a} \in [\terms]^k} \opnorm{e^{\ii H \omega} F_{\cluster{a}}^{P, \omega} e^{-\ii H \omega}}
    } f(\omega) \diff \omega \\
    &\leq \sum_{k > K} \sum_{\cluster{a} \in [\terms]^k} \opnorm{F_{\cluster{a}}^{P, \omega}} \\
    &\leq \sum_{k > K} \sum_{\substack{\cluster{a} \in [\terms]^k \\ \cluster{a} \text{ cluster from } j}} \frac{\beta^k}{k!} \\
    &\leq 2e(e \degree \beta)^K
    \end{split}
\end{equation}
For the second bound, we pay the truncation error from the Lieb--Robinson series.
\begin{equation}
    \begin{split}
    \label{eq:truncate-G-2}
    \text{\eqref{eq:truncate-G}.(2)}
    &\leq \sum_{k > 0} \sum_{\cluster{a} \in [\terms]^k} \sum_{r \geq R} \int_{-\infty}^\infty \opnorm{F_{\cluster{a}, r}^{P, \omega}
    } f(\omega) \diff \omega \\
    &\leq \sum_{k > 0} \sum_{\substack{\cluster{a} \in [\terms]^k \\ \cluster{a} \text{ cluster from } j}} \sum_{r \geq R} \int_{-\infty}^\infty \abs{S_{\cluster{a}}} \cdot  \frac{\beta^k}{k!} \cdot \parens[\Big]{
        \frac{(2(\degree + 1)\locality \abs{\omega})^{r-1}}{(r-1)!}
        + \frac{(2(\degree + 1)\locality \abs{\omega})^r}{r!}
    } f(\omega) \diff \omega \\
    &\leq \sum_{k > 0} \sum_{\substack{\cluster{a} \in [\terms]^k \\ \cluster{a} \text{ cluster from } j}} \sum_{r \geq R} \abs{S_{\cluster{a}}} \cdot  \frac{\beta^k}{k!} \cdot \frac{2}{\pi} \parens[\Big]{
        (2(\degree + 1)\locality \beta)^{r-1}
        + (2(\degree + 1)\locality \beta)^r
    } \\
    &\leq \sum_{k > 0} \sum_{\substack{\cluster{a} \in [\terms]^k \\ \cluster{a} \text{ cluster from } j}} \abs{S_{\cluster{a}}} \cdot  \frac{\beta^k}{k!} \cdot 2 (2(\degree + 1)\locality \beta)^{R-1} \\
    &\leq 2 (2(\degree + 1)\locality \beta)^{R-1} \sum_{k > 0} \sum_{\substack{\cluster{a} \in [\terms]^k \\ \cluster{a} \text{ cluster from } j}} k \locality \frac{\beta^k}{k!} \\
    &\leq 4e\locality (2(\degree + 1)\locality \beta)^{R-1}
    \end{split}
\end{equation}
Returning to \eqref{eq:truncate-total}, we combine all these bounds---\eqref{eq:truncate-A-norm}, \eqref{eq:truncate-A}, \eqref{eq:truncate-G}, \eqref{eq:truncate-G-1}, and \eqref{eq:truncate-G-2}---to get the final bound
\begin{align*}
    &\trnorm{\calL_j^*(\rho) - \calL_{j, \delta}^*(\rho)} \\
    &\leq \sum_{P \in \jumps{j}} 2\opnorm{G^P - G_\delta^P}
        + 2\opnorm{A^P - A_\delta^P}(\opnorm{A^P} + \opnorm{A_\delta^P}) \\
    &\leq \sum_{P \in \jumps{j}} 2(2e(e\degree \beta)^K + 4e\locality (2(\degree + 1)\locality \beta)^{R - 1}) + 2(2e (e \degree \beta)^K)(2(1 + 2e)) \\
    &\leq 500\parens[\Big]{(e\degree \beta)^K + (2(\degree + 1)\locality^2 \beta)^{R - 1}} \\
    &\leq 1000 (3(\degree + 1)\locality^2 \beta)^{\ceil{\log(1000/\delta) / \log(1/(3(\degree + 1)\locality^2 \beta))}} \leq \delta\,,
\end{align*}
as stated.
\end{proof}

\begin{lemma}[Controlling evolution error with Lindbladian error]
    \label{lem:truncation-evolution-error}
    For two Lindbladians $\calL, \calL'$,
    \begin{align*}
        \dnorm{e^{\calL t} - e^{\calL' t}} \leq t \dnorm{\calL - \calL'}
    \end{align*}
\end{lemma}
\begin{proof}
We can write
\begin{align*}
    e^{\calL t} - e^{\calL' t}
    &= \lim_{\delta \to 0} \parens[\Big]{(1 + \delta \calL t)^{1/\delta} - (1 + \delta \calL' t)^{1/\delta}} \\
    &= \lim_{\delta \to 0} \sum_{s = 0}^{1/\delta - 1} \parens[\Big]{(1 + \delta \calL t)^{1/\delta - s} (1 + \delta \calL' t)^{s} - (1 + \delta \calL t)^{1/\delta - s - 1} (1 + \delta \calL' t)^{s + 1}} \\
    \dnorm{e^{\calL t} - e^{\calL' t}}
    &\leq \lim_{\delta \to 0} \sum_{s = 0}^{1/\delta - 1} \dnorm[\Big]{(1 + \delta \calL t)^{1/\delta - s} (1 + \delta \calL' t)^{s} - (1 + \delta \calL t)^{1/\delta - s - 1} (1 + \delta \calL' t)^{s + 1}} \\
    &\leq \lim_{\delta \to 0} \sum_{s = 0}^{1/\delta - 1} \dnorm{(1 + \delta \calL t)}^{1/\delta - s - 1} \dnorm{\delta (\calL - \calL') t} \dnorm{1 + \delta \calL' t}^s \\
    &\leq \lim_{\delta \to 0} \sum_{s = 0}^{1/\delta - 1} (1 + \bigOs{\delta^2})^{1/\delta - s - 1} \dnorm{\delta (\calL - \calL') t} (1 + \bigOs{\delta^2})^{s} \\
    &= \dnorm{(\calL - \calL') t}
\end{align*}
\end{proof}

\subsection{Evolution of Wasserstein transport plans}
\ewin{This subsubsection checked by me.}

In this subsection, we prove the following result, demonstrating how transport plans update under evolution by the Lindbladian.
Recall the definition of an update matrix from \cref{subsec:wass} and the Hamiltonian-related notation $S_{\cluster{a}}$ and $\ball(S_{\cluster{a}}, r)$ from \cref{subsec:ham}.

\begin{theorem}[Wasserstein norm of the Lindblad evolution at a site]
    \label{thm:wasserstein-growth-from-jump}
    For $i \in [\qubits]$ and $\beta < 1/(2(\degree + 1)\locality)$, consider the map
    \begin{align*}
        \Phi_i(\rho) = \rho + \delta  \calL_i^*(\rho),
    \end{align*}
    where $\calL_i^*(\rho)$ is as specified in \cref{def:our-l}.
    Then $\Phi_i$ has an update matrix of the form $\id + \delta Q^{(i)} + \bigOs{\delta^2}$, where
    \begin{multline*}
        Q^{(i)} \coloneqq -4 E_{\braces{i}} + \wh{Q}^{(i)}, \\
        \text{ where } \wh{Q}^{(i)} \coloneqq \sum_{k > 0} \sum_{\substack{\cluster{a} \in [\terms]^k \\ \cluster{a} \text{ cluster from } i}} \sum_{r > 0} \mu_{\cluster{a}, r} E_{S_{\cluster{a}, r}}
        \text{ and } \mu_{\cluster{a}, r} \coloneqq 22\locality \frac{\beta^k}{(k-1)!} ((\degree + 1)\locality \beta)^{r-1}
    \end{multline*}
    Here, we use the notation $S_{\cluster{a}, r} = \ball(S_{\cluster{a}}, r)$.
\end{theorem}

\begin{proof}
Let $X$ be a traceless Hermitian matrix with an associated transport plan $\parens{X_j}_{j \in [\qubits]}$ and cost vector $x \in \R_{\geq 0}^{\qubits}$.
Then our goal is to show that $\Phi_i(X)$ has a transport plan $\braces{Y_k}_k$ whose cost vector is entry-wise bounded by $(\id + \delta Q^{(i)} + \bigOs{\delta^2}) x$.
First, we split our expression into coherent and dissipative parts.
\begin{equation}
\label{eqn:diss-coh-split-single-site}
\begin{split}
    \Phi_i(X)
    &= X + \delta \sum_{P \in \jumps{i}} \calL^P(X) \nonumber \\
    &= \underbrace{ \frac{1}{2}\parens[\Big]{X - 2\delta \sum_{P \in \jumps{i}} \ii [G^P, X]} }_{\eqref{eqn:diss-coh-split-single-site}.(1) } \\
    &\hspace{4em} + \underbrace{ \frac{1}{2}\parens[\Big]{X + 2\delta \sum_{P \in \jumps{i}} \parens[\Big]{A^P X (A^P)^\dagger - \frac12\braces{(A^P)^\dagger A^P, X}}} }_{\eqref{eqn:diss-coh-split-single-site}.(2)}
\end{split}
\end{equation}

\ewindefer{Right now, we call these claims before we state them.} \aineshdefer{i can move them around at some point}
By applying \cref{prop:wasserstein-coherent-term}, the coherent part~\eqref{eqn:diss-coh-split-single-site}.(1) has a transport plan with cost at most $(\id + 2\delta Q_{\cohe}^{(i)} + \bigOs{\delta^2}) x$.
Similarly, by applying \cref{prop:wasserstein-dissipative-term}, the dissipative part~\eqref{eqn:diss-coh-split-single-site}.(2) has a transport plan with cost at most $(\id + 2\delta Q_{\diss}^{(i)} + \bigOs{\delta^2}) x$.
Thus, $\Phi_i(X)$ has a transport plan bounded by the average of these two bounds, $(\id + \delta (Q_{\cohe}^{(i)} + Q_{\diss}^{(i)}) + \bigOs{\delta^2}) x$.
Since
\begin{align*}
    & (Q_{\cohe}^{(i)} + Q_{\diss}^{(i)})x \\
    &\vleq \parens[\Big]{-4 E_{\braces{i}} + \sum_{k > 0} \sum_{\substack{\cluster{a} \in [\terms]^k \\ \cluster{a} \text{ cluster from } i}} \parens[\Big]{12 \frac{\beta^k}{k!} E_{S_{\cluster{a}}} + 10 \locality \frac{\beta^k}{(k-1)!} \sum_{r \geq 0} ((\degree + 1)\locality \beta)^{r} E_{S_{\cluster{a},r+1}}}}x \\
    &\vleq \parens[\Big]{-4 E_{\braces{i}} + 22 \locality \sum_{k > 0} \sum_{\substack{\cluster{a} \in [\terms]^k \\ \cluster{a} \text{ cluster from } i}} \frac{\beta^k}{(k-1)!} \sum_{r \geq 0} ((\degree + 1)\locality \beta)^{r} E_{S_{\cluster{a},r+1}}} x
\end{align*}
Notice that we get $-4E_{\braces{i}}$ from the dissipative term; and up to a constant factor, the coherent part dominates the cost.
\end{proof}

\subsubsection{Bounding Wasserstein growth of the coherent term}
\label{subsection:wasserstein-coherent-term}
\ewindefer{This section checked}

In this subsection, we prove the following proposition, which describes how a transportation plan updates under the action of the coherent part of our Lindbladian.

\begin{claim}[Wasserstein norm of the coherent evolution]
\label{prop:wasserstein-coherent-term}
    Let $X$ be a traceless Hermitian matrix with an associated transport plan $\parens{X_j}_{j \in [\qubits]}$ and cost vector $x$.
    Further, for $i \in [\qubits]$ and $\beta \leq 1/((\degree + 1)\locality)$, let
    \begin{align*}
        \Phi_{i,\cohe}(\rho) = \rho - \ii \delta \sum_{P \in \jumps{i}} [G^P, \rho],
    \end{align*}
    where $G^P$ is as specified in \cref{def:our-l}.
    Then $\Phi_{i,\cohe}(X)$ has a transport plan $\braces{Y_k}_k$ whose cost vector is entry-wise bounded by $(\id + \delta Q^{(i)}_{\cohe} + \bigOs{\delta^2}) x$, where
    \begin{align*}
        Q^{(i)}_{\cohe}
        = 30 \locality \sum_{k > 0} \sum_{\substack{\cluster{a} \in [\terms]^k \\ \cluster{a} \text{ cluster from } i}} \frac{\beta^k}{(k-1)!} \sum_{r \geq 0} ((\degree + 1)\locality \beta)^{r} E_{S_{\cluster{a},r+1}}.
    \end{align*}
\end{claim}

\ewin{This follows as a corollary of the following lemma.}

\begin{proposition}[Wasserstein norm of the coherent evolution]
\label{thm:bounded-norm-coherent-term-pauli}
    Let $X$ be a traceless Hermitian matrix with an associated transport plan $\parens{X_j}_{j \in [\qubits]}$ and cost vector $x$.
    Further, let $G^P$ be the matrix as specified in \cref{def:our-l}, where $\supp(P) = \braces{i}$.
    Then for any $\beta < 1/(\locality(\degree+1))$, the matrix $X - \ii \delta [G^P, X]$ has a transport plan whose cost vector is entry-wise bounded by $(\id + \delta Q^P_{\cohe} + \bigOs{\delta^2}) x$, where
    \begin{align*}
        Q^P_{\cohe}
        = 10 \locality \sum_{k > 0} \sum_{\substack{\cluster{a} \in [\terms]^k \\ \cluster{a} \text{ cluster from } i}} \frac{\beta^k}{(k-1)!} \sum_{r \geq 0} ((\degree + 1)\locality \beta)^{r} E_{S_{\cluster{a},r+1}}.
    \end{align*}
\end{proposition}
\begin{proof}
By linearity (\cref{fact:linear-plans}), we can rescale so that $\wnorm{X} = \frac12 \trnorm{X} = 1$.
Further, it suffices to prove the statement for ``edges'', i.e.\ matrices $X$ such that, for some $j \in [\qubits]$, $\tr_j(X) = 0$, and so which have an associated cost vector of $e_j$.

Using the Fourier transform expression for $G^P$ from \cref{claim:fourier-transform-expression}, we have
\begin{equation}
\label{eqn:plugging-fourier-expansion-into-[G,X]}
\begin{split}
    X - \ii \delta [G^P, X]
    &= X - \ii \delta \int_{-\infty}^\infty f(\omega) [F^P(\omega), X] \diff \omega \\
    &= X - \ii \delta\sum_{k > 0} \sum_{\cluster{a} \in [\terms]^k} \sum_{r \geq 0} \int_{-\infty}^\infty f(\omega) [F_{\cluster{a}, r}^{P,\omega}, X] \diff \omega \\
    &= \underbrace{ \parens[\Big]{ 
        X - \ii \delta \sum_{k > 0} \sum_{\cluster{a} \in [\terms]^k} \sum_{r \geq 0} \iver{j \notin \supp(F_{\cluster{a}, r})} \int_{-\infty}^\infty f(\omega) [F_{\cluster{a}, r}^{P,\omega}, X] \diff \omega
    } }_{\eqref{eqn:plugging-fourier-expansion-into-[G,X]}.(1)} \\
    & \qquad - \underbrace{ \parens[\Big]{
        \ii \delta \sum_{k > 0} \sum_{\cluster{a} \in [\terms]^k} \sum_{r \geq 0} \iver{j \in \supp(F_{\cluster{a}, r})} \int_{-\infty}^\infty f(\omega) [F_{\cluster{a}, r}^{P,\omega}, X] \diff \omega
    }}_{\eqref{eqn:plugging-fourier-expansion-into-[G,X]}.(2)}
\end{split}
\end{equation}
Focusing on the first term, observe that $\tr_j($\eqref{eqn:plugging-fourier-expansion-into-[G,X]}.(1)$) = 0$, since $\tr_j(X) = 0$ and $\tr_j([Y, X]) = 0$ provided $Y$ is not supported on $j$.
\ewin{stopped here}
So, by \cref{lem:w1-to-trace}, \eqref{eqn:plugging-fourier-expansion-into-[G,X]}.(1) has a transport plan with cost vector bounded by $\frac{1}{2}\trnorm{\eqref{eqn:plugging-fourier-expansion-into-[G,X]}.(1)} \cdot  e_j$.
To bound this trace norm, notice that by linearity of the commutator,
\begin{align*}
    \eqref{eqn:plugging-fourier-expansion-into-[G,X]}.(1)
    &= X - \ii \delta \bracks[\Big]{\sum_{k > 0} \sum_{\cluster{a} \in [\terms]^k} \sum_{r \geq 0} \iver{j \notin \supp(F_{\cluster{a}, r})} \int_{-\infty}^\infty f(\omega) F_{\cluster{a}, r}^{P,\omega} \diff \omega, X}.
\end{align*}
Therefore, by \cref{cor:trace-norm-lindblad-step}, this trace norm is at most $(1 + \bigOs{\delta^2})\trnorm{X}$ and the transport plan has cost vector bounded by $\Paren{ 1 + \bigOs{\delta^2} } e_j$.
\ewin{don't we need convergence?}

For term \eqref{eqn:plugging-fourier-expansion-into-[G,X]}.(2), we know that the support of $[F_{\cluster{a}, r}^{P,\omega}, X]$ is contained in a ball of radius $r$ around $S_{\cluster{a}}$; we denote this set $S_{\cluster{a}, r} \coloneqq \ball(S_{\cluster{a}}, r)$.
Then by \cref{lem:w1-commutator} the transport plan for \eqref{eqn:plugging-fourier-expansion-into-[G,X]}.(2) has a cost vector bounded by
\begin{align*}
    &  \delta \sum_{k > 0} \sum_{\substack{\cluster{a} \in [\terms]^k \\ \cluster{a} \text{ cluster from } i}} \sum_{r \geq 0} \iver{j \in S_{\cluster{a}, r}} \parens[\Big]{\int_{-\infty}^\infty f(\omega) \trnorm{[F_{\cluster{a}, r}^P(\omega), X]} \diff \omega} e_{S_{\cluster{a},r}} \\
    &= \delta \parens[\Bigg]{\sum_{k > 0} \sum_{\substack{\cluster{a} \in [\terms]^k \\ \cluster{a} \text{ cluster from } i}} \sum_{r \geq 0} \parens[\Big]{\int_{-\infty}^\infty f(\omega) \trnorm{[F_{\cluster{a}, r}^P(\omega), X]} \diff \omega} E_{S_{\cluster{a},r}}} e_j \\
    &\vleq \delta \parens[\Bigg]{\sum_{k > 0} \sum_{\substack{\cluster{a} \in [\terms]^k \\ \cluster{a} \text{ cluster from } i}} \frac{\beta^k}{k!} \parens[\Big]{6 E_{S_{\cluster{a}}} + 4 \sum_{r > 0} \abs{S_{\cluster{a}}} ((\degree + 1)\locality \beta)^{r-1} E_{S_{\cluster{a},r}}}} e_j \\
    &\vleq \delta \parens[\Bigg]{\underbrace{10 \locality \sum_{k > 0} \sum_{\substack{\cluster{a} \in [\terms]^k \\ \cluster{a} \text{ cluster from } i}} \frac{\beta^k}{(k-1)!} \sum_{r \geq 0} ((\degree + 1)\locality \beta)^{r} E_{S_{\cluster{a},r+1}}}_{Q_\cohe^P}} e_j
\end{align*}
The last inequality uses that $\abs{S_{\cluster{a}}} \leq k\locality$.
The first inequality is derived by bounding the operator norm of $F_{\cluster{a},r}^{P,\omega}$ using \cref{lem:coherent-subterm-expansion} and bounding the integrals using \cref{lem:fourier-integral-on-monomial}.
For $r = 0$, this is done in the following way.
\begin{align*}
    \int_{-\infty}^\infty f(\omega) \trnorm{[F_{\cluster{a}, 0}^{P,\omega}, X]} \diff \omega
    \leq 4\int_{-\infty}^\infty f(\omega) \opnorm{F_{\cluster{a}, 0}^{P,\omega}} \diff \omega
    \leq \frac{8\beta^k}{k!}\int_{-\infty}^\infty f(\omega) \diff \omega
    \leq 6\frac{\beta^k}{k!}
\end{align*}
For $r > 0$, we do the folllowing.
\begin{align*}
    &\quad \int_{-\infty}^\infty f(\omega) \trnorm{[F_{\cluster{a}, r}^{P,\omega}, X]} \diff \omega \\
    &\leq 4\int_{-\infty}^\infty f(\omega) \opnorm{F_{\cluster{a}, r}^P(\omega)} \diff \omega  \\
    &\leq 4 \abs{S_{\cluster{a}}} \frac{\beta^k}{k!} \int_{-\infty}^\infty f(\omega) \parens[\Big]{
            \frac{(2(\degree + 1)\locality \abs{\omega})^{r-1}}{(r-1)!}
            + \frac{(2(\degree + 1)\locality \abs{\omega})^r}{r!}
    } \diff \omega \\
    &= 4\abs{S_{\cluster{a}}} \frac{\beta^k}{k!} \parens[\Big]{
            \frac{(2(\degree + 1)\locality )^{r-1}}{(r-1)!}
            \int_{-\infty}^\infty f(\omega) \abs{\omega}^{r-1} \diff \omega
            + \frac{(2(\degree + 1)\locality)^r}{r!}
            \int_{-\infty}^\infty f(\omega) \abs{\omega}^{r} \diff \omega
    } \\
    &\leq \frac{8}{\pi} \abs{S_{\cluster{a}}} \frac{\beta^k}{k!} \parens[\Big]{ 
            ((\degree + 1)\locality \beta/\pi)^{r-1}
            + ((\degree + 1)\locality \beta/\pi)^{r}
    } \\
    &\leq 4\abs{S_{\cluster{a}}} \frac{\beta^k}{k!} ((\degree + 1)\locality \beta)^{r-1}
\end{align*}
Here, we also use that $\beta < 1/((\degree + 1)\locality)$.

Combining the transport plans for \eqref{eqn:plugging-fourier-expansion-into-[G,X]}.(1) with \eqref{eqn:plugging-fourier-expansion-into-[G,X]}.(2), we get a transport plan for $X - \ii \delta [G^P, X]$ with cost vector bounded by $(1 + \bigOs{\delta^2}) e_j + \delta Q_\cohe^P e_j$, as desired.
\end{proof}
 

    

\subsubsection{Bounding Wasserstein growth of the dissipative term}
\ewin{I went over this section}

Our goal for this section is to prove \cref{prop:wasserstein-dissipative-term}, which describes how the dissipative term of the Lindblad jump on site $i$ affects transport plans.

\begin{claim}[Wasserstein norm of the dissipative evolution]
\label{prop:wasserstein-dissipative-term}
    Let $X$ be a traceless Hermitian matrix with an associated transport plan $\parens{X_j}_{j \in [\qubits]}$ and cost vector $x$.
    Further, for $i \in [\qubits]$ and $\beta < 1/(4(\degree + 1))$, let
    \begin{align*}
        \Phi_{i,\diss}(\rho) = \rho + \delta \sum_{P \in \jumps{i}} \parens[\Big]{A^P \rho (A^P)^\dagger - \frac12\braces{(A^P)^\dagger A^P, \rho}},
    \end{align*}
    where $A^P$ is as specified in \cref{def:our-l}.
    Then $\Phi_{i,\diss}(X)$ has a transport plan $\braces{Y_k}_k$ whose cost vector is entry-wise bounded by $(\id + \delta Q^{(i)}_{\diss} + \bigOs{\delta^2}) x$, where
    \begin{align*}
        Q^{(i)}_{\diss}
        = - 4 E_{\braces{i}} + 12 \sum_{t > 0} \frac{\beta^t}{t!} \sum_{\substack{\cluster{z} \in [\terms]^t \\ \cluster{z} \text{ cluster from } i}} E_{S_{\cluster{z}}} .
    \end{align*}
\end{claim}


\begin{lemma}[Wasserstein norm of a few-site map]
    \label{lem:wasserstein-local-map}
    For matrices $K$ and $L$, consider a map $\Psi$ of the form
    \begin{align*}
        \Psi(\rho) = K \rho L^\dagger + L \rho K^\dagger
            - \frac12 \braces{L^\dagger K + K^\dagger L, \rho} \,.
    \end{align*}
    Further suppose that $S = \supp(K) \cup \supp(L)$.
    Then this map is trace-preserving.
    Moreover, if $X$ is a traceless Hermitian matrix, then $\Psi(X)$ has a transport plan $\braces{Y_k}_k$ whose cost vector is entry-wise bounded by $4 \opnorm{K} \opnorm{L} \trnorm{X} e_S$.
\end{lemma}
\begin{proof}
By inspection, $\tr(\Psi(X)) = \tr(X)$.
Further, for a traceless Hermitian $X$, we note that
\begin{align*}
    \tr_{S}\parens[\Big]{
        K X L^\dagger + L X K^\dagger - \frac12 \braces{L^\dagger K + K^\dagger L, X}
    } = 0,
\end{align*}
using cyclicity of partial trace (\cref{fact:cyclicity}).
So, we can use \cref{lem:w1-few-site} to conclude that this expression has a transport plan whose cost vector is bounded by
\begin{align*}
    \trnorm[\Big]{
        K X L^\dagger + L X K^\dagger - \frac12 \braces{L^\dagger K + K^\dagger L, X}
    } e_S
    \vleq 4 \opnorm{K} \opnorm{L} \trnorm{X} e_S \, ,
\end{align*}
where we use the triangle inequality and the sub-multiplicativity of the trace norm.
\end{proof}

We are now ready to complete the proof of \cref{prop:wasserstein-dissipative-term}.

\begin{proof}[Proof of \cref{prop:wasserstein-dissipative-term}]
It suffices to prove the statement for normalized ``edges'' $X$, meaning that for some $j \in [\qubits]$, $\tr_j(X) = 0$, and $\wnorm{X} = \frac12\trnorm{X} = 1$.
If we can construct transport plans as specified for such $X$'s, then we can do so for all $X$'s by linearity (\cref{fact:linear-plans}).
Recall, 
\begin{equation}
\label{eqn:dissipative-term-single-site}
    \Phi_{i,\diss}(X) = X + \delta \sum_{P \in \jumps{i}} \parens[\Big]{A^P X (A^P)^\dagger - \frac12\braces{(A^P)^\dagger A^P, X}}
\end{equation}
By \cref{lem:hadamard}, we can write our jump operators $A^P$ as
\begin{align*}
    A^P = e^{-\frac{\beta}{4} H} P e^{\frac{\beta}{4} H}
    = \sum_{\cluster{a}} A_{\cluster{a}}^P,
\end{align*}
where $\cluster{a}$ is a vector of terms of any length, $\cluster{a} \in \bigcup_{k \geq 0} [\terms]^k$, and $A_{\cluster{a}}^P$ is a nested commutator, $A_{\cluster{a}}^P = \frac{(\beta/4)^k}{k!} [H_{a_k},[\dots[H_{a_1}, P]\dots]]$.
More specifically, $A_{\cluster{a}}^P$ is only non-zero when $\cluster{a}$ is a cluster starting from $i$, and then, $\opnorm{A_{\cluster{a}}^P} \leq (\beta/2)^k/k!$, where $k$ is the size of the cluster.
We will need throughout that this series converges, meaning that $\sum_{\cluster{a}} \opnorm{A_{\cluster{a}}^P}$ is bounded by a constant; this follows from our bound on $\beta$, by \cref{lem:had-bounds}.
\ewin{Make sure that the above discussion also occurs in the coherent stuff.}

Now, we break our argument into two cases: $i = j$ and $i \neq j$.

\paragraph{The $i = j$ case.}
This case is where the ``contraction'' appears.
We can decompose $\Phi_{i, \diss}$ in the following way.
\begin{equation} \label{eq:diss-i=j}
\begin{split}
    \Phi_{i, \diss}(X)
    &= X + \delta \sum_{P \in \jumps{i}} \sum_{\cluster{a}, \cluster{b}} \parens[\Big]{A_{\cluster{a}}^P X (A_{\cluster{b}}^P)^\dagger - \frac12\braces{(A_{\cluster{b}}^P)^\dagger A_{\cluster{a}}^P, X}} \\
    &= \underbrace{
        X + \delta \sum_{P \in \jumps{i}} (PXP - X)
    }_{\text{\eqref{eq:diss-i=j}.(1)}} \\
    &\qquad
    + \underbrace{
        \delta \sum_{P \in \jumps{i}} \sum_{\substack{\cluster{a},\cluster{b} \\ (\cluster{a}, \cluster{b}) \neq (\varnothing, \varnothing)}}
        \parens[\Big]{A_{\cluster{a}}^P X (A_{\cluster{b}}^P)^\dagger - \frac12\braces{(A_{\cluster{b}}^P)^\dagger A_{\cluster{a}}^P, X}}
    }_{\text{\eqref{eq:diss-i=j}.(2)}}.
\end{split}
\end{equation}
The key observation is that \eqref{eq:diss-i=j}.(1) is a depolarizing channel, and therefore, it contracts $X$ when $\tr_i(X) = 0$.
Concretely, for an $M$ with $\tr_i(M) = 0$,
\begin{align*}
    M + \sigma_X^{(i)} M \sigma_X^{(i)} + \sigma_Y^{(i)} M \sigma_Y^{(i)} + \sigma_Z^{(i)} M \sigma_Z^{(i)} = 0.
\end{align*}
So, $\text{\eqref{eq:diss-i=j}.(1)} = (1 - 4\delta)X$.
Because $\tr_i(\text{\eqref{eq:diss-i=j}.(1)}) = 0$, by \cref{lem:w1-to-trace}, it has a transport plan with cost
\begin{equation} \label{eq:diss-i=j-pt1}
    \frac12\trnorm{\text{\eqref{eq:diss-i=j}.(1)}} e_i
    = \frac12(1 - 4\delta)\trnorm{X} e_i
    = (1 - 4\delta) e_i.
\end{equation}
Now, we consider \eqref{eq:diss-i=j}.(2).
We first symmetrize the elements of the sum, to get that
\begin{align*}
    \text{\eqref{eq:diss-i=j}.(2)}
    &= \sum_{P \in \jumps{i}} \sum_{\substack{\cluster{a},\cluster{b} \\ (\cluster{a}, \cluster{b}) \neq (\varnothing, \varnothing)}}
    \frac{\delta}{2} \parens[\Big]{
        A_{\cluster{a}}^P X (A_{\cluster{b}}^P)^\dagger
        + A_{\cluster{b}}^P X (A_{\cluster{a}}^P)^\dagger
        - \frac12\braces{(A_{\cluster{b}}^P)^\dagger A_{\cluster{a}}^P + (A_{\cluster{a}}^P)^\dagger A_{\cluster{b}}^P, X}
    }.
\end{align*}
By \cref{lem:wasserstein-local-map}, the summand corresponding to $\cluster{a}, \cluster{b}$ has a transport plan with cost vector at most $\frac{\delta}{2}(4 \opnorm{A_{\cluster{a}}^P} \opnorm{A_{\cluster{b}}^P} \trnorm{X} e_{S_{\cluster{a}} \cup S_{\cluster{b}}}) = 4\delta\opnorm{A_{\cluster{a}}^P} \opnorm{A_{\cluster{b}}^P} e_{S_{\cluster{a}} \cup S_{\cluster{b}}}$.
Consequently, \eqref{eq:diss-i=j}.(2) has a transport plan with cost vector at most
\begin{equation} \label{eq:diss-i=j-pt2}
\begin{split}
    & 4 \delta \sum_{P \in \jumps{i}} \sum_{\substack{\cluster{a},\cluster{b} \\ (\cluster{a}, \cluster{b}) \neq (\varnothing, \varnothing)}}
        \opnorm{A_{\cluster{a}}^P} \opnorm{A_{\cluster{b}}^P} e_{S_{\cluster{a}} \cup S_{\cluster{b}}} \\
    &= 4 \delta \sum_{P \in \jumps{i}} \sum_{\substack{\cluster{a},\cluster{b} \\ i \in S_{\cluster{a}} \cup S_{\cluster{b}}}}
        \opnorm{A_{\cluster{a}}^P} \opnorm{A_{\cluster{b}}^P} e_{S_{\cluster{a}} \cup S_{\cluster{b}}}
\end{split}
\end{equation}
The equality above holds because $A_{\cluster{a}}^P$ is only non-zero when $\cluster{a}$ is empty or contains $i$.
Altogether, combining \eqref{eq:diss-i=j-pt1} with \eqref{eq:diss-i=j-pt2}, we have that for $i = j$, $\Phi_{i,\diss}(X)$ has a transport plan with cost vector at most
\begin{align*}
    (1 - 4\delta) e_i
    + 4 \delta \sum_{P \in \jumps{i}} \sum_{\substack{\cluster{a},\cluster{b} \\ i \in S_{\cluster{a}} \cup S_{\cluster{b}}}}
        \opnorm{A_{\cluster{a}}^P} \opnorm{A_{\cluster{b}}^P} e_{S_{\cluster{a}} \cup S_{\cluster{b}}}.
\end{align*}

\paragraph{The $i \neq j$ case.}
In this case, we want to control how a jump on site $i$ affects an edge on site $j$: this growth should be exponentially small in the distance between $i$ and $j$.
We formalize this by first decomposing $A^P$ into the part supported on site $j$ and the rest.
\begin{align*}
    A^P = \underbrace{
        \sum_{\cluster{a}} \iver{j \not\in S_{\cluster{a}}} A_{\cluster{a}}^P
    }_{A^{P,\wh{j}}} + \underbrace{
        \sum_{\cluster{a}} \iver{j \in S_{\cluster{a}}} A_{\cluster{a}}^P
    }_{A^{P,j}}
\end{align*}
This gives a corresponding decomposition of $\Phi_{i, \diss}$:
\begin{multline} \label{eq:diss-ineqj}
    \Phi_{i,\diss}(X)
    = \underbrace{ X + \delta \sum_{P \in \jumps{i}} \parens[\Big]{A^{P,\wh{j}} X (A^{P,\wh{j}})^\dagger - \frac12\braces{(A^{P,\wh{j}})^\dagger A^{P,\wh{j}}, X}} }_{\eqref{eq:diss-ineqj}.(1)} \\
    + \delta \sum_{P \in \jumps{i}} \Big( A^{P,j} X (A^{P,\wh{j}})^\dagger + A^{P,\wh{j}} X (A^{P,j})^\dagger + A^{P,j} X (A^{P,j})^\dagger \\
    - \frac12\braces{(A^{P,\wh{j}})^\dagger A^{P,j} + (A^{P,j})^\dagger A^{P,\wh{j}} + (A^{P,j})^\dagger A^{P,j}, X} \Big)
\end{multline}
\ewindefer{Figure out better notation}
Let \eqref{eq:diss-ineqj}.(2) be the rest of the expression, $\eqref{eq:diss-ineqj} - \text{\eqref{eq:diss-ineqj}.(1)}$. 

To construct a transport plan for \eqref{eq:diss-ineqj}.(1), first write
\begin{align*}
    \text{\eqref{eq:diss-ineqj}.(1)}
    = \frac13 \sum_{P \in \jumps{i}} \parens[\Big]{ X + 3 \delta \sum_{P \in \jumps{i}} \parens[\Big]{A^{P,\wh{j}} X (A^{P,\wh{j}})^\dagger - \frac12\braces{(A^{P,\wh{j}})^\dagger A^{P,\wh{j}}, X}} }.
\end{align*}
Then, for one summand of the above sum, $\tr_j(\text{\eqref{eq:diss-ineqj}.(1)}) = 0$, since the support of $A^{P, \wh{j}}$ does not contain $j$.
So, by \cref{lem:w1-to-trace}, it has a transport plan with cost vector $\frac12 \trnorm{\text{\eqref{eq:diss-ineqj}.(1)}}e_j$, and by \cref{lem:dnorm-lindblad-bound}, this trace norm is at most $(1 + \frac12\opnorm{3 \delta (A^{P,\wh{j}})^\dagger A^{P,\wh{j}}}^2)\trnorm{X}$.
Combining these together, we have that \eqref{eq:diss-ineqj}.(1) has a transport plan with cost vector
\begin{equation} \label{eq:diss-ineqj-pt1}
\begin{split}
    & \frac13 \sum_{P \in \jumps{i}} \frac12(1 + \frac12\opnorm{3 \delta (A^{P,\wh{j}})^\dagger A^{P,\wh{j}}}^2)\trnorm{X} e_j \\
    &= \frac13 \sum_{P \in \jumps{i}} (1 + \frac12\opnorm{3 \delta (A^{P,\wh{j}})^\dagger A^{P,\wh{j}}}^2) e_j \\
    & \vleq (1 + \bigOs{\delta^2}) e_j.
\end{split}
\end{equation}
In the final bound, we are implicitly using that $\opnorm{A^{P, \wh{j}}}$ is finite.

As for the second part, we decompose the components into clusters, symmetrizing the elements of the sum as we did in the $i = j$ case.
\begin{align*}
    \text{\eqref{eq:diss-ineqj}.(2)}
    &= \sum_{P \in \jumps{i}} \sum_{\substack{\cluster{a}, \cluster{b} \\ j \in S_{\cluster{a}} \cup S_{\cluster{b}}}}
    \frac{\delta}{2} \parens[\Big]{
        A_{\cluster{a}}^P X (A_{\cluster{b}}^P)^\dagger
        + A_{\cluster{b}}^P X (A_{\cluster{a}}^P)^\dagger
        - \frac12\braces{(A_{\cluster{b}}^P)^\dagger A_{\cluster{a}}^P + (A_{\cluster{a}}^P)^\dagger A_{\cluster{b}}^P, X}
    }.
\end{align*}
By \cref{lem:wasserstein-local-map}, the $\cluster{a}, \cluster{b}$ part of the sum has a transport plan with cost at most $$\frac{\delta}{2}(4 \opnorm{A_{\cluster{a}}^P} \opnorm{A_{\cluster{b}}^P} \trnorm{X} e_{S_{\cluster{a}} \cup S_{\cluster{b}}}) = 4\delta \opnorm{A_{\cluster{a}}^P} \opnorm{A_{\cluster{b}}^P} e_{S_{\cluster{a}} \cup S_{\cluster{b}}}.$$
Across the full sum, \eqref{eq:diss-ineqj}.(2) has a transport plan with cost at most
\begin{equation} \label{eq:diss-ineqj-pt2}
\begin{split}
    4 \delta \sum_{P \in \jumps{i}} \sum_{\substack{\cluster{a},\cluster{b} \\ j \in S_{\cluster{a}} \cup S_{\cluster{b}}}}
        \opnorm{A_{\cluster{a}}^P} \opnorm{A_{\cluster{b}}^P} e_{S_{\cluster{a}} \cup S_{\cluster{b}}}
\end{split}
\end{equation}
Altogether, combining \eqref{eq:diss-ineqj-pt1} with \eqref{eq:diss-ineqj-pt2}, we have that for $i \neq j$, $\Phi_{i,\diss}(X)$ has a transport plan with cost vector at most
\begin{align*}
    (1 + \bigOs{\delta^2}) e_j
    + 4 \delta \sum_{P \in \jumps{i}} \sum_{\substack{\cluster{a},\cluster{b} \\ j \in S_{\cluster{a}} \cup S_{\cluster{b}}}}
        \opnorm{A_{\cluster{a}}^P} \opnorm{A_{\cluster{b}}^P} e_{S_{\cluster{a}} \cup S_{\cluster{b}}}.
\end{align*}

\paragraph{Combining the two cases.}
What we have shown is that, for general $j$, $\Phi_{i, \diss}(X)$ has a transport plan with cost vector at most
\begin{align*}
    & (1 - 4\delta \iver{i = j} + \bigOs{\delta^2}) e_j
    + 4 \delta \sum_{P \in \jumps{i}} \sum_{\substack{\cluster{a},\cluster{b} \\ j \in S_{\cluster{a}} \cup S_{\cluster{b}}}}
        \opnorm{A_{\cluster{a}}^P} \opnorm{A_{\cluster{b}}^P} e_{S_{\cluster{a}} \cup S_{\cluster{b}}} \\
    &= \parens[\Bigg]{
    (1 + \bigOs{\delta^2}) I
    - 4 \delta E_{\braces{i}}
    + 4 \delta \sum_{P \in \jumps{i}} \sum_{\substack{\cluster{a},\cluster{b}}}
        \opnorm{A_{\cluster{a}}^P} \opnorm{A_{\cluster{b}}^P} E_{S_{\cluster{a}} \cup S_{\cluster{b}}}
    } e_j.
\intertext{This is now in the form we want: we have a transport plan for $\Phi(X)$ whose cost vector is linear in the cost vector of $X$.
We now perform some bounds to simplify the expression of this inner matrix.
So far, we have used very little about the jump operators: we only used that they can be decomposed into operators of different-sized supports.
Now, we will use that these jump operators are in fact quasi-local.
In particular, by \cref{lem:hadamard}, we know that $\opnorm{A_{\cluster{a}}^P} \leq (\beta/2)^k/k!$, where $k$ is the length of the vector $\cluster{a}$.
Further, $A_{\cluster{a}}^P = 0$ unless $\cluster{a}$ is a cluster starting from $i$, the site $P$ is supported on.
So, we can continue bounding:}
    &\vleq \parens[\Bigg]{
    (1 + \bigOs{\delta^2}) I
    - 4 \delta E_{\braces{i}}
    + 4 \delta \sum_{P \in \jumps{i}} \sum_{k,\ell \geq 0} \sum_{\substack{\cluster{a} \in [\terms]^k,\cluster{b} \in [\terms]^\ell \\ \cluster{a}, \cluster{b} \text{ clusters from } i }} \frac{(\beta/2)^{k+\ell}}{k! \ell!} \cdot E_{S_{\cluster{a}} \cup S_{\cluster{b}}}
    } e_j \\
    &\vleq \parens[\Bigg]{
    (1 + \bigOs{\delta^2}) I
    - 4 \delta E_{\braces{i}}
    + 12 \delta \sum_{k,\ell \geq 0} \sum_{\substack{\cluster{a} \in [\terms]^k,\cluster{b} \in [\terms]^\ell \\ \cluster{a}, \cluster{b} \text{ clusters from } i }} \frac{(\beta/2)^{k+\ell}}{k! \ell!} \cdot E_{S_{\cluster{a}} \cup S_{\cluster{b}}}
    } e_j \\
    \intertext{
        Now, we can re-parametrize, and treat $\cluster{a}$ concatenated with $\cluster{b}$ as a length $k + \ell$ vector, $\cluster{z} \in [\terms]^{k + \ell}$.
        When $\cluster{a}$ and $\cluster{b}$ are clusters from $i$, so is $\cluster{z}$.
        Further, $S_{\cluster{a}} \cup S_{\cluster{b}} = S_{\cluster{z}}$.
        So, we have
    }
    &\vleq \parens[\Bigg]{
    (1 + \bigOs{\delta^2}) I
    - 4 \delta E_{\braces{i}}
    + 12 \delta \sum_{t \geq 0} \sum_{\substack{\cluster{z} \in [\terms]^t \\ \cluster{z} \text{ cluster from } i}} \parens[\Big]{\sum_{s = 0}^t \frac{(\beta/2)^{t}}{s! (t-s)!}}E_{S_{\cluster{z}}}
    } e_j \\
    &= \parens[\Bigg]{
    (1 + \bigOs{\delta^2}) I
    - 4 \delta E_{\braces{i}}
    + 12 \delta \sum_{t \geq 0} \frac{\beta^t}{t!} \sum_{\substack{\cluster{z} \in [\terms]^t \\ \cluster{z} \text{ cluster from } i}} E_{S_{\cluster{z}}}
    } e_j \qedhere
\end{align*}
\end{proof}

\subsubsection{Lemmas about the Wasserstein evolution}
\ewindefer{TODO clean this up}

The $Q^{(i)}$ matrix which appears in \cref{thm:wasserstein-growth-from-jump} is crucial.
Since this matrix describes the behavior of our Lindbladian $\calL_i^*$ in the Wasserstein metric, controlling the Lindbladian reduces to controlling this matrix.
\ewindefer{Explain with a little more detail, maybe.}
We now prove the key facts about this matrix which we will use for our downstream applications.

\begin{lemma}[The $Q^{(i)}$'s act on different parts of the space] \label{lem:Q-contracts}
    Let $Q^{(i)} = -4E_i + \wh{Q}^{(i)}$ be as defined in \cref{thm:wasserstein-growth-from-jump}, with $\beta \leq 1/(10000 \locality^3 \growth^2 \degree)$.
    Then 
    \begin{align*}
        \norm[\Big]{\sum_{i=1}^\qubits \wh{Q}^{(i)}}_{1 \to 1}
        \leq \max_{j \in [\qubits]} \sum_{i=1}^\qubits \norm{\wh{Q}^{(i)} e_j}_{1}
        \leq 1.
    \end{align*}
\end{lemma}
\begin{proof}
Recall that
\begin{align*}
    \wh{Q}^{(i)} \coloneqq \sum_{k > 0} \sum_{\substack{\cluster{a} \in [\terms]^k \\ \cluster{a} \text{ cluster from } i}} \sum_{r > 0} \mu_{\cluster{a}, r} E_{S_{\cluster{a}, r}}
        \text{, where } \mu_{\cluster{a}, r} \coloneqq 22\locality \frac{\beta^k}{(k-1)!} ((\degree + 1)\locality \beta)^{r-1}.
\end{align*}
So,
\begin{align*}
    \norm{\wh{Q}^{(i)} e_j}_1
    &\leq \sum_{k > 0} \sum_{\substack{\cluster{a} \in [\terms]^k \\ \cluster{a} \text{ cluster from } i}} \sum_{r > 0} \mu_{\cluster{a}, r} \norm{E_{S_{\cluster{a}, r}} e_j}_1 \\
    &= \sum_{k > 0} \sum_{\substack{\cluster{a} \in [\terms]^k \\ \cluster{a} \text{ cluster from } i}} \sum_{r \geq \max(1, \dist(S_{\cluster{a}}, j))} \mu_{\cluster{a}, r} \abs{S_{\cluster{a}, r}}
\end{align*}
Summing this over all $i$, we have that
\begin{equation} \label{eq:Q-contracts-sum}
\begin{split}
    \sum_{i = 1}^{\qubits} \norm{\wh{Q}^{(i)} e_j}_1
    & \leq \sum_{i = 1}^\qubits \sum_{k > 0} \sum_{\substack{\cluster{a} \in [\terms]^k \\ \cluster{a} \text{ cluster from } i}} \sum_{r \geq \max(1, \dist(S_{\cluster{a}}, j))} \mu_{\cluster{a}, r} \abs{S_{\cluster{a}, r}} \\
    &\leq \locality \sum_{k > 0} \sum_{\substack{\cluster{a} \in [\terms]^k \\ \cluster{a} \text{ cluster}}} \sum_{r \geq \max(1, \dist(S_{\cluster{a}}, j))} \mu_{\cluster{a}, r} \abs{S_{\cluster{a}, r}} \\
    &\leq \locality \sum_{k > 0} \sum_{\substack{\cluster{a} \in [\terms]^k \\ \cluster{a} \text{ cluster}}} \sum_{r \geq \max(1, \dist(S_{\cluster{a}}, j))} 22\locality \frac{\beta^k}{(k-1)!} ((\degree + 1)\locality \beta)^{r-1} k \locality \growth^r \\
    &= 22 \locality^3 \growth \sum_{k > 0} \frac{\beta^kk}{(k-1)!} \sum_{\substack{\cluster{a} \in [\terms]^k \\ \cluster{a} \text{ cluster}}} \sum_{r \geq \max(1, \dist(S_{\cluster{a}}, j))} (\growth(\degree + 1)\locality\beta)^{r-1} \\
    &\leq 44 \locality^3 \growth \sum_{k > 0} \frac{\beta^kk}{(k-1)!} \sum_{\substack{\cluster{a} \in [\terms]^k \\ \cluster{a} \text{ cluster}}} (\growth(\degree + 1)\locality\beta)^{\max(0,\dist(S_{\cluster{a}}, j) - 1)} \\
    &\leq 1000 \locality^3 \growth^2 \sum_{k > 0} \beta^kk^2 (e \degree)^{k-1} \\
    &\leq 2000 \locality^3 \growth^2 \beta
    \end{split}
\end{equation}
Note that $\mu_{\cluster{a}, r}$ does not depend on $i$.
The second inequality uses that every cluster $\cluster{a} = (a_1,\dots,a_k)$ is a ``cluster from'' every site in $a_1$, of which there are at most $\locality$ possible options.
The third inequality uses the definition of $\mu_{\cluster{a}, r}$ and the growth parameter \cref{def:ham-growth-parameter}.
The fourth inequality uses that $\beta \leq 1/(2\growth(\degree + 1)\locality)$.
The fifth calls \cref{lem:num-clusters-quasi-2}, using that $\beta \leq 1/(2\growth^2(\degree + 1)\locality)$.
The sixth uses that $\sum_{k > 0} c^{k-1} k^2 \leq 2$ for $c < 1/10$; here, we take $c = \beta e \degree$.
The final bound is at most 1 provided that $\beta \leq 1/(10000 \locality^3 \growth^2)$.
\end{proof}

\begin{corollary}[$Q^{(i)}$ has bounded norm] \label{cor:Q-1-1}
    Let $Q^{(i)} = -4E_i + \wh{Q}^{(i)}$ be as defined in \cref{thm:wasserstein-growth-from-jump}, with $\beta \leq 1/(10000 \locality^3 \growth^2 \degree)$.
    Then 
    \begin{align*}
        \norm{Q^{(i)}}_{1 \to 1}
        \leq \norm{-4E_i}_{1 \to 1} + \norm{\wh{Q}^{(i)}}_{1 \to 1}
        \leq 5.
    \end{align*}
\end{corollary}

The following bound is used for the CMI proof.
\begin{lemma}[$Q^{(i)}$ is quasi-local] \label{lem:Q-quasilocal}
    Let $x, y \in \R^{\qubits}$ be vectors supported on sets $X, Y \subseteq [\qubits]$, respectively, and satisfying $\norm{x}_1 = \norm{y}_1 = 1$.
    Then, for $\wh{Q}^{(i)}$ as defined in \cref{thm:wasserstein-growth-from-jump},
    \begin{align*}
        y^\dagger \wh{Q}^{(i)} x \leq 150 \locality (25 \growth \locality (\degree + 1)\beta)^{\max(1, \max(\dist(i, X), \dist(i, Y)) - 1)},
    \end{align*}
    provided $\beta \leq 1 / (100 \growth \locality (\degree + 1))$.
    \ewin{TODO prove this consequence.}
    In particular, when $\beta \leq \gamma^2/(c \growth \locality^2 (\degree + 1))$ for $c$ a sufficiently large constant,
    \begin{align*}
        y^\dagger \wh{Q}^{(i)} x \leq \gamma^{2 \max(\dist(i, X), \dist(i, Y))}
        \leq \gamma^{\dist(i, X) + \dist(i, Y)}.
    \end{align*}
\end{lemma}
\begin{proof}
In this proof, let $\Delta \coloneqq \max(\dist(i, X), \dist(i, Y))$.
Using the definition of $\wh{Q}^{(i)}$ from \cref{thm:wasserstein-growth-from-jump}, we write out the expression for $y^\dagger \wh{Q}^{(i)} x$ and then bound it.
\begin{equation} \label{eq:Q-quasilocal-long}
\begin{split}
    y^\dagger \wh{Q}^{(i)} x
    &= \sum_{k > 0} \sum_{\substack{\cluster{a} \in [\terms]^k \\ \cluster{a} \text{ cluster from } i}} \sum_{r > 0} \mu_{\cluster{a}, r} y^\dagger E_{S_{\cluster{a}, r}} x \\
    &\leq \sum_{k > 0} \sum_{\substack{\cluster{a} \in [\terms]^k \\ \cluster{a} \text{ cluster from } i}} \sum_{r > 0} \mu_{\cluster{a}, r} \abs{S_{\cluster{a}, r}} \iver{r \geq \dist(S_{\cluster{a}}, X)} \iver{r \geq \dist(S_{\cluster{a}}, Y)} \\
    &= \sum_{k > 0} \sum_{\substack{\cluster{a} \in [\terms]^k \\ \cluster{a} \text{ cluster from } i}} \sum_{r \geq \max(1, \dist(S_{\cluster{a}}, X), \dist(S_{\cluster{a}}, Y))}\mu_{\cluster{a}, r} \abs{S_{\cluster{a}, r}}  \\
    &\leq \sum_{k > 0} \sum_{\substack{\cluster{a} \in [\terms]^k \\ \cluster{a} \text{ cluster from } i}} \sum_{r \geq \max(1, \dist(S_{\cluster{a}}, X), \dist(S_{\cluster{a}}, Y))}22 \locality \frac{\beta^k}{(k-1)!} ((\degree + 1)\locality \beta)^{r-1} k \locality \growth^r  \\
    &= 22\locality^2 \growth \sum_{k > 0} \frac{\beta^k k}{(k-1)!} \sum_{\substack{\cluster{a} \in [\terms]^k \\ \cluster{a} \text{ cluster from } i}} \sum_{r \geq \max(1, \dist(S_{\cluster{a}}, X), \dist(S_{\cluster{a}}, Y))} (\growth (\degree + 1)\locality \beta)^{r-1} \\
    &\leq 44\locality^2 \growth \sum_{k > 0} \frac{\beta^k k}{(k-1)!} \sum_{\substack{\cluster{a} \in [\terms]^k \\ \cluster{a} \text{ cluster from } i}} (\growth (\degree + 1)\locality \beta)^{\max(1, \dist(S_{\cluster{a}}, X), \dist(S_{\cluster{a}}, Y)) - 1} \\
    &\leq 44\locality^2 \growth \sum_{k > 0} \frac{\beta^k k}{(k-1)!} \sum_{\substack{\cluster{a} \in [\terms]^k \\ \cluster{a} \text{ cluster from } i}} (\growth (\degree + 1)\locality \beta)^{\max(1, \Delta - k) - 1} \\
    &\leq 44\locality^2 \growth \sum_{k > 0} ((\degree + 1)\beta)^k k^2 (\growth (\degree + 1)\locality \beta)^{\max(1, \Delta - k) - 1}
\end{split}
\end{equation}
Above, we use \cref{lem:hadamard}.
\ewindefer{TODO fill in argument}

We consider two cases.
When $\Delta = 0$, $1$, or $2$, we only ever see the first expression in the max.
\begin{align*}
    \sum_{k > 0} ((\degree + 1)\beta)^k k^2 (\growth (\degree + 1)\locality \beta)^{\max(1, \Delta - k) - 1}
    = \sum_{k > 0} ((\degree + 1)\beta)^k k^2
    \leq 2(\degree + 1)\beta
\end{align*}
The inequality uses that $\beta(\degree + 1) \leq 1/10$.
When $\Delta > 2$, we see both parts of the max.
\begin{align*}
    & \sum_{k > 0} ((\degree + 1)\beta)^k k^2 (\growth (\degree + 1)\locality \beta)^{\max(1, \Delta - k) - 1} \\
    &= ((\degree + 1)\beta)^{\Delta - 1} \sum_{k=1}^{\Delta - 2} k^2(\growth \locality)^{\Delta - k - 1} + \sum_{k \geq \Delta - 1} ((\degree + 1)\beta)^k k^2 \\
    &\leq (\Delta - 1)^3((\degree + 1)\beta)^{\Delta - 1}(\growth \locality)^{\Delta - 2} + 2(\Delta - 1)^2((\degree + 1)\beta)^{\Delta - 1} \\
    &= (\Delta - 1)^3((\degree + 1)\beta)^{\Delta - 1}\parens{(\growth \locality)^{\Delta - 2} + 2}
\end{align*}
Plugging this into the expression from \cref{eq:Q-quasilocal-long}, we continue bounding.
When $\Delta = 0$, $1$, or $2$, we have
\begin{equation}
\begin{split}
    y^\dagger \wh{Q}^{(i)} x
    &\leq 88\locality (\growth \locality (\degree + 1)\beta).
\end{split}
\end{equation}
When $\Delta > 2$, we have
\begin{equation}
\begin{split}
    y^\dagger \wh{Q}^{(i)} x
    &\leq 44\locality^2 \growth (\Delta - 1)^3((\degree + 1)\beta)^{\Delta - 1}\parens{(\growth \locality)^{\Delta - 2} + 2} \\
    &\leq 150\locality (\Delta - 1)^3(\growth \locality (\degree + 1)\beta)^{\Delta - 1} \\
    &\leq 150\locality (25\growth \locality (\degree + 1)\beta)^{\Delta - 1}
\end{split}
\end{equation}
This gives the desired bound.
\end{proof}


\section{Rapid mixing at high temperature}

In this section, we introduce a quantum generalization of the well-studied Dobrushin condition~\cite{dobrushin1970prescribing,dobrushin1987completely} and show that it implies fast mixing for the balanced Lindbladian defined in~\cref{def:our-l}.
The argument proceeds by defining a quantum analog of the Dobrushin influence matrix and then executing a path coupling argument.

We proceed in two steps: first we introduce a quantum analog of the well-studied \emph{Dobrushin condition}, and show that whenever this condition holds, the Lindbladian dynamics mixes rapidly.
\ewin{This description sems slightly imprecise..}
In the quantum setting, we have to modify the definition of the Dobrushin influence matrix, by replacing total variation distance by a quantum analog of Wasserstein-$1$ and provide a \emph{path coupling} argument using this metric.
Second, we show that below the stated critical inverse temperature, the \emph{quantum Dobrushin condition} is satisfied.

\subsection{Quantum Dobrushin condition and rapid mixing}
\ewin{Checked}

We begin by defining our quantum generalization of the Dobrushin influence matrix.
We use quantum Wasserstein distance as our metric over quantum states~\cite{dmtl21}.
\ewin{Again, I think this is slightly imprecise.}

\begin{definition}[Quantum Dobrushin condition]
\label{def:quantum-dobrushin}
    Let $\Phi$ be a trace-preserving linear map which admits a decomposition into trace-preserving linear maps $\Phi = \Phi_1 + \dots + \Phi_\qubits$.
    \ewin{$\sigma \to \varrho$}
    We define the associated influence matrix $D^{(\Phi)}$ to have entries
    \begin{align*}
        D^{(\Phi)}_{i,j} = \max_{\rho,\, \sigma \,j\text{-neighbors}} \wnorm[\Big]{\Phi_i(\rho - \sigma)},
    \end{align*}
    where the maximization is over $\rho$ and $\sigma$ which are density matrices of quantum states and $j$-neighbors, meaning that that $\tr_j(\rho - \sigma) = 0$.
    Namely, when $\rho$ and $\sigma$ are neighbors, then $\wnorm{\rho - \sigma} \leq 1$.
    Given some $0<\gamma<1$, we then define the quantum Dobrushin condition to be 
    \begin{equation*}
        \norm{D^{(\Phi)}}_{1\to 1}
        = \max_{x : \norm{x}_1 = 1} \norm{D^{(\Phi)} x }_1
        \leq 1 - \frac{\gamma}{\qubits}\,.
    \end{equation*}
\end{definition}

In this definition, we think of the $\Phi_i$'s in the decomposition $\Phi = \Phi_1 + \dots + \Phi_\qubits$'s as being the part of $\Phi$ which acts on site $i$. In contrast to the classical case where the Dobrushin influence matrix maximizes an update to configurations $\rho$ and $\sigma$ over worst-case couplings in total-variation distance, our quantum analog works with Wasserstein distance instead.

Given the aforementioned condition, we show that the dynamics obtained by repeatedly applying the map $\Phi$ mixes rapidly.

\begin{lemma}[Dobrushin condition implies rapid mixing]
\label{thm:dobrushin-implies-mixing}
    For a trace-preserving map $\Phi$ which satisfies the quantum Dobrushin condition (\cref{def:quantum-dobrushin}) with parameter $\gamma$, there is a unique fixed point $\sigma$ with $\tr(\sigma) = 1$, meaning that $\Phi(\sigma) = \sigma$.
    Further, for any initial matrix $\rho$ with $\tr(\rho) = 1$, and for any error parameter $\eps > 0$, we have that
    \begin{align*}
        \trnorm{\Phi^{\tau}(\rho) - \sigma} \leq \eps \trnorm{\rho - \sigma}\, ,
    \end{align*}
    \ewindefer{I actually don't think that we have enough assumptions here to conclude that $\trnorm{\sigma}$ is bounded.. I think you need that $\Phi$ is completely positive, which is not necessarily the case. But it's fine; we handle this later.}
    provided $\tau \geq \frac{\qubits}{\gamma}\log\frac{\qubits}{\eps}$.
\end{lemma}

We caution the reader that, in this section, $\Phi$ is not necessarily a quantum channel, so it need not send PSD matrices to PSD matrices.
Instead, we consider $\Phi$ as a map operating on all trace 1 Hermitian matrices; it is in this space where we have a unique fixed point.
We will eventually take a limit such that $\Phi$ converges to a completely positive channel, which will ensures that this fixed point is in fact a valid quantum state.
\ewindefer{There might need to be some more busy work here; we need that $\Phi$ converging implies its behavior converges, which probably just follows from some kind of compactness argument.}

\ewin{optimal gate complexity}

To prove \cref{thm:dobrushin-implies-mixing}, we begin by observing that the entries of the Dobrushin influence matrix can be reformulated as follows:
\begin{fact}[Defining the Dobrushin influence matrix in terms of Wasserstein growth]
\label{fact:influence-in-w1-arb-x}
    For a linear map on quantum states $\Phi$, the Dobrushin influence matrix $D^{(\Phi)}$ as defined in \cref{def:quantum-dobrushin} can be equivalently defined as
    \begin{align*}
        D^{(\Phi)}_{i,j}
        = \max_{\substack{X : \tr_j(X) = 0 \\ X \neq 0}} \frac{\wnorm{\Phi_i(X)}}{\wnorm{X}}
        = \max_{\substack{X : \tr_j(X) = 0 \\ \wnorm{X} = 1}} \wnorm{\Phi_i(X)}
    \end{align*}
\end{fact}
\begin{proof}
The second equality is obvious.
To prove the first equality, in one direction,
\begin{align}
    \max_{\rho,\, \sigma \,j\text{-neighbors}} \wnorm[\Big]{\Phi_i(\rho - \sigma)}
    \leq \max_{\substack{X : \tr_j(X) = 0 \\ X \neq 0}} \frac{\wnorm{\Phi_i(X)}}{\wnorm{X}},
\end{align}
because $\wnorm{\rho - \sigma} \leq 1$ and the right-hand side is a maximization over a larger domain than the left-hand side.
In the other direction, we note that we can decompose any matrix $X = X_+ - X_-$ where $X_+$ and $X_-$ are both PSD and $\tr(X_+) = \tr(X_-) = \frac12\trnorm{X}$, by taking $X_+$ to be the projection of $X$ onto its eigenvectors with positive eigenvalue, and $X_-$ to be the rest of it.
Because $\tr_j(X) = 0$, $\wnorm{X} = \frac12 \trnorm{X}$ by \cref{lem:w1-to-trace}.
So, we have that
\begin{equation}
\label{eqn:substituting-x-j-states}
    \frac{\wnorm{\Phi_i(X)}}{\wnorm{X}}
    = \wnorm[\Big]{\Phi_i\parens[\Big]{ \frac{ X_{+}}{\tr(X_+)} - \frac{ X_{-}}{\tr\parens{X_{-}} } }} \leq \max_{\rho,\, \sigma \,j\text{-neighbors}} \wnorm[\Big]{\Phi_i(\rho - \sigma)},
\end{equation}
where the last line uses that $X_+/\tr(X_+)$ and $X_-/\tr(X_-)$ are density matrices which are $j$-neighbors.
%
\end{proof}

Next, we prove our key lemma: if the Dobrushin condition holds, then $\Phi$ contracts the space of quantum states in the quantum Wasserstein metric.

\begin{lemma}[Relating contraction to influence]
\label{lem:relating-contraction-to-influence}
    Let $\rho$ and $\varrho$ be trace 1 Hermitian matrices.
    For a trace-preserving map $\Phi$ with corresponding Dobrushin influence matrix $D^{\Phi}$ and some integer $\tau \geq 0$,
    \begin{align*}
        \wnorm{\Phi^\tau(\rho) - \Phi^\tau(\varrho)} \leq \norm{D^{(\Phi)}}_{1 \to 1}^\tau \wnorm{\rho - \varrho} \,.
    \end{align*}
\end{lemma}
\begin{proof}
Let $X$ be an arbitrary traceless Hermitian matrix.
We will first show that $\wnorm{\Phi(X)} \leq \norm{D^{(\Phi)}}_{1 \to 1} \wnorm{X}$.
Let $\parens{ X^{(j)} }_{j\in [\qubits]}$ be an optimal transport plan for $X$ and let $x \in \R^\qubits$ be its cost vector; that is, let $\sum_{j=1}^\qubits X_{(j)} = X$, with $\tr_j(X_{(j)}) = 0$ and $\wnorm{X} = \sum_{j \in [\qubits]} x_j$ where $x_j = \frac12 \trnorm{X_{(j)}}$. 
Then, using that the channel admits a decomposition of the form $\Phi = \sum_{i \in [n]} \Phi_i$ and applying triangle inequality, we have
\begin{equation}
    \begin{split}
        \wnorm{\Phi(X)}
        = \wnorm[\Big]{\sum_{j=1}^n \Phi(X_{(j)})}  
        & \leq \sum_{j=1}^n \sum_{i=1}^n \wnorm{\Phi_i(X_{(j)})} \\
        & \leq \sum_{j=1}^n \sum_{i=1}^n D^{(\Phi)}_{i,j} \cdot \wnorm{X_{(j)}}
        = \norm{D^{(\Phi)} x}_1 \\
        & \leq \norm{D^{(\Phi)}}_{1\to 1} \norm{x}_1
        = \norm{D^{(\Phi)}}_{1\to 1} \wnorm{X}
    \end{split}
\end{equation}
where the second inequality follows from \cref{fact:influence-in-w1-arb-x}.

We can iterate this inequality.
Since $\Phi^\tau(X)$ is always Hermitian and traceless,
\begin{align}
    \wnorm{\Phi^\tau(X)} \leq \norm{D^{(\Phi)}}_{1\to 1} \wnorm{\Phi^{\tau - 1}(X)}.
\end{align}
Iterating this gives that $\wnorm{\Phi^\tau(X)} \leq \norm{D^{(\Phi)}}_{1\to 1}^\tau \wnorm{X}$.
The statement follows by taking $X = \rho - \varrho$.
\end{proof}

We can now complete the proof of \cref{thm:dobrushin-implies-mixing} as follows: 
\begin{proof}[Proof of \cref{thm:dobrushin-implies-mixing}]
By \cref{lem:relating-contraction-to-influence}, when $\Phi$ satisfies the quantum Dobrushin condition, it is a contraction map in Wasserstein norm.
So, applying the Banach fixed point theorem, we get that $\Phi$ has a unique fixed point among trace 1 Hermitian matrices.
Call this fixed point $\sigma$.

Using that $\frac12 \trnorm{X} \leq \wnorm{X} \leq \frac\qubits2 \trnorm{X}$ (\cref{lem:w1-to-trace}), we have
\begin{equation}
    \begin{split}
        \trnorm{\Phi^{\tau}(\rho) - \sigma}
        \leq 2\cdot \wnorm{\Phi^{\tau}(\rho - \sigma) }
        & \leq 2 \cdot  \norm{D^{(\Phi)}}^\tau_{1\to 1} \wnorm{\rho - \sigma}  \\
        & \leq \qubits \cdot \Paren{1 - \frac{\gamma }{\qubits}}^\tau \trnorm{\rho - \sigma},
    \end{split}
\end{equation}
where the second inequality follows from \cref{lem:relating-contraction-to-influence}.
Setting $\tau \geq \frac{n}{\gamma} \log\Paren{ \frac{n}{\eps}}$ concludes the proof. 
\end{proof}

\subsection{Rapid mixing for continuous dynamics}

Next, we show that, above a constant temperature, the Gibbs state of a local Hamiltonian satisfies the quantum Dobrushin condition we introduce in \cref{def:quantum-dobrushin}.

\begin{lemma}[High-temperature Gibbs states satisfy Dobrushin]
\label{lemma:gibbs-sat-dobrushin}
    Let $H$ be a $(\degree, \locality)$-low-intersection Hamiltonian and let $\sigma$ be $H$'s Gibbs state at inverse temperature $\beta <\beta_c = 1/(10^5 \locality^3 \growth^2 \degree)$.
    Then, for sufficiently small $\delta > 0$, the map $\Phi = \calI + \frac{\delta}{n} \calL^*$, where $\calL^*$ is the Lindbladian defined in \cref{def:our-l}, satisfies the quantum Dobrushin condition with parameter $\gamma = \delta$.
\end{lemma}

With this lemma, we are now ready to prove rapid mixing for the continuous dynamics.

\rapid*

\begin{proof}[Proof of~\cref{thm:main-rapid}]
By \cref{lemma:gibbs-sat-dobrushin} and \cref{thm:dobrushin-implies-mixing}, for the input state $\rho$ and sufficiently small $\delta > 0$, for any $\tau \geq \frac{\qubits}{\delta} \log\frac{2\qubits}{\eps}$,
\begin{align*}
    \trnorm[\Big]{(\calI + \delta \calL^*)^{\tau}(\rho) - \sigma} \leq \frac{\eps}{2}\trnorm{\rho - \sigma} \leq \eps.
\end{align*}
We know that $\sigma$ is the fixed point of $\calI + \delta \calL^*$ (\cref{lem:stationarity}).
So, taking the limit $\delta \to 0$, we have that, for any $\tau \geq \qubits \log\frac{2\qubits}{\eps}$,
\begin{align*}
    \trnorm{e^{\calL^* \tau}(\rho) - \sigma} \leq \eps,
\end{align*}
as desired.
\ewindefer{This argument feels sketchy to me but whatever}
\end{proof}

It remains to prove that high-temperature Gibbs states satisfy the quantum Dobrushin condition. This amounts to relating the column sums of the influence matrix to the update matrices corresponding to the transport plans that we introduced in \cref{thm:wasserstein-growth-from-jump}.

\begin{proof}[Proof of \cref{lemma:gibbs-sat-dobrushin}]
We split our Lindbladian up into pieces in the following way:
\begin{align*}
    \Phi = \calI + \frac{\delta}{\qubits} \calL^* &=  \frac{1}{\qubits} \sum_i (\underbrace{\calI + \delta \calL_i^*}_{\Phi_i}) \\
    \Phi_i(\rho) &= \rho + \delta \sum_{P \in \braces{\sigma_X^{(i)},\,\sigma_Y^{(i)},\,\sigma_Z^{(i)}}} \parens[\Big]{- \ii [G^P , \rho] +  A^P \rho (A^P)^\dagger - \frac12\braces{(A^P)^\dagger A^P, \rho}}\,.
\end{align*}
We want to analyze the Dobrushin influence matrix of $\Phi$, which by \cref{fact:influence-in-w1-arb-x} is
\begin{align*}
    D_{i,j}^{(\Phi)} = \max_{\substack{X : \tr_j(X) = 0 \\ \wnorm{X} = 1}} \frac{1}{\qubits} \wnorm{\Phi_i(X)}.
\end{align*}
We will show that the maximum column sum, $\max_j \sum_{i=1}^\qubits D_{i,j}^{(\Phi)}$, is at most $1 - \frac{\delta}{\qubits}$.
This is equal to $\norm{D^{(\Phi)}}_{1 \to 1}$, so our bound implies the Dobrushin condition.

First, we bound the entries of $D^{(\Phi)}$.
Let $X$ be a matrix with $\tr_j(X) = 0$ and $\wnorm{X} = 1$.
Then it has a transport plan of $e_j$, so by \cref{thm:wasserstein-growth-from-jump}, $\Phi_i(X)$ has a transport plan of
\begin{align*}
    (\id + \delta Q^{(i)} + \bigOs{\delta^2}) e_j, \text{ where }
    Q^{(i)} &\coloneqq -4 E_{\braces{i}} + \wh{Q}^{(i)}
\end{align*}
To get a bound on $\wnorm{\Phi_i(X)}$, we take the $\ell_1$ norm of this transport plan.
Maximizing over all $X$ and using triangle inequality, we get
\begin{align*}
    \qubits D_{i,j}^{(\Phi)}
    &= \max_{\substack{X : \tr_j(X) = 0 \\ \wnorm{X} = 1}} \wnorm{\Phi_i(X)} \\
    &\leq \norm{(\id + \delta Q^{(i)} + \bigOs{\delta^2}) e_j}_1 \\
    &\leq \norm{(\id + \delta (-4 E_{\braces{i}} + \wh{Q}^{(i)}) + \bigOs{\delta^2}) e_j}_1 \\
    &\leq (1 - \iver{i = j}4\delta) + \delta \norm{\wh{Q}^{(i)} e_j}_1 + \bigOs{\delta^2}
\end{align*}
Now, we sum the entries over a column.
\begin{equation} \label{eq:cts-dobrushin-sum}
    \begin{split}
        \sum_{i = 1}^{\qubits} D_{i,j}^{(\Phi)}
        &\leq \frac{1}{\qubits} \sum_{i = 1}^{\qubits} \parens[\Big]{(1 - \iver{i = j}4\delta) + \delta \norm{\wh{Q}^{(i)} e_j}_1 + \bigOs{\delta^2}} \\
        &= 1 - \frac{4\delta}{\qubits} + \frac{\delta}{\qubits} \sum_{i = 1}^{\qubits} \norm{\wh{Q}^{(i)} e_j}_1 + \bigOs{\delta^2} \\
        &\leq 1 - \frac{3\delta}{\qubits} + \bigOs{\delta^2}
    \end{split}
\end{equation}
The final inequality follows from \cref{lem:Q-contracts}.
Therefore, for sufficiently small $\delta$, $\norm{D^{(\Phi)}}_{1\to 1} \leq 1 - \frac{\delta}{\qubits}$, which is the desired Dobrushin condition.
\end{proof}

\subsection{Rapid mixing for discrete dynamics}
\ewindefer{TODO clean this up}
Next, we show that our framework readily extends to proving rapid mixing for the discrete dynamics, where the $e^{\calL^* \tau}$ is applied for a constant time $\tau$.

\begin{definition}[Discrete Dynamics]
    \label{def:discrete-dynamics}
    For a Lindbladian $\calL$ which decomposes as $\calL = \sum_{i = 1}^\qubits \calL_i$, we define its \emph{discrete} dynamics with time steps of size $\Delta$ to be given by the channel
    \begin{align*}
        \Phi_{\calL,\Delta}(\rho) = \frac{1}{\qubits} \sum_{i=1}^\qubits e^{\calL_i \Delta}(\rho).
    \end{align*}
    The dynamics are performed by repeatedly applying $\Phi_{\calL, \Delta}$ to the input.
\end{definition}

\begin{theorem}[Discrete high-temperature dynamics satisfy quantum Dobrushin]
    For all $\Delta < 1/5$ and for $\beta < 1/(10^5 \locality^3 \growth^2 \degree)$, the channel $\Phi_{\calL^*, \Delta}$ satisfies the Dobrushin condition from \cref{def:quantum-dobrushin} with parameter $\Delta$.
\end{theorem}
\begin{proof}
We bound the entries of $D^{(\Phi)}$.
Let $X$ be a matrix with $\tr_j(X) = 0$ and $\wnorm{X} = 1$.
Then it has a transport plan of $e_j$, so by \cref{thm:wasserstein-growth-from-jump} and \cref{lem:update-matrix}, $\Phi_i(X)$ has a transport plan of $e^{\Delta Q^{(i)}}$, where $Q^{(i)} = -4 E_{\braces{i}} + \wh{Q}^{(i)}$.

To get a bound on $\wnorm{\Phi_i(X)}$, we take the $\ell_1$ norm of this transport plan.
Maximizing over all $X$ and using triangle inequality, we get
\begin{align*}
    \qubits D_{i,j}^{(\Phi)}
    = \max_{\substack{X : \tr_j(X) = 0 \\ \wnorm{X} = 1}} \wnorm{\Phi_i(X)}
    \leq \norm{e^{\Delta Q^{(i)}} e_j}_1
\end{align*}
Now, we sum the entries over a column.
\begin{equation} \label{eq:dsc-dobrushin-sum}
    \begin{split}
        \sum_{i = 1}^{\qubits} D_{i,j}^{(\Phi)}
        & \leq \frac{1}{\qubits} \sum_{i = 1}^\qubits \norm{e^{\Delta Q^{(i)}} e_j}_1 \\
        & \leq \frac{1}{\qubits} \sum_{i = 1}^\qubits \parens[\Big]{\norm{(I + \Delta Q^{(i)})e_j}_1 + \sum_{t \geq 2} \frac{\Delta^t}{t!} \norm{(Q^{(i)})^t e_j}} \\
        & \leq 1 - \frac{3\Delta}{\qubits} + \frac{1}{\qubits} \sum_{t \geq 2} \frac{\Delta^t}{t!} \sum_{i=1}^{\qubits} \norm{(Q^{(i)})^t e_j} \\
        & \leq 1 - \frac{3\Delta}{\qubits} + \frac{1}{\qubits} \sum_{t \geq 2} \frac{\Delta^t}{t!} \parens[\Big]{\max_{i \in [\qubits]} \norm{Q^{(i)}}_{1 \to 1}}^{t-1} \sum_{i=1}^{\qubits} \norm{Q^{(i)} e_j}  \\
        & \leq 1 - \frac{3\Delta}{\qubits} + \frac{1}{\qubits} \sum_{t \geq 2} \frac{\Delta^t}{t!} 5^{t-1} \cdot 2 \leq 1 - \frac{\Delta}{\qubits}\,.
    \end{split}
\end{equation}
Above, we use \cref{lem:Q-contracts} and \cref{cor:Q-1-1}.
\end{proof}

With this and \cref{thm:dobrushin-implies-mixing}, \cref{thm:main-discrete} follows as a corollary.

\discrete*

\subsection{Conditions for rapid mixing}

From an update matrix $Q$, we can directly prove rapid mixing without going through the Dobrushin influence matrix.
This is done in \cite{dgj09} to get tighter control on (and therefore a weaker condition for) rapid mixing.

\begin{lemma}[A rapid mixing condition from update matrices]
    Let $\Phi$ be a trace-preserving map with update matrix $Q$.
    Then, provided $\opnorm{Q} \leq 1 - \gamma$, for any states with density matrix $\rho$ and $\varrho$,
    \begin{align*}
        \wnorm{\Phi^{\tau}(\rho) - \Phi^{\tau}(\varrho)} \leq \eps,
    \end{align*}
    provided $\tau \geq \frac{1}{\gamma} \log\frac{\qubits}{\eps}$.
    In other words, $\Phi$ mixes to a stationary state in $\tau$ iterations.
\end{lemma}
\begin{proof}
Let $X = \rho - \varrho$.
Then there is a transport plan of $X$ whose cost vector $x$ is entrywise bounded between $0$ and $1$.
Such a transport plan is constructed in \cite[Proposition 2]{dmtl21}.

Then, by the definition of the update matrix, $\Phi^\tau(X)$ has a transport plan with cost vector bounded by $Q^\tau x$.
So,
\begin{multline*}
    \wnorm{\Phi^{\tau}(\rho) - \Phi^{\tau}(\varrho)}
    = \wnorm{\Phi^{\tau}(X)}
    \leq \norm{Q^\tau x}_1
    \leq \sqrt{\qubits}\norm{Q^\tau x}
    \leq \sqrt{\qubits}\opnorm{Q^\tau} \norm{x} \\
    \leq \sqrt{\qubits}(1 - \gamma)^\tau \norm{x}
    \leq \qubits(1 - \gamma)^\tau
    \leq \eps.
\end{multline*}
Above, we use norm conversion to convert the $\ell_1$ norm bounds to $\ell_2$ norm bounds; then, we finally use that the entries of $x$ are between $0$ and $1$ to bound $\norm{x} \leq \sqrt{\qubits}$.
\end{proof}

\begin{remark}[Comparing Wasserstein norm to oscillator norm] \label{rmk:osc}
    There is one other technique to prove rapid mixing in the literature on  quantum Lindbladians: showing that the adjoint, $\Phi^\dagger$, exhibits contraction in \emph{oscillator norm}~\cite{rfa24a}.
    Here, the oscillator norm of $A$ is defined to be
    \begin{align*}
        \norm{A}_{\text{osc}} = \sum_{i \in [\qubits]} \opnorm{A - \tfrac12 \id_i \otimes \tr_i(A)}.
    \end{align*}
    A simple argument shows that contraction in this norm implies mixing.
    Like our condition, this choice of norm is inspired by the classical Dobrushin condition.
    In particular, path coupling has a dual presentation more common in the statistical physics literature, which tracks a Lipschitz test function in observable space rather than a path in configuration space~\cite{ah87,dgj09}.
    Oscillator norm bears some resemblance to this dual formulation.
    Quantum Wasserstein norm's dual formulation is as follows \cite[Proposition 9]{dmtl21}:
    \begin{multline*}
        \wnorm{X} = \max_A \braces[\Big]{\tr(AX) \,\Big|\, \norm{A}_{W_1*} \leq 1}
        \text{ where } \norm{A}_{W_1*} = 2 \max_{i \in [\qubits]} \min_{A_i \in \C^{2^{\qubits - 1} \times 2^{\qubits - 1}}} \opnorm{A - \id \otimes A_i}.
    \end{multline*}
    Contraction in the state space, $\wnorm{\Phi(X)} \leq \alpha \wnorm{X}$ for all $X$, is equivalent to contraction in the observable space, $\norm{\Phi^\dagger(A)}_{W_1*} \leq \alpha \norm{A}_{W_1*}$ for all $A$.
    Moreover, the Wasserstein dual norm is related to oscillator norm~\cite[Proposition 15]{dmtl21}, since $\frac12 \tr_i(A)$ turns out to always be a near-optimal choice for $A_i$:
    \begin{align*}
        \norm{A}_{\text{osc}} \leq \norm{A}_{W_1*} \leq 2\qubits \norm{A}_{\text{osc}}.
    \end{align*}
    Altogether, this shows that contraction in Wasserstein norm and oscillator norm are related by a factor of $\qubits$ (where this factor comes from the oscillator norm being a sum over sites rather than the max).
    This clarifies the relationship between oscillator norm and the more conventional Wasserstein norm; however, it remains unclear whether one is more powerful than another, from the perspective of controlling mixing times.
\end{remark}

\section{Decay of conditional mutual information at high temperature}
\label{sec:cmi}

The goal of this section is to prove the following.

\cmi*

We begin by recalling that CMI can be bounded by the performance of a corresponding \emph{recovery map}:

\begin{fact}[{Bounding CMI through recovery maps \cite[Eq.\ (10)]{fr15} and \cite{bsw15}}]
    \label{fact:cmi-to-recovery-map}
    Let $\sigma$ be a mixed state on $\qubits$ qubits, and let $A, B, C$ be disjoint subsets of $[\qubits]$.
    Let $\rho$ be a ``recovered'' state, attained by discarding the $C$ register of $\sigma$, and applying a channel to the $B$ channel to create a new $C$ register.
    Then
    \begin{align*}
        I_{\sigma}(A : C \mid B) \leq 7 \abs{A} \trnorm{\sigma - \rho}^{1/2}.
    \end{align*}
\end{fact}

For a set $C \subseteq [\qubits]$ and a radius $\Delta \in \N$, our recovery map will be to evolve with respect to the Lindbladian $\calL^*_{ \ball(C, \Delta) } =  \sum_{i \in \ball(C, \Delta)} \calL_i^*$, where $\calL_i^*$ is defined in \cref{def:our-l}.
This map is not local, but quasi-local; so, we will evolve according to a truncated evolution, $\calL_{i, \delta}^*$.

\begin{claim}[Rapid mixing on a subsystem]
    \label{claim:cmi-mixing}
    For a subset $C \subseteq [\qubits]$, let $\rho$ be such that $\tr_C(\rho) = \tr_C(\sigma)$.
    Then, for any integer $\Delta \geq 1$ and any time $t > 0$,
    \begin{align*}
        \trnorm{e^{\calL^*_{ \ball(C, \Delta) }  t} (\rho) - \sigma} \leq 2 (e^{-t/2}\abs{C} + e^{2t} 10^{-\Delta}).
    \end{align*}
\end{claim}
\begin{proof}
Let $X = \rho - \sigma$.
By assumption, $\tr_C(X) = 0$, so by \cref{lem:w1-few-site} it has a transport plan whose cost vector $x$ satisfies
\begin{align*}
    x \vleq \trnorm{X} e_C \vleq 2 e_C\,,
\end{align*}
where $e_C$ is the indicator vector of the subsystem $C$. 
Consequently, by \cref{lem:stationarity},
\begin{equation*}
    e^{\calL^*_{ \ball(C, \Delta) }  t} (\rho) - \sigma = e^{\calL^*_{ \ball(C, \Delta) } t}(X)
\end{equation*}
has a transport plan with cost vector $y$, where
\begin{align*}
    y
    \vleq e^{t \sum_{i \in \ball(C, \Delta)} Q^{(i)}} x
    \vleq 2 e^{t \sum_{i \in \ball(C, \Delta)} Q^{(i)}} e_C
    = 2 e^{4t \sum_{i \in \ball(C, \Delta)} (-E_{\braces{i}} + \frac14\wh{Q}^{(i)}) } e_C
    \vleq 2 e^{t \sum_{i \in \ball(C, \Delta)} 4t R^{(i)} } e_C
\end{align*}
where $Q^{(i)}$ is as defined in \cref{thm:wasserstein-growth-from-jump} and $R^{(i)}$ is as defined in \cref{lem:cmi-contraction}.
Finally, by \cref{lem:cmi-contraction},
\begin{equation*}
    \norm{y}_1
    \leq 2 \norm{e^{t \sum_{i \in \ball(C, \Delta)} 4t R^{(i)} } e_C}_1
    \leq 2 (e^{-t/2}\abs{C} + e^{2t} 10^{-\Delta}). \qedhere
\end{equation*}
\end{proof}

\begin{claim}[Locality of the recovery map]
    \label{claim:cmi-locality}
    The map can be truncated to a radius of $\Delta + \bigO{\log(1/\delta)}$ around $C$ while incurring $\delta$ error.
    \begin{align*}
        \dnorm{e^{\sum_{i \in \ball(C, \Delta)} \calL_{i, \delta}^* t} - e^{\sum_{i \in \ball(C, \Delta)} \calL_i^* t}} \leq t \delta \abs{\ball(C, \Delta)} \leq t \delta (1 + \Delta)^\power \abs{C}.
    \end{align*}
\end{claim}
\begin{proof}
This follows from \cref{lem:truncation-error,lem:truncation-evolution-error}.
    \ewindefer{TODO}
\end{proof}

\begin{proof}[Proof of \cref{thm:cmi-main}]
We use \cref{fact:cmi-to-recovery-map}, so it suffices to find a recovery map.
Our choice of recovery map will be the following circuit:
\begin{enumerate}
    \item Given $\sigma_B$, adjoin $\abs{C}$ qubits initialized to the maximally mixed state;
    \item To $\sigma_B \otimes (\frac{1}{2^{\abs{C}}}\id)_C$, apply the Lindbladian
    \begin{align*}
        e^{\sum_{i \in \ball(C, \Delta)} \calL_{i, \delta}^* t}
    \end{align*}
    for $\Delta = \dist(A, C) / 4$, $\delta = e^{-\const{}\dist(A, C)}$ and $t = \dist(A, C) / 8$.
\end{enumerate}
Call the final state over $A$, $B$, and $C$ $\rho$.
We choose parameters such that $\Delta + \bigO{\log(1 / \delta)} < \dist(A, C)$, so this channel is not supported on $A$.
The distance between the recovered state and the original state can be bounded in the following way.
By \cref{claim:cmi-mixing} and \cref{claim:cmi-locality},
\begin{align*}
    \trnorm{\sigma - \rho}
    &\leq \trnorm{\sigma - e^{\sum_{i \in \ball(C, \Delta)} \calL_{i}^* t}(\sigma_{AB} \otimes (\tfrac{1}{2^{\abs{C}}}\id)_C)}
    + \trnorm{e^{\sum_{i \in \ball(C, \Delta)} \calL_{i}^* t}(\sigma_{AB} \otimes (\tfrac{1}{2^{\abs{C}}}\id)_C) - \rho} \\
    &\leq 2(e^{-t/2} \abs{C} + e^{2t} 10^{-\Delta}) + t \delta (1 + \Delta)^\power \abs{C} \\
    &\leq 2(e^{-t/2} \abs{C} + e^{2t} 10^{-\Delta}) + t \delta (1 + \Delta)^\power \abs{C} \\
    &= \bigO{\power^{\power}\abs{C} e^{-\const{}\dist(A, C)}}
\end{align*}
Plugging this into \cref{fact:cmi-to-recovery-map}, we have the desired bound:
\begin{align*}
    I_\sigma(A : C \mid B) = \bigO{\power^{\power}\abs{A}\abs{C} e^{-\const{}\dist(A, C)}}
\end{align*}
\end{proof}

\subsection{Key lemma: local contraction in Wasserstein norm}

\ewindefer{Here's what we actually need to prove.}
\begin{lemma}[Wasserstein contraction on a subset of the space] \label{lem:cmi-contraction}
    Let $R^{(i)}$ be the matrix
    \begin{align*}
        R^{(i)} \coloneqq - E_{\braces{i}} + \frac{1}{4} r^{(i)} (r^{(i)})^\dagger
        \text{ where } r_j^{(i)} \coloneqq c^{\dist(i, j)}.
    \end{align*}
    Where $c < 1/(\const{}e^{8\power})$.
    Then, for a subset $C \subseteq [\qubits]$
    \begin{align*}
        \norm[\big]{e^{t \sum_{i \in \ball(C, \Delta)} R^{(i)}} e_{C}}_1
        \leq e^{-t/2}\abs{C} + e^{2t} 10^{-\Delta}
    \end{align*}
    \ewindefer{Sufficiently small $1/c > 10^{10} e^{-8 \power}$; sloppy with constants here}
\end{lemma}

\begin{lemma}[Convergence of a sum over paths]
    \label{lem:cmi-paths-sum}
For a constant $c \in (0, 1)$, and for some initial point $i_0 \in [\qubits]$ and final point $i_\ell \in [\qubits]$, the sum
\begin{equation} \label{eq:cmi-distance-sum-converges}
\begin{split}
    \sum_{i_1, \dots, i_{\ell - 1} \in [\qubits]} c^{\dist(i_\ell, i_{\ell - 1}) + \dots + \dist(i_1, i_0)}
    \leq \const{}^{\ell - 1} \sum_{k \geq \dist(i_\ell, i_0)} (1 + k)^\power \sqrt{c}^k
\end{split}
\end{equation}
where $\const{} = (\sum_{x \geq 0} (1 + x)^\power \sqrt{c}^{x})^2$.
In particular, for $c < \frac{1}{1000} e^{-2 \power}$, we have that
\begin{equation} \label{eq:cmi-distance-sum-converges-tighter}
\begin{split}
    \sum_{i_1, \dots, i_{\ell - 1} \in [\qubits]} c^{\dist(i_\ell, i_{\ell - 1}) + \dots + \dist(i_1, i_0)}
    \leq 1.1^\ell (e^\power \sqrt{c})^{\dist(i_\ell, i_0)}
    \leq 1.1^\ell / 25^{\dist(i_\ell, i_0)}.
\end{split}
\end{equation}

\end{lemma}
\begin{proof}
We prove this by induction.
The base case is $\ell = 1$, in which case the above inequality is that $c^{\abs{i_1 - i_0}} \leq (4c)^{\abs{i_1 - i_0}}$, which is clearly true.
For the inductive case, we write
\begin{equation}
\begin{split}
    & \sum_{i_1, \dots, i_{\ell - 1} \in [\qubits]} c^{\dist(i_\ell, i_{\ell - 1}) + \dots + \dist(i_1, i_0)} \\
    &= \sum_{i_{\ell - 1} \in [\qubits]} c^{\dist(i_\ell, i_{\ell - 1})} \sum_{i_1, \dots, i_{\ell - 2} \in [\qubits]} c^{\dist(i_{\ell - 1}, i_{\ell - 2}) + \dist(i_{\ell - 2}, i_{\ell - 3}) + \dots + \dist(i_1, i_0)} \\
    &\leq \sum_{i_{\ell - 1} \in [\qubits]} c^{\dist(i_\ell, i_{\ell - 1})} \const{}^{\ell - 2} \sum_{k \geq \dist(i_{\ell - 1}, i_0)} (1 + k)^\power \sqrt{c}^k  \\
    &= \const{}^{\ell - 2} \sum_{\Delta_1 \geq 0} \sum_{\Delta_2 \geq 0} \abs[\Big]{\ball(i_{\ell}, \Delta_1) \cap \ball(i_0, \Delta_2)} c^{\Delta_1} \sum_{k \geq \Delta_2 } (1 + k)^\power \sqrt{c}^k \\
    &\leq \const{}^{\ell - 2} \sum_{\Delta_1 \geq 0} \sum_{\Delta_2 \geq 0} \iver{\Delta_1 + \Delta_2 \geq \dist(i_\ell, i_0)} (1 + \min(\Delta_1, \Delta_2))^\power c^{\Delta_1} \sum_{k \geq \Delta_2 } (1 + k)^\power \sqrt{c}^k \\
    &\leq \const{}^{\ell - 2} \sum_{k' \geq \dist(i_\ell, i_0)} \sum_{j=0}^{k'} (1 + \min(j, k' - j))^\power c^{k' - j} \sum_{k \geq j } (1 + k)^\power \sqrt{c}^k \\
    &= \const{}^{\ell - 2} \sum_{k' \geq \dist(i_\ell, i_0)} \sum_{j=0}^{k'} (1 + \min(j, k' - j))^\power c^{k' - j} (1 + j)^\power \sqrt{c}^j \sum_{x \geq 0} \parens[\Big]{\frac{1 + x + j}{1 + j}}^\power \sqrt{c}^{x} \\
    &\leq \const{}^{\ell - 1.5} \sum_{k' \geq \dist(i_\ell, i_0)} \sum_{j=0}^{k'} (1 + \min(j, k' - j))^\power c^{k' - j} (1 + j)^\power \sqrt{c}^j \\
    &\leq \const{}^{\ell - 1.5} \sum_{k' \geq \dist(i_\ell, i_0)} (1 + k')^\power \sqrt{c}^{k'} \sum_{j=0}^{k'} (1 + k' - j)^\power \sqrt{c}^{k' - j} \\
    &= \const{}^{\ell - 1.5} \sum_{k' \geq \dist(i_\ell, i_0)} (1 + k')^\power \sqrt{c}^{k'} \sum_{j'=0}^{k'} (1 + j')^\power \sqrt{c}^{j'} \\
    &\leq \const{}^{\ell - 1} \sum_{k' \geq \dist(i_\ell, i_0)} (1 + k')^\power \sqrt{c}^{k'}
\end{split}
\end{equation}
For the series to converge, we use that $c < 1$.
When $c < \frac{1}{1000}e^{2\power}$, we can continue bounding, using that $1 + x \leq e^x$, so that
\begin{align*}
    \const{} = \parens[\Big]{\sum_{x \geq 0} (1 + x)^\power \sqrt{c}^x }^2
    \leq \parens[\Big]{\sum_{x \geq 0} (e^\power\sqrt{c})^x }^2
    \leq \parens[\Big]{\sum_{x \geq 0} 1000^{-x/2} }^2
    \leq 1.1.
\end{align*}
Consequently,
\begin{equation*}
    \sum_{i_1, \dots, i_{\ell - 1} \in [\qubits]} c^{\dist(i_\ell, i_{\ell - 1}) + \dots + \dist(i_1, i_0)}
    \leq 1.1^{\ell - 1} \sum_{k \geq \dist(i_\ell, i_0)} (e^\power \sqrt{c})^k
    \leq 1.1^{\ell} (e^\power \sqrt{c})^{\dist(i_\ell, i_0)}.
\end{equation*}
\end{proof}

\begin{proof}[Proof of \cref{lem:cmi-contraction}]
Let $\wh{R}^{(i)} = r^{(i)} (r^{(i)})^\dagger$.
Then we can write
\begin{align}
    e^{t \sum_{i \in \ball(C, \Delta)} R^{(i)}} e_{C}
    &= e^{\frac{t}{4} \sum_{i \in \ball(C, \Delta)} \wh{R}^{(i)} - t\sum_{i \in \ball(C, \Delta)} E_{\braces{i}} } e_{C} \nonumber \\
    &= e^{\frac{t}{4} \sum_{i \in \ball(C, \Delta)} \wh{R}^{(i)} + t\sum_{i \not\in \ball(C, \Delta)} E_{\braces{i}} } e^{-t} e_{C} \nonumber \\
    &=  e^{-t} e^{\frac{t}{4} \sum_{i \in \ball(C, \Delta)} \wh{R}^{(i)} + t\wh{R}^{(0)} } e_{C} \nonumber \\
\intertext{
    The last line follows by pulling out a factor of $t \id = t \sum_{i=1}^\qubits E_{\braces{i}}$, where we denote $\wh{R}^{0} \coloneqq \sum_{i \not\in \ball(C, \Delta)} E_{\braces{i}}$.
    The $e^{-t}$ is a scalar we can factor out; then, we expand out the exponential, to get
}
    &= e^{-t}\sum_{k \geq 0} \frac{t^k}{k!} \parens[\Big]{\frac{1}{4} \sum_{i \in \ball(C, \Delta)} \wh{R}^{(i)} + \wh{R}^{(0)} }^k e_{C} \nonumber
\intertext{
    For a sequence $i = (i_1,\dots,i_k) \in (\ball(C,\Delta) \cup \braces{0})^k$, let $z(i)$ be the set of indices $j$ for which $i_j = 0$.
    Then we can rewrite the summation as follows.
}
    &= e^{-t} \sum_{k \geq 0} \frac{t^k}{k!} \sum_{T \subseteq [k]} \frac{1}{4^{k - \abs{T}}} \underbrace{\sum_{\substack{i \in (\ball(C, \Delta) \cup \braces{0})^k \\ z(i) = T}} \wh{R}^{(i_k)} \dots \wh{R}^{(i_1)} e_C}_{\coloneqq s_T}.
    \label{eq:cmi-main-bound}
\end{align}
We now consider these sums $s_T$ individually.
Notice that the summation $s_T$ splits into a product over summations.
Let $\wh{R}[0] \coloneqq \wh{R}^{(0)}$ and $\wh{R}[1] \coloneqq \sum_{i \in \ball(C, \Delta)} \wh{R}^{(i)}$.
Then
\begin{align*}
    s_T &= \wh{R}[\iver{k \not\in T}] \dots \wh{R}[\iver{1 \not\in T}] e_C.
\end{align*}
In order to bound these sums, we make three observations.
First, we can bound sums of entries of $r^{(i)}$:
\begin{align*}
    \sum_{j \in S} r_j^{(i)}
    \leq \sum_{j : \dist(j, i) \geq \dist(S, i)} r_j^{(i)}
    \leq \sum_{k \geq \dist(S, i)} (1 + k)^\power c^k
    \leq \sum_{k \geq \dist(S, i)} (e^\power c)^k
    \leq 1.1 (e^\power c)^{\dist(S, i)}.
\end{align*}
Next, by \cref{lem:cmi-paths-sum},
\begin{equation} \label{eq:cmi-inner-products}
\begin{split}
    (r^{(i)})^\dagger r^{(j)}
    = \sum_{k=1}^\qubits c^{\dist(i, k) + \dist(k,j)}
    \leq 1.1 (e^\power \sqrt{c})^{\dist(i, j)}.
\end{split}
\end{equation}
By a similar argument,
\begin{equation} \label{eq:cmi-R-1-1}
\begin{split}
    \norm[\Big]{\sum_{i=1}^\qubits \wh{R}^{(i)}}_{1 \to 1}
    &\leq \max_{j \in [\qubits]} \sum_{i=1}^\qubits \norm{\wh{R}^{(i)} e_j}_1 \\
    &= \max_{j \in [\qubits]} \sum_{i=1}^\qubits (r^{(i)})^\dagger e_j \norm{r^{(i)}}_1 \\
    &\leq \max_{j \in [\qubits]} \sum_{i=1}^\qubits \parens[\Big]{c^{\dist(i, j)} \sum_{j' = 1}^{\qubits} c^{\dist(i, j')}} \\
    &\leq 1.1.
\end{split}
\end{equation}
First, consider the case that $T = \varnothing$.
Then, by \eqref{eq:cmi-R-1-1},
\begin{align*}
    \norm[\big]{s_T}_1 &= \norm[\big]{\wh{R}[1]^k e_C}_1
    \leq \norm[\big]{\wh{R}[1]}_{1 \to 1}^k \abs{C}
    \leq 1.1^k \abs{C}.
\end{align*}
Now, suppose $T \neq \varnothing$.
Let $j^*$ be the smallest $j$ for which $j \in T$.
If $j^* = 1$, then $s_T = 0$; this is because $\wh{R}^{(0)} e_C = \sum_{i \not\in \ball(C, \Delta)} E_{\braces{i}} e_C = 0$.
Otherwise,
\begin{align*}
    \norm{s_T}_1
    &= \norm[\big]{\wh{R}[\iver{k \not\in T}] \dots \wh{R}[\iver{j^* + 1 \not\in T}] \wh{R}[0] \wh{R}[1]^{j^* - 1} e_C}_1 \\
    &\leq \norm[\big]{\wh{R}[1]}_{1 \to 1}^{k - j^* - \abs{T}+1} \norm[\big]{\wh{R}[0]}_{1 \to 1}^{\abs{T} - 1} \norm[\big]{\wh{R}[0] \wh{R}[1]^{j^* - 1} e_C}_1 \\
    &\leq 1.1^{k - j^* - \abs{T} + 1} \norm[\big]{\wh{R}[1] \wh{R}[0]^{j^* - 1} e_1}_1
\end{align*}
What remains is to compute the final norm expression:
\begin{align*}
    \wh{R}[0] \wh{R}[1]^{j^* - 1} e_C
    &= \sum_{i_1,\dots,i_{j^* - 1} \in \ball(C, \Delta)} \wh{R}[0] \wh{R}^{(i_{j^* - 1})} \dots \wh{R}^{(i_1)}  e_C \\
    &= \sum_{i_1,\dots,i_{j^* - 1} \in \ball(C, \Delta)} (\wh{R}[0] r^{(i_{j^* - 1})}) \dots \parens[\big]{(r^{(i_{3})})^\dagger r^{(i_{2})}} \parens[\big]{(r^{(i_{2})})^\dagger r^{(i_{1})}} \parens[\big]{(r^{(i_{1})})^\dagger e_C}.
\end{align*}
Consequently,
\begin{align*}
    \norm{\wh{R}[0] \wh{R}[1]^{j^* - 1} e_C}_1
    &\leq \sum_{i_1,\dots,i_{j^* - 1} \in \ball(C, \Delta)} \norm{\wh{R}[0] r^{(i_{j^* - 1})}}_1 \parens[\Big]{\prod_{j=1}^{j^* - 2} \parens[\big]{(r^{(i_{j+1})})^\dagger r^{(i_{j})}}} (r^{(i_{1})})^\dagger e_C  \\
    &\leq \sum_{i_1,\dots,i_{j^* - 1} \in \ball(C, \Delta)} 1.1 (e^\power c)^{-\dist([\qubits] \setminus \ball(C, \Delta), i_{j^* - 1})} \parens[\Big]{\prod_{j=1}^{j^* - 2} 1.1 (e^\power \sqrt{c})^{\dist(i_{j+1}, i_j)}} 1.1 (e^\power c)^{-\dist(i_{1}, C)} \\
    &\leq 1.1^{j^*} \sum_{i_1,\dots,i_{j^* - 1} \in \ball(C, \Delta)} (e^\power \sqrt{c})^{\dist([\qubits] \setminus \ball(C, \Delta), i_{j^* - 1}) + \dots + \dist(i_3, i_2) + \dist(i_{2}, i_1) + \dist(i_1, C) } \\
    &\leq 1.25^{j^*} / 10^{\Delta + 1} \\
    &\leq 1.25^{k} / 10^{\Delta + 1}
\end{align*}
\ewindefer{using \cref{fact:cmi-paths-sum}}
In summary, we have shown that $s_T \leq 1.25^k / 10^\Delta$ for non-empty $T$, while $s_{\varnothing} \leq 1.1^k \abs{C}$.
We can then plug this back into the derivation in \eqref{eq:cmi-main-bound} to get
\begin{align*}
    \norm[\big]{e^{t \sum_{i=1}^\Delta R^{(i)}} e_{\braces{1}}}_1
    &\leq e^{-t} \sum_{k \geq 0} \frac{t^k}{k!} \sum_{T \subseteq [k]} \frac{1}{4^{k - \abs{T}}} \norm{s_T}_1 \\
    &\leq e^{-t} \sum_{k \geq 0} \frac{t^k}{k!} \parens[\Big]{\frac{1.1^k}{4^k}\abs{C} + \sum_{\substack{T \subseteq [k] \\ T \neq \varnothing}} \frac{1}{4^{k - \abs{T}}} 1.25^k/10^{\Delta}} \\
    &\leq e^{-t} \sum_{k \geq 0} \frac{t^k}{k!} \parens[\Big]{\frac{1.1^k}{4^k}\abs{C} + 1.25^{k}10^{-\Delta}(1 + 1/4)^k} \\
    &\leq e^{-t}(e^{1.1t/4}\abs{C} + e^{2t} 10^{-\Delta}) \\
    &\leq e^{-t/2}\abs{C} + e^{2t} 10^{-\Delta}
\end{align*}
\end{proof}


\section*{Acknowledgments}
\addcontentsline{toc}{section}{Acknowledgments}

The authors thank Álvaro Alhambra, Chi-Fang Chen, Kuikui Liu, Cambyse Rouzé, and  Nikhil Srivastava for enlightening discussions. Part of this work was done while A.B. was visiting the Simon's Institute.
E.T.\ is supported by the Miller Institute for Basic Research in Science, University of California Berkeley.

\printbibliography

\appendix
\section{Detailed balance} \label{sec:detail}

In this section, we describe the quantum generalization of detailed balance, and prove that our Lindbladian (\cref{def:our-l}) satisfies it.
We follow the exposition of Lin~\cite{Lin24}, which in turn recalls the results of~\cite{wocjan2023szegedy} and~\cite{carlen2020non}.

Recall from the definition (\cref{def:lindbladian}) that a Lindbladian takes the form
\begin{equation*}
    \calL(\rho) = - \ii [G, \rho] + \sum_k \parens[\big]{K_k \rho K_k^\dagger - \frac{1}{2} \braces{K_k^\dagger K_k, \rho} }\,.
\end{equation*}
We will also consider its adjoint $\calL^\dagger$, the map for which $\tr(A^\dagger \calL(B)) = \tr(\calL^\dagger(A)^\dagger B)$ for all $A, B$:
\begin{align*}
    \calL^\dagger(\rho) &= \ii [G, \rho] + \sum_k \parens[\big]{K_k^\dagger \rho K_k - \frac12\braces{K_k^\dagger K_k, \rho}}.
\end{align*}
Note that $\calL$ and its adjoint commute with the conjugate transpose map: $\calL(A^\dagger)^\dagger = \calL(A)$ and $\calL^\dagger(A^\dagger)^\dagger = \calL^\dagger(A)$.

\begin{definition}[Quantum detailed balance conditions]
\label{def:detailed-balance}
    For a state with invertible density matrix $\sigma$, we say that a Lindbladian $\calL$ (\cref{def:lindbladian}) satisfies \emph{Gelfand--Naimark--Segal (GNS) detailed balance} if, for some $s \in [0,1] \setminus \braces{\frac12}$, $\calL$ satisfies the following condition:
    \begin{align} \label{eq:db}
        \text{for all $B$, }
        \calL(\sigma^{1-s}B\sigma^s) &= \sigma^{1-s} \calL^\dagger(B) \sigma^s.
    \end{align}
    (One can show that if this holds for some such $s$, it holds for all $s \in [0,1] \setminus \braces{\frac12}$.)
    We say that $\calL$ satisfies \emph{Kubo--Martin--Schwinger (KMS) detailed balance} if $\calL$ satisfies the condition for $s = \frac12$.
\end{definition}

These notions of detailed balance reduce to the classical case when $\calL$ is classical.
In the same way that classical detailed balance can be thought of as a self-adjointness condition, the above conditions translate to $\calL$ being self-adjoint with respect to a particular $\sigma$-dependent inner product.

\begin{fact}[$\sigma$-detailed balance implies stationarity at $\sigma$]
    If a Lindbladian $\calL$ satisfies either GNS or KMS detailed balance, then $e^{\calL t}(\sigma) = \sigma$.
    This follows from taking $B = \id$ in \eqref{eq:db}:
    \begin{align*}
        \calL(\sigma) &= \sigma^{1-s} \calL^\dagger(\id) \sigma^s = 0,
    \end{align*}
    using that $\calL^\dagger(\id) = 0$ for any Lindbladian.
\end{fact}

GNS detailed balance is a very strong condition, enforcing the correct transitions on essentially an eigenvector-by-eigenvector basis.
Lindbladians which satisfy this condition therefore have a highly restricted form which can be characterized formally, and further are believed to be intractable to implement for large system sizes (both on a quantum computer and in nature).
KMS detailed balance is strictly weaker than GNS detailed balance; this is the condition that our Lindbladian satisfies.

\begin{lemma}[KMS detailed balance]
\label{lem:kms}
The Lindbladians $\calL^*$ and $\calL^P$ introduced in \cref{def:our-l} satisfies KMS detailed balance (\cref{def:detailed-balance}), i.e.\ for any Hermitian $P$, for all $B$,
\begin{align*}
    \calL^P(\sigma^{1/2}B\sigma^{1/2}) &= \sigma^{1/2} {\calL^P}^\dagger(B) \sigma^{1/2}.
\end{align*}
\end{lemma}
\begin{proof}
The KMS detailed balance condition is closed under taking linear combinations, so since $\calL^*$ is a sum over $\calL^P$ for $P$ Hermitian, it is enough to prove KMS detailed balance for $\calL^P$.

Now, consider $\calL^P$ for some $P$.
For this proof, we will drop the superscripts and refer to $A^P$ and $G^P$ as $A$ and $G$, respectively.
We expand out the left-hand side, using that $A = \sigma^{1/4} P \sigma^{-1/4}$, so that $\sigma^{-1/2}A\sigma^{1/2} = A^\dagger$.
\begin{align*}
    \calL^P(\sigma^{1/2}B\sigma^{1/2}) 
    &= - \ii [G, \sigma^{1/2}B\sigma^{1/2}]
        + A \sigma^{1/2}B\sigma^{1/2} A^\dagger - \frac12\braces{A^\dagger A, \sigma^{1/2}B\sigma^{1/2}} \\
    &= - \ii [G, \sigma^{1/2}B\sigma^{1/2}]
        + \sigma^{1/2} A^\dagger B A \sigma^{1/2} - \frac12\braces{A^\dagger A, \sigma^{1/2}B\sigma^{1/2}}
\end{align*}
On the other hand, the right-hand side expands out to
\begin{align*}
    \sigma^{1/2} {\calL^P}^\dagger(B) \sigma^{1/2}
    &= \ii \sigma^{1/2} [G, B] \sigma^{1/2} + \sigma^{1/2} A^\dagger B A \sigma^{1/2} - \frac12 \sigma^{1/2} \braces{A^\dagger A, B} \sigma^{1/2} 
\end{align*}
Subtracting the two, we get that
\begin{align*}
    & \calL^P(\sigma^{1/2}B\sigma^{1/2}) - \sigma^{1/2} {\calL^P}^\dagger(B) \sigma^{1/2}\\
    &= - \ii [G, \sigma^{1/2}B\sigma^{1/2}]
        - \frac12\braces{A^\dagger A, \sigma^{1/2}B\sigma^{1/2}}
        - \ii \sigma^{1/2} [G, B] \sigma^{1/2}
        + \frac12 \sigma^{1/2} \braces{A^\dagger A, B} \sigma^{1/2} \\
    &= \parens[\big]{-\ii G\sigma^{1/2} - \ii \sigma^{1/2} G - \frac12 A^\dagger A \sigma^{1/2} + \frac12 \sigma^{1/2} A^\dagger A}B\sigma^{1/2} \\
    & \hspace{2em} + \sigma^{1/2}B\parens[\big]{\ii \sigma^{1/2} G + \ii G \sigma^{1/2} - \frac12 \sigma^{1/2} A^\dagger A + \frac12 A^\dagger A \sigma^{1/2}} \\
    &= \parens[\big]{-\ii \braces{G, \sigma^{1/2}} + \frac12 [\sigma^{1/2}, A^\dagger A]}B\sigma^{1/2} + \sigma^{1/2}B\parens[\big]{-\ii \braces{G, \sigma^{1/2}} + \frac12 [\sigma^{1/2}, A^\dagger A]}^\dagger
\end{align*}
Finally, to show that this is zero, we show that $-\ii \braces{G, \sigma^{1/2}} + \frac12 [\sigma^{1/2}, A^\dagger A] = 0$.
Using \cref{def:our-l}, we have that $G = \frac{-\ii}{2} A^\dagger A \circ \braces[\Big]{\frac{\sigma_i^{1/2} - \sigma_j^{1/2}}{\sigma_i^{1/2} + \sigma_j^{1/2}}}_{ij}$.
So:
\begin{align*}
    -\ii \braces{G, \sigma^{1/2}} + \frac12 [\sigma^{1/2}, A^\dagger A]
    &= A^\dagger A \circ \braces[\Big]{-\ii\frac{-\ii(\sigma_i^{1/2} - \sigma_j^{1/2})}{2(\sigma_i^{1/2} + \sigma_j^{1/2})}(\sigma_i^{1/2} + \sigma_j^{1/2}) + \frac12(\sigma_i^{1/2} - \sigma_j^{1/2})}_{ij} \\
    &= A^\dagger A \circ \braces[\Big]{-\frac{1}{2}(\sigma_i^{1/2} - \sigma_j^{1/2}) + \frac12(\sigma_i^{1/2} - \sigma_j^{1/2})}_{ij}
    = 0
\end{align*}
We are done.
\end{proof}

\end{document}